\documentclass{article}

\usepackage[a4paper,top=1in,bottom=1in,left=1.25in,right=1.25in]{geometry}

\usepackage[T1]{fontenc}
\usepackage{newtxtext,newtxmath}

\usepackage{mathtools} 
\usepackage{bm}

\usepackage{amsthm} 

\usepackage{graphicx}
\usepackage{booktabs,threeparttable,multirow,tabularx}

\usepackage{setspace}
\usepackage[round,authoryear]{natbib}

\usepackage{xurl}
\usepackage[colorlinks=true, linkcolor=black, citecolor=black, urlcolor=blue]{hyperref}

\newtheorem{lemma}{Lemma}
\newtheorem{definition}{Definition}[section]
\newtheorem{proposition}[definition]{Proposition}

\title{Corporate Bond Yield Curve Modeling: A Rating-Based Regime-Switching Generalized CIR Approach}

\author{%
Maochun Xu\thanks{Department of Applied Mathematics, Xi’an Jiaotong--Liverpool University, Suzhou, China.}%
\and
Yunqi Liang \thanks{Department of Engineering Science, University of Oxford, Oxford, OX1 3PJ, United Kingdom Email: \texttt{Yunqi.Liang@reuben.ox.ac.uk}}
\and
Yi Hong\thanks{Corresponding author. Department of Applied Mathematics, Xi’an Jiaotong--Liverpool University, Suzhou, China. Email: \texttt{Yi.Hong@xjtlu.edu.cn}.}%
}

\begin{document}

\maketitle

\begin{center}
Working Paper\\
This version: March 2026
\end{center}

\begin{abstract} 
Persistent shifts in term-structure dynamics undermine the stability of single-regime models in long samples. We develop an arbitrage-free regime-switching generalized CIR (RS–GCIR) model that jointly prices the Chinese government bond (CGB) curve and corporate bond curves. To capture the systematic transmission from interest-rate conditions to credit spreads, we structure the model into two blocks and price corporate bonds conditional on the prevailing rate regime. The rate block features a two-state RS–GCIR short-rate process estimated from CGB zero-coupon curves, while the credit block embeds CIR-type credit factors in an intensity-based framework for rating migration and default. We implement a block-recursive Unscented Kalman Filter (UKF) procedure—filtering the rate block first and the credit block next—using weekly data from 2014--2025, a period that begins with the onset of China’s modern corporate default cycle. We identify two persistent rate regimes with distinct level--volatility profiles. Relative to single-regime benchmarks, regime switching improves joint curve fit, delivers economically interpretable filtered regime probabilities, and sharpens the decomposition of corporate yields into discounting and credit compensation.
\end{abstract}




\newpage

\section{Introduction}
\label{sec:intro}

Since the landmark onshore bond default of Shanghai Chaori Solar in March 2014, China’s credit bond market has evolved through pronounced and persistent shifts in the level and risk characteristics of the sovereign term structure \citep{lin2017bonded,GengPan2023SOE}. Over the subsequent decade, the risk-free discount environment has experienced a material downshift—from the relatively elevated yield levels prevailing in the mid-2010s to a markedly lower-rate regime in more recent years—accompanied by changes in volatility and policy-driven risk sentiment. Against this backdrop, investment-grade corporate curves have co-moved tightly with sovereign phases, implying that observed credit spreads are strongly state dependent rather than well summarized by a single stationary premium.

This evidence matters for valuation. In reduced-form models, defaultable bond prices are pinned down by two primitives—the discount curve and the compensation for credit events (default and rating migration)—both of which are allowed to depend on the macro–financial state under the pricing measure. The intensity-based framework with stochastic rating transitions \citep{DuffieSingleton1999, feldhutter2008decomposing} formalizes how credit compensation varies with the economic environment, and therefore can inherit the state dependence embedded in the sovereign term structure. Related term-structure decompositions further emphasize that movements in corporate yields reflect joint variation in discounting and credit components, motivating specifications that let the credit block load on the prevailing rate environment rather than forcing a single set of dynamics to rationalize heterogeneous episodes. Consequently, a single-regime affine multi-factor model can become statistically brittle in long China samples: it pools structurally different rate–credit episodes, blurs the economic interpretation of risk premia, and mechanically worsens pricing performance when both discounting conditions and credit compensation shift across regimes.

Against this backdrop, we develop an arbitrage-free, rating-based reduced-form framework to price China’s onshore investment-grade corporate bonds with explicit default and migration risk. The model is organized into a rate block and a credit block.
In the rate block, we jointly model the China government bond (CGB) and China Development Bank (CDB) term structures to recover persistent rate regimes and to deliver an operational discount curve for corporate valuation. Consistent with market practice, we take the CDB curve as the discounting reference and measure corporate spreads relative to CDB, reflecting the central role of policy-bank benchmarks in pricing and hedging \citep{ZhangWu2019CreditBonds}. At the same time, CGB and CDB are not interchangeable: their yield spread is systematic and time-varying and is more naturally interpreted as a non-credit services component than as CDB default risk. We therefore model the CGB--CDB spread as a latent convenience-yield factor in the spirit of \citet{KrishnamurthyVissingJorgensen2012,ChenChenHeLiuXie2023}, which prevents sovereign convenience premia from being mechanically loaded into measured corporate credit spreads. In the credit block, corporate bonds are priced in a Lando-style migration-default framework that maps intensities and rating-transition generators into defaultable bond prices and yield decompositions \citep{DuffieSingleton1999,lando1998cox,feldhutter2008decomposing}, yielding a transparent separation of discounting (CDB-based) and credit compensation.

We estimate the model in a nonlinear regime-switching state-space form by
maximum likelihood. Likelihood evaluation uses a regime-switching unscented
Kalman filter (UKF), which is designed for nonlinear transition and measurement
maps \citep{JulierUhlmann1997,WanVanDerMerwe2000}. Estimation follows a
block-recursive strategy \citep{FeldhutterNielsen2012}. We first filter the rate factors and the rate-regime
probabilities from the CGB and CDB curves. We then filter credit states and the
systemic credit regime conditional on the rate-block posterior summaries.
A central implementation principle is price-level mixing. Defaultable
bond pricing is exponential in the state.

Empirically, the regime-switching model delivers a materially tighter fit for investment-grade corporate curves across ratings and maturities. The implied decomposition is economically cleaner: a sizable share of spread variation is attributed to shifts in discount-rate conditions and the CGB--CDB spread, while the remainder reflects state-dependent credit compensation through default and migration risk. Consistent with this credit-market improvement, the rate block also benefits from regime switching: the CGB and CDB term structures are well summarized by the rate regime with distinct mean and volatility patterns, and the filtered regime probabilities line up with major macro and funding episodes. The estimated spread further isolates sovereign non-credit services from the discount component.


This paper makes three contributions. First, we propose an arbitrage-free regime-switching term-structure model for China that links a two-regime GCIR rate block to a rating-based reduced-form credit block. The framework is designed to price, within a single coherent system, the zero-coupon term structures of Chinese government bonds (CGBs), policy bank bonds (CDBs), and onshore investment-grade corporate bonds, thereby allowing discount-rate regimes to propagate into credit pricing through a disciplined no-arbitrage mapping. Second, we develop a tractable estimation strategy that is consistent with the convex structure of asset pricing under regime uncertainty. Specifically, we implement a block-recursive maximum-likelihood procedure using a regime-switching unscented Kalman filter, and we perform regime aggregation at the price level (rather than mixing states or yields) to preserve the convexity implied by no-arbitrage valuation and to avoid Jensen-type distortions when moving between state variables and prices. Third, we construct a China-specific rating-migration and default-intensity generator for risk-neutral corporate bond pricing. Because China lacks an established market-standard risk-neutralization transition/intensity generator comparable to those commonly used in developed markets, we provide a transparent and implementable pipeline—from China-specific empirical rating migration inputs to a continuous-time generator and its risk-neutral adjustment—so that the resulting matrix can be directly embedded in reduced-form pricing and decompositions of corporate yields into discounting and credit components.

Our paper bridges three strands of research and addresses a gap at their intersection by jointly modeling persistent regime shifts in the rate and credit environments, including their interaction in valuation, while keeping the credit block identified and economically interpretable. The first study latent regime shifts in interest-rate dynamics, beginning with \citet{Hamilton1989}. It extends to no-arbitrage regime-switching affine term-structure models that deliver state-contingent bond pricing while preserving tractable valuation \citep{van2020regime,CelaryKruhnerEksiAltay2024}. Recent contributions further improve tractability by characterizing Markov-modulated affine processes and related specifications. These include variants with richer state dependence, semi-Markov modulation, or additional sources of risk, yet retain exponential-affine pricing formulas \citep{RodrigoMamon2021,KurtFrey2022,SiuElliott2025,MolibeliVanVuuren2025}. The second strand develops reduced-form credit pricing with rating migration and default in an exponential-affine setting. It provides a disciplined mapping from intensities and rating-transition generators to defaultable bond prices and yield decompositions \citep{DuffieSingleton1999,lando1998cox,feldhutter2008decomposing}. The third strand shows that corporate credit spreads contain a large component tied to discount-rate conditions and time-varying risk premia, which can dominate purely expected-loss movements in long samples \citep{GilchristZakrajsek2012}.

The rest of the paper is organized as follows. Section~\ref{sec:model} presents
the model and pricing implications. Section~\ref{sec:data} describes the data.
Section~\ref{sec:estimation} details the state-space representation and the
block-recursive regime-switching UKF likelihood. Section~\ref{sec:empirical}
reports the empirical results. Section~\ref{sec:conclusion} concludes.

\section{Model}
\label{sec:model}

We develop a joint pricing framework for China’s bond market that embeds the benchmark rate term structure---constructed from Chinese government bonds (CGBs) and policy-bank bonds such as China Development Bank (CDB) bonds---together with rating-sorted corporate bond yields in a unified state-space setting. The key insight is that the interest-rate environment and the credit cycle can co-move in valuation and filtered beliefs, yet their regime changes need not be synchronized. Accordingly, we introduce two latent regime processes governing the rate and credit sides, respectively, within a regime-switching affine (GCIR-type) term-structure model, in the spirit of the Markov-switching tradition and regime-shift term-structure literature \citep{Hamilton1989,BansalZhou2002,DaiSingletonYang2006}.

We interpret the rate regime $s_t^{r}$ as summarizing monetary and interest-rate conditions (high-rate versus low-rate), and the credit regime $s_t^{c}$ as capturing shifts in economy-wide risk appetite and financing constraints (credit expansion versus contraction), consistent with evidence that credit spreads and financial conditions are strongly cyclical \citep{BernankeGertlerGilchrist1999,GilchristZakrajsek2012}. We adopt a block-learning structure: the rate regime and rate factors are identified mainly from the CGB/CDB curves, while the credit regime is identified from the cross-section and term structure of rating-based corporate yields conditional on the extracted rate block. This separation leads to a two-step estimation approach and highlights inference issues associated with generated regressors \citep{Pagan1984,murphy2002estimation}. For corporate bonds, we follow a rating-based reduced-form treatment of credit risk \citep{lando1998cox,feldhutter2008decomposing} and allow a regime-dependent transmission channel through which rate factors load into corporate default intensities.

\subsection{State dynamics with regime switching}
\label{subsec:theory_block_learning}

Let the latent state vector be
\begin{equation}
X_t=\bigl(X_{1,t},X_{2,t},X_{3,t},X_{4,t}\bigr)^\top\in\mathbb{R}^4 .
\end{equation}
The first three factors $\{X_{1,t},X_{2,t},X_{3,t}\}$ constitute a rate block identified from the CGB curve together with policy-bank (CDB) yields/spreads, while $X_{4,t}$ forms a credit block. In particular, $X_{3,t}$ is a CDB--CGB convenience-yield factor that captures a systematic spread between CDB and CGB bond valuations driven by non-credit components such as liquidity and collateral value \citep{KrishnamurthyVissingJorgensen2012,hu2022review}.

We introduce two latent regime components: a benchmark-rate regime $s_t^{r}$ and a systemic credit regime $s_t^{c}$,
\begin{align}
&s_t^{r}\in\{L,H\}, &&L:\ \text{low-rate environment},\quad H:\ \text{high-rate environment}, \nonumber\\
&s_t^{c}\in\{E,C\}, &&E:\ \text{credit expansion},\quad C:\ \text{credit contraction}.
\label{eq:regime_labels}
\end{align}
For notational convenience in state-contingent valuation, we define the composite regime index
\begin{equation}
S_t \equiv (s_t^{r},s_t^{c})\in\mathcal S \equiv \{L,H\}\times\{E,C\}.
\label{eq:joint_state_def}
\end{equation}
Crucially, \eqref{eq:joint_state_def} is introduced only as a bookkeeping device: the transition mechanism of the rate and credit regimes is specified separately, and the inference of the rate regime does not use credit-block information.

Both $s_t^{r}$ and $s_t^{c}$ are modeled as time-homogeneous continuous-time Markov chains with generators $Q^{r}$ and $Q^{c}$, respectively. For parsimony, we impose separable transition intensities in continuous time. Equivalently, the joint regime
$S_t=(s_t^r,s_t^c)\in\mathcal S^r\times\mathcal S^c$ is a CTMC with generator given by
\begin{equation}
Q = Q^{r}\oplus Q^{c},
\label{eq:joint_generator}
\end{equation}
where $\oplus$ denotes the Kronecker sum.\footnote{The Kronecker-sum restriction is imposed for parsimony: it implies that the \emph{transition intensities} of the rate and credit regimes are separable in continuous time. This does not preclude strong co-movement between interest rates and credit spreads. In particular, even with fixed $Q^{r}$ and $Q^{c}$, the \emph{filtered beliefs} about $s_t^{r}$ and $s_t^{c}$ can co-move because both blocks load on common economic forces through discounting and state-contingent valuation. The restriction can be relaxed by allowing a general generator on $\mathcal S$ if the data support cross-regime transition dependence.}
For $m\in\{r,c\}$ and $i\neq j$,
\begin{equation}
q^{m}_{ij}\ge 0,\qquad q^{m}_{ii}=-\sum_{j\neq i}q^{m}_{ij}.
\label{eq:regime_generator}
\end{equation}

We estimate the model under a block-recursive learning restriction---rate-regime inference uses only rate-block information---while maintaining separable regime transition intensities as in \eqref{eq:joint_generator}; details are deferred to Section~\ref{subsec:block_recursive}. This two-regime labeling is motivated by two strands of evidence. On the rate side, Markov-switching term-structure models document persistent shifts between low- and high-rate environments, reflecting changes in the macro--policy backdrop and discount-rate dynamics \citep{Hamilton1989,BansalZhou2002}. On the credit side, theory and evidence link economy-wide credit spreads and financing conditions to recurrent expansion--contraction phases of the credit cycle \citep{BernankeGertlerGilchrist1999,GilchristZakrajsek2012}.

\paragraph{RS--GCIR dynamics under $\mathbb P$ and $\mathbb Q$.}
We adopt a regime-switching generalized CIR (RS--GCIR) specification. The rate-block factors are governed by $s_t^{r}$, while the credit factor is governed by $s_t^{c}$. Let $B_t=(B_{1,t},B_{2,t},B_{3,t},B_{4,t})^\top$ be a four-dimensional standard Brownian motion (under $\mathbb P$) with $\mathbb E[dB_t\,dB_t^\top]=I_4\,dt$. To allow for contemporaneous correlation across factor innovations, let $L$ satisfy $\Omega=LL^\top$ and define $W_t$ by $dW_t=L\,dB_t$, so that $\mathbb E[dW_t\,dW_t^\top]=\Omega\,dt$.
Under the physical measure $\mathbb{P}$,
\begin{equation}
dX_{k,t}
= \kappa_k(s_{k,t})\bigl(\theta_k(s_{k,t})-X_{k,t}\bigr)\,dt
+ \sqrt{\alpha_k(s_{k,t})+\beta_k(s_{k,t})X_{k,t}}\;dW_{k,t},
\qquad k=1,2,3,4,
\label{eq:rs_gcir_P}
\end{equation}
where $s_{k,t}=s_t^{r}$ for $k=1,2,3$ and $s_{k,t}=s_t^{c}$ for $k=4$.
We impose standard admissibility conditions ensuring $\alpha_k(s)+\beta_k(s)X_{k,t}\ge 0$ almost surely.\footnote{For $\beta_k(s)>0$, letting $Y_{k,t}\equiv X_{k,t}+\alpha_k(s)/\beta_k(s)$ yields a CIR form with diffusion coefficient $\sqrt{\beta_k(s)Y_{k,t}}$, for which regime-by-regime Feller-type conditions can be imposed if desired.}

For valuation we work under an equivalent martingale measure $\mathbb{Q}$. We specify diffusion risk premia in affine form so that the RS--GCIR class is preserved under $\mathbb{Q}$. Specifically, for each factor $k$ and the relevant regime state $s\in\mathcal S_k$ (i.e., $\mathcal S_k=\{L,H\}$ for $k\le 3$ and $\mathcal S_k=\{E,C\}$ for $k=4$), we set the market price of diffusion risk to be
\begin{equation}
\Lambda_{k,t}
=\lambda_{k,s}\sqrt{\alpha_k(s)+\beta_k(s)X_{k,t}},
\qquad \lambda_{k,s}\in\mathbb{R}.
\label{eq:mpr}
\end{equation}
Let $\Lambda_t=(\Lambda_{1,t},\dots,\Lambda_{4,t})^\top$. The Girsanov change of measure
\begin{equation}
dW_{t}^{\mathbb Q}=dW_{t}+\Lambda_{t}\,dt
\end{equation}
implies that, conditional on $s_{k,t}=s$, the risk-neutral dynamics remain of G--CIR type:
\begin{equation}
dX_{k,t}
=
\kappa^{\mathbb Q}_k(s)\big(\theta^{\mathbb Q}_k(s)-X_{k,t}\big)\,dt
+
\sqrt{\alpha_k(s)+\beta_k(s)X_{k,t}}\;dW^{\mathbb Q}_{k,t},
\label{eq:rs_gcir_Q}
\end{equation}
with the parameter mapping
\begin{equation}\label{eq:rn-map-rsg}
\kappa^{\mathbb Q}_k(s)=\kappa_k(s)+\beta_k(s)\lambda_{k,s},
\qquad
\theta^{\mathbb Q}_k(s)=
\frac{\kappa_k(s)\theta_k(s)-\alpha_k(s)\lambda_{k,s}}{\kappa^{\mathbb Q}_k(s)}.
\end{equation}

Given the affine specification $\Lambda_{k,t}=\lambda_{k,s}\sqrt{\alpha_k(s)+\beta_k(s)X_{k,t}}$,
the $\mathbb{P}\to\mathbb{Q}$ mapping in \eqref{eq:rn-map-rsg} can be inverted. When $\beta_k(s)\neq 0$,
\begin{equation}\label{eq:implied-lambda-1}
\lambda_{k,s}=\frac{\kappa_k^{\mathbb Q}(s)-\kappa_k(s)}{\beta_k(s)}.
\end{equation}
Equivalently, when $\alpha_k(s)\neq 0$,
\begin{equation}\label{eq:implied-lambda-2}
\lambda_{k,s}=\frac{\kappa_k(s)\theta_k(s)-\kappa_k^{\mathbb Q}(s)\theta_k^{\mathbb Q}(s)}{\alpha_k(s)}.
\end{equation}
In the generic GCIR case $(\alpha_k(s),\beta_k(s)>0)$, \eqref{eq:implied-lambda-1} and
\eqref{eq:implied-lambda-2} coincide and provide a convenient consistency check. We impose $\kappa^{\mathbb Q}_k(s)>0$ for all $(k,s)$ to ensure mean reversion under $\mathbb{Q}$.
For parsimony, we assume the regime generators are invariant under $\mathbb{P}$ and $\mathbb{Q}$; hence the transition matrix over horizon $\Delta$ is $P^{m}(\Delta)=\exp(Q^{m}\Delta)$ for $m\in\{r,c\}$ (and $P(\Delta)=\exp(Q\Delta)$ for the composite chain when needed).

\subsection{Block-Recursive Identification and Two-Step Filtering Structure}
\label{subsec:block_recursive}

Let $Y_t^{r}$ denote observations from the benchmark rate block constructed from the China government bond (CGB) yield curve together with China Development Bank (CDB) yields, and let $Y_t^{c}$ denote the credit block observations (rating-sorted corporate yields). Define the filtration
\begin{equation}
\mathcal{I}_t^{r}:=\sigma\{Y_u^{r}:u\le t\},\qquad
\mathcal{I}_t^{c}:=\sigma\{Y_u^{c}:u\le t\},\qquad
\mathcal{I}_t:=\mathcal{I}_t^{r}\vee\mathcal{I}_t^{c}.
\label{eq:filtrations}
\end{equation}
We use $X_t^{r}\equiv (X_{1,t},X_{2,t},X_{3,t})$ for the rate block and $X_t^{c}\equiv X_{4,t}$ for the credit block. We work on a discrete observation grid with sampling interval $\Delta$ and normalize $\Delta$ to one unit in estimation. While the latent regimes driving the rate and credit blocks evolve with fixed (time-homogeneous) generators $Q^r$ and $Q^c$, we adopt a block-recursive learning structure: the benchmark rate regime and factors are identified from $\mathcal{I}_t^{r}$, and the credit regime is identified from $\mathcal{I}_t^{c}$ conditional on the rate block extracted in the first stage. This structure is motivated by evidence that risk-free-rate (monetary-policy) shocks transmit strongly to corporate yields and spreads, and that price discovery for benchmark curve movements tends to occur first in the more liquid government-bond market \citep{ioannidou2015monetary,Wright2011}. The key restrictions are: (i) the benchmark-rate regime is inferred from rate-block information only, and (ii) the credit regime is inferred from credit-block information conditional on a finite-dimensional summary of the first-stage rate-block posterior. Importantly, we do not allow the transition intensities of the credit regime to depend on the rate regime; any conditioning of the form $\pi_{t\mid t}^{c}(s^{c}\mid s^{r};\cdot)$ arises from state-contingent valuation and the observation structure rather than cross-regime transition dependence.

\begin{proposition}[Credit-regime inference conditional on the rate block]\label{prop:block_learning_general}
Let $\mathcal I_t=\mathcal I_t^r\vee \mathcal I_t^c$. Let $(X_t^{r},s_t^{r})$ denote the latent state driving the rate block (including its regime component), and let $(X_t^{c},s_t^{c})$ denote the latent state driving the credit block (including its regime component). Define the rate-block filtering distribution at time $t$ by
\[
\xi_t^{r}(\cdot)\ \equiv\ \pi\!\bigl((X_t^{r},s_t^{r})\in \cdot\ \big|\ \mathcal I_t^{r}\bigr).
\]
Let the rate-block posterior information be summarized by
\[
\mathcal R_t
\ \equiv\
\sigma\bigl\{\xi_t^{r}\bigr\},
\]
i.e., the smallest $\sigma$-field that makes the time-$t$ filtering distribution $\xi_t^r$ measurable.
Assume a block-recursive structure: (i) the rate-block observation and state equations are autonomous in the sense that, conditional on $(X_t^{r},s_t^{r})$, the rate-block data $\mathcal I_t^r$ do not depend on credit-side objects; and (ii) for each $t$,
\begin{equation}
s_t^c \ \perp\!\!\!\perp\ \mathcal I_t^r\ \big|\ \bigl(\mathcal I_t^c, X_t^r, s_t^r\bigr).
\label{eq:block_ci}
\end{equation}
Then, for any $t$ and any $s$ in the support of $s_t^c$,
\begin{equation}
\pi\!\left(s_t^{c}=s\ \big|\ \mathcal I_t\right)
=
\pi\!\left(s_t^{c}=s\ \big|\ \mathcal I_t^{c}\vee \mathcal R_t\right).
\label{eq:prop_exact_general}
\end{equation}
Moreover,
\begin{equation}
\pi\!\left(s_t^{c}=s\ \big|\ \mathcal I_t^{c},\mathcal I_t^{r}\right)
=
\mathbb E\!\left[
\pi\!\left(s_t^{c}=s\ \big|\ \mathcal I_t^{c}, X_t^r, s_t^r\right)
\ \Big|\ \mathcal I_t^{c},\mathcal I_t^{r}
\right],
\label{eq:prop_iterated_general}
\end{equation}
so that rate-block observations affect inference about $s_t^{c}$ only through their implications for the posterior of the latent rate-block state $(X_t^r,s_t^r)$ (and hence through the implied benchmark discount curve).
\end{proposition}

\begin{proof}
Fix $t$ and $s$. By \eqref{eq:block_ci},
\[
\pi\!\left(s_t^{c}=s\ \big|\ \mathcal I_t^{c},\mathcal I_t^{r},X_t^r,s_t^r\right)
=
\pi\!\left(s_t^{c}=s\ \big|\ \mathcal I_t^{c},X_t^r,s_t^r\right).
\]
Taking conditional expectations given $(\mathcal I_t^{c},\mathcal I_t^{r})$ yields \eqref{eq:prop_iterated_general}. Finally, $\mathcal R_t$ is generated by the time-$t$ filtering distribution $\xi_t^{r}$, so any functional of $\mathcal I_t^r$ that enters inference only through the induced posterior of $(X_t^r,s_t^r)$ can be conditioned on $\mathcal R_t$ without loss, giving \eqref{eq:prop_exact_general}.
\end{proof}

Proposition \ref{prop:block_learning_general} is stated in terms of the exact posterior object $\mathcal R_t$. In estimation we inject a finite-dimensional summary $\widehat{\mathcal R}_t$ constructed from regime-conditional filtered states and regime beliefs. This approximation is disciplined by the first-stage regime-switching filter output and is treated as a generated regressor in the second-stage likelihood construction.

The rate-block observation equation is
\begin{equation}
Y_t^{r}=h^{r}\bigl(X_{1:3,t},s_t^{r}\bigr)+\varepsilon_t^{r},
\qquad
\varepsilon_t^{r}\mid s_t^{r}\sim\mathcal{N}\!\left(0,R^{r}_{s_t^{r}}\right),
\label{eq:obs_rate}
\end{equation}
with state dynamics \eqref{eq:rs_gcir_P} for $k=1,2,3$. A regime-switching filter yields the (time-$t$) filtered and (one-step-ahead) predicted beliefs
\begin{equation}
\pi_{t\mid t}^{r}(s)=\pi\!\left(s_t^{r}=s\mid \mathcal{I}_t^{r}\right),
\qquad
\pi_{t+\Delta\mid t}^{r}(s)=\pi\!\left(s_{t+\Delta}^{r}=s\mid \mathcal{I}_t^{r}\right),
\label{eq:pi_r_filtered_predicted}
\end{equation}
as well as filtered factor paths $\widehat{X}_{1:3,t}$.

Proposition \ref{prop:block_learning_general} implies that, for credit-regime inference and pricing, the rate-block data matter only through the posterior information $\mathcal R_t$ they generate about the latent rate-block state. In implementation, we inject a tractable finite-dimensional summary that preserves regime-contingent information from the first stage,
\begin{equation}
\widehat{\mathcal{R}}_t
=\sigma\Bigl\{\widehat{X}_{1,t}(s^{r}),\widehat{X}_{2,t}(s^{r}),\widehat{X}_{3,t}(s^{r}),\pi^{r}_{t\mid t}(s^{r})
:\ s^{r}\in\mathcal S^{r}\Bigr\},
\label{eq:rate_injection}
\end{equation}
where $\widehat{X}_{1:3,t}(s^{r})$ denotes the regime-conditional filtered rate factors under the event $\{s_t^{r}=s^{r}\}$ produced by the regime-switching filter, and $\pi^{r}_{t\mid t}(s^{r})$ denotes the corresponding filtered rate-regime belief.

The credit observation equation is
\begin{equation}
Y_t^{c}
= h^{c}\!\Bigl(X_{4,t},s_t^{c};\widehat{X}_{1:3,t}(s_t^{r}),s_t^{r}\Bigr)
+ \varepsilon_t^{c},
\qquad
\varepsilon_t^{c}\mid s_t^{c}\sim\mathcal{N}\!\left(0,R^{c}_{s_t^{c}}\right),
\label{eq:obs_credit}
\end{equation}
with $X_{4,t}$ following \eqref{eq:rs_gcir_P} indexed by $s_t^{c}$.
For each fixed $s^{r}\in\mathcal S^{r}$, define the conditional credit-regime filtering belief
\begin{equation}
\pi_{t\mid t}^{c}(s^{c}\mid s^{r})
=\pi\!\left(s_t^{c}=s^{c}\ \middle|\ \mathcal{I}_t^{c},\widehat{\mathcal{R}}_t,\,s_t^{r}=s^{r}\right).
\end{equation}
The (marginal) credit-regime belief integrates over the rate regime using the first-stage beliefs:
\begin{equation}
\pi_{t\mid t}^{c}(s^{c})
=\sum_{s^{r}\in\mathcal S^{r}}\pi_{t\mid t}^{c}(s^{c}\mid s^{r})\,\pi_{t\mid t}^{r}(s^{r}).
\label{eq:pi_c_marginal_mixture}
\end{equation}

We implement the two-step (block-recursive) estimation on a discrete observation grid and normalize the observation interval to one unit. The two-step RS--UKF recursions and the corresponding likelihood construction are given in Section~\ref{sec:estimation}. Throughout, we work with filtered (rather than smoothed) beliefs to avoid look-ahead bias.

\subsection{Rating-based corporate pricing under RS--GCIR}
\label{subsec:pricing_setup}

This section specifies the key pricing objects: (i) CGB and CDB zero-coupon bonds used for cash-flow valuation, and (ii) rating-based default and migration intensities for corporate bonds. We follow the rating-based reduced-form framework of \citet{lando1998cox,feldhutter2008decomposing} and combine it with the RS--GCIR state dynamics introduced above. The benchmark-rate regime $s_t^{r}$ and the systemic credit regime $s_t^{c}$ are modeled as separate latent Markov components with separable transition intensities (Section~\ref{subsec:theory_block_learning}), so that the credit-regime generator $Q^{c}$ does not vary across rate environments. Cross-block interactions therefore enter pricing only through state-contingent discounting and through the credit-risk driver that maps rate factors into default and migration intensities.

A distinctive feature of China’s onshore bond market is that the CGB curve may embed a sizable convenience component---stemming from liquidity, collateral specialness, and regulatory demand---which can depress CGB yields relative to CDB yields. Using the CGB curve to discount risky corporate cash flows would therefore mechanically inflate prices and compress implied credit spreads. We instead use the CDB curve as the benchmark discount curve for valuing corporate cash flows and explicitly model the CGB--CDB convenience-yield (spread) within the rate block. This choice isolates non-credit convenience effects from credit-risk pricing and yields a discounting benchmark that is closer to the marginal valuation curve used by market participants for corporate bonds.

Under $\mathbb Q$, we adopt an affine short-rate normalization in which the common benchmark discount-rate level is spanned by the first two rate factors, while the third factor captures the CGB--CDB convenience yield:
\begin{align}
r_t^{\mathrm{cgb}}
&= X_{1,t}+X_{2,t},
\label{eq:short_rate_cgb_norm}\\
r_t^{\mathrm{cdb}}
&= X_{1,t}+X_{2,t}+X_{3,t}.
\label{eq:short_rate_cdb_norm}
\end{align}
This normalization is an identification choice: any constant or affine level term can be absorbed into the definition of the state vector (equivalently, one may work with demeaned factors and a constant, without changing bond-price implications). With \eqref{eq:short_rate_cgb_norm}--\eqref{eq:short_rate_cdb_norm}, $X_{3,t}$ isolates the systematic spread between the CGB and CDB curves attributed to convenience yields, while $s_t^{r}$ governs the common rate environment for both markets.

We price CGB and CDB discount bonds under full regime switching, allowing $s_u^{r}$ to evolve over $u\in[t,T]$. For each $s\in\mathcal S^{r}$, define the regime-conditional discount-bond prices
\begin{align}
P_{\mathrm{cgb}}^{(s)}(t,T)
&=
\mathbb E_t^{\mathbb Q}\!\left[
\exp\!\left(-\int_t^T r_u^{\mathrm{cgb}}\,du\right)
\ \Big|\ X_t,\,s_t^{r}=s
\right],
\label{eq:bond_price_def_cgb_rs}\\
P_{\mathrm{cdb}}^{(s)}(t,T)
&=
\mathbb E_t^{\mathbb Q}\!\left[
\exp\!\left(-\int_t^T r_u^{\mathrm{cdb}}\,du\right)
\ \Big|\ X_t,\,s_t^{r}=s
\right].
\label{eq:bond_price_def_cdb_rs}
\end{align}
These regime-conditional prices solve a standard system of coupled pricing equations induced by the Markov generator $Q^{r}$; equivalently, they admit a dynamic-programming representation on the observation grid. In the RS--GCIR setting, the one-step pricing operator preserves an exponential-affine functional class, so the regime-conditional bond prices can be evaluated efficiently by backward recursion using the transition matrix $P^{r}(\Delta)=\exp(Q^{r}\Delta)$. Computational details are provided in Appendix~\ref{app:pricing}.

Observable bond prices are convex mixtures of regime-conditional prices using the time-$t$ rate-regime belief:
\begin{align}
P_{\mathrm{cgb}}(t,T)
&=\sum_{s^{r}\in\mathcal S^{r}}\pi_{t\mid t}^{r}(s^{r})\,P_{\mathrm{cgb}}^{(s^{r})}(t,T),
\label{eq:mixture_cgb}\\
P_{\mathrm{cdb}}(t,T)
&=\sum_{s^{r}\in\mathcal S^{r}}\pi_{t\mid t}^{r}(s^{r})\,P_{\mathrm{cdb}}^{(s^{r})}(t,T),
\label{eq:mixture_cdb}
\end{align}
where $\pi_{t\mid t}^{r}(s^{r})$ denotes the filtered belief about $s_t^{r}$ delivered by the rate-block filter. The associated yields are
$y^{\mathrm{cgb}}(t,T)=-\tau^{-1}\log P_{\mathrm{cgb}}(t,T)$ and
$y^{\mathrm{cdb}}(t,T)=-\tau^{-1}\log P_{\mathrm{cdb}}(t,T)$, where $\tau=T-t$.

We price corporate bonds in a rating-based reduced-form framework in the spirit of \citet{feldhutter2008decomposing}. 
For the non-regime benchmark, the Cox--Markov pricing identity and the associated spectral (mode) decomposition are derived in Appendix~\ref{sec:cox-pricing-onepiece}. Let $\eta_t\in\{1,\dots,K\}$ denote the issuer’s rating, where $\{1,\dots,K-1\}$ are non-default ratings and $K$ is an absorbing default state. The key pricing ingredients are (i) the discounting short rate for risky cash flows and (ii) a common positive driver that scales rating migration and default intensities.

Risky cash flows are discounted using the CDB curve. Under our normalization,
\begin{equation}
r_t^{\mathrm{cdb}} = X_{1,t}+X_{2,t}+X_{3,t},
\label{eq:rt_cdb_recall}
\end{equation}
so that the convenience-yield factor $X_{3,t}$ enters corporate valuation \emph{only through discounting} (i.e., through the pricing kernel), rather than through the default driver.

To accommodate four pass-through coefficients while keeping regime transitions separable (i.e., $Q^c$ invariant across $s^r$), we define a joint-regime dependent intensity driver. For each $(s^{r},s^{c})\in\mathcal S^{r}\times\mathcal S^{c}$, define
\begin{equation}
\mu_t^{(s^{r},s^{c})}
\ \equiv\
\mu\!\bigl(X_t,(s^{r},s^{c})\bigr)
\ :=\
c_{s^{c}\mid s^{r}}\bigl(X_{1,t}+X_{2,t}\bigr)+X_{4,t},
\label{eq:mu_driver_joint_state}
\end{equation}
where
\begin{equation}
c_{s^{c}\mid s^{r}}\in\{c_{E\mid L},c_{C\mid L},c_{E\mid H},c_{C\mid H}\}.
\end{equation}
This specification allows the pass-through from the rate block to credit risk to differ across rate environments without letting the credit-regime switching speed depend on the rate regime (the generator $Q^{c}$ remains the same for all $s^{r}$). We impose parameter and admissibility restrictions so that $\mu_t^{(s^{r},s^{c})}$ remains nonnegative in-sample; in estimation and pricing we additionally apply a small floor when needed to guard against numerical violations.

A modeling choice that matters directly for the pricing recursion is the treatment of recovery upon default. In the baseline specification, we impose a zero-recovery benchmark and interpret the resulting default component as an effective loss intensity under the risk-neutral measure. This assumption is adopted as a parsimonious normalization rather than as a literal statement that every defaulted bond delivers no ex post payoff. Available evidence for the Chinese credit bond market suggests that realized recoveries are generally low and highly dispersed, with substantial mass near the lower bound. Using default settlement data for 2014--2021, \cite{LiLiuWang2025} document that the average recovery rate is about 8\% within one year after default and about 11\% on a cumulative basis by the end of 2021, and that the cross-sectional distribution is markedly polarized, with observations concentrated near 0\% and 100\%. Complementary evidence in \cite{PanZhangTao2021} further shows that recovery outcomes vary systematically with bond and issuer characteristics, indicating substantial heterogeneity rather than a stable recovery parameter that can be cleanly identified in reduced-form term-structure estimation. Against this background, adopting a zero-recovery benchmark provides a defensible first-order approximation for our pricing exercise: it avoids introducing an additional weakly identified parameter and is consistent with an environment in which realized recoveries are, on average, low, delayed, and often close to zero. Accordingly, the pricing formulas below are written in terms of loss-adjusted intensities under \(\mathbb{Q}\). Under \(\mathbb{Q}\) and fractional recovery of market value, we work with loss-adjusted intensities. For each joint regime $(s^{r},s^{c})$ and $i\in\{1,\dots,K-1\}$ set
\[
a_{iK}^{(s^{r},s^{c})}(t)=\nu_i\,\mu_t^{(s^{r},s^{c})},\qquad \nu_i>0,
\]
and for migrations between non-default states $j\neq i$ impose
\[
a_{ij}^{(s^{r},s^{c})}(t)=\lambda_{ij}\,\mu_t^{(s^{r},s^{c})},\qquad i\neq j,\ \lambda_{ij}\ge 0.
\]
Let $\Lambda=(\lambda_{ij})$ denote a baseline rating generator (with diagonal entries given by negative row sums), and let $\widetilde\Lambda\in\mathbb R^{(K-1)\times(K-1)}$ denote its sub-generator on the non-default rating states. Assume $\widetilde\Lambda$ is diagonalizable and write the spectral decomposition
\[
\widetilde\Lambda=\widetilde B\,\widetilde D\,\widetilde B^{-1},
\qquad
\widetilde D=\mathrm{diag}(d_1,\dots,d_{K-1}),
\qquad d_j<0,
\]
where $\widetilde B$ collects the eigenvectors of $\widetilde\Lambda$ and $\widetilde D$ contains its eigenvalues. Following \citet{lando1998cox, feldhutter2008decomposing}, we encode the absorbing default state by defining the ``default column'' of $\widetilde B^{-1}$ as
\begin{equation}
\big[\widetilde B^{-1}\big]_{j,K}:=-\sum_{k=1}^{K-1}\big[\widetilde B^{-1}\big]_{j,k},
\qquad j=1,\dots,K-1.
\label{eq:default_column_Binv}
\end{equation}
This yields the (row-$i$) mode weights
\begin{equation}
w_{ij}:=-\big[\widetilde B\big]_{i,j}\,\big[\widetilde B^{-1}\big]_{j,K},
\qquad i=1,\dots,K-1,\ \ j=1,\dots,K-1,
\label{eq:wij_def}
\end{equation}
which satisfy $w_{ij}\ge 0$ and $\sum_{j=1}^{K-1}w_{ij}=1$.\footnote{Nonnegativity and unit row sums hold under the standard rating-based generator restrictions used in \citet{lando1998cox, feldhutter2008decomposing}; in practice we enforce these restrictions when calibrating $\widetilde\Lambda$.}

Let $S_t=(s_t^r,s_t^c)$ denote the joint regime with generator $Q=Q^r\oplus Q^c$.
For each joint regime $S=(s^r,s^c)$ and each non-default rating $i\in\{1,\dots,K-1\}$, define the time-$t$
pre-default value of a zero-coupon corporate bond maturing at $T$ by the mode decomposition
\begin{equation}
v^{\,i,S}(t,T)
=
\sum_{j=1}^{K-1} w_{ij}\,
\mathbb E_t^{\mathbb Q}\!\left[
\exp\!\left(
\int_t^T
\Big(
d_j\,\mu_u^{S_u}-r_u^{\mathrm{cdb}}
\Big)\,du
\right)
\ \Big|\ X_t,\,S_t=S
\right],
\label{eq:lando_mode_joint_kernel}
\end{equation}
where $\mu_u^{S_u}$ uses the contemporaneous joint regime $S_u=(s_u^r,s_u^c)$ in \eqref{eq:mu_driver_joint_state}. In the RS--GCIR setting, the pricing operator for each eigenmode preserves an exponential-affine functional class. Accordingly, $v^{\,i,S}(t,T)$ can be computed efficiently on the observation grid by backward recursion using the joint transition matrix $P(\Delta)=\exp(Q\Delta)$  (see Appendix~\ref{app:pricing} for the one-step affine blocks and the recursion).

Because defaultable-bond valuation is nonlinear through an exponential pricing kernel, observable prices must be formed by convex mixtures of joint-regime conditional prices at the \emph{price level}. Under block-recursive learning, we construct joint weights without joint filtering:
\begin{equation}
v^{\,i}(t,T)
=
\sum_{s^{r}\in\mathcal S^{r}}
\sum_{s^{c}\in\mathcal S^{c}}
\pi_{t\mid t}^{r}(s^{r})\,
\pi_{t\mid t}^{c}\!\left(s^{c}\mid s^{r};\widehat{\mathcal R}_t\right)\,
v^{\,i,(s^{r},s^{c})}(t,T),
\label{eq:corp_joint_mixture_block}
\end{equation}
where $\pi_{t\mid t}^{r}$ and $\pi_{t\mid t}^{c}(\cdot\mid\cdot)$ are filtered beliefs delivered by the two-step filter. This price-level mixing avoids Jensen-type distortions that would arise from averaging drivers or states before entering the exponential kernel. The corresponding corporate yield is $y^{i}(t,T)=-\tau^{-1}\log v^{\,i}(t,T)$, where $\tau=T-t$.

\section{Data and Empirical Design}
\label{sec:data}

This section describes the data used in the empirical analysis and the construction of key pricing inputs.
We use weekly fitted zero-coupon term structures from ChinaBond for (i) the rates block (CGB and CDB curves) and
(ii) the credit block (corporate curves by domestic rating) over 18 April 2014 to 31 July 2025.
To separate discounting from credit compensation, we construct matched-maturity spreads (CDB over CGB; corporate over CDB)
on a common maturity grid.
In addition, we build rating-migration and default inputs for pricing: a China-specific \(\mathbb P\)-measure transition matrix
from public rating-migration disclosures, and a parsimonious \(\mathbb Q\)-measure transition system and its intensity
representation implied by market term structures.
\subsection{Chinese market data}
China’s onshore bond market is dominated on the rates side by central government bonds (CGB) and policy-bank bonds, especially China Development Bank (CDB) bonds. The credit segment is largely issued by state-owned enterprises (SOEs), local government financing vehicles (LGFVs), and a smaller set of private firms. Domestic credit ratings follow a letter-grade system broadly comparable to international conventions, but market activity is highly concentrated in high ratings. Using Wind trading-volume data, Figure~\ref{fig:rating_share_2024} shows that secondary-market trading volume clusters in the AAA/AA range, motivating our focus on investment-grade rating buckets.

Our empirical analysis uses weekly zero-coupon yield curves from ChinaBond\footnote{\url{https://www.chinabond.com.cn/}} over the sample from 18 April 2014 to 31 July 2025. We start shortly after the first widely recognized onshore corporate bond default (Shanghai Chaori Solar on 7 March 2014), which marked a turning point in domestic credit-risk pricing. For the rate block, we use the ChinaBond fitted CGB zero-coupon curve and the fitted CDB zero-coupon curve. We retain maturities of 1, 2, 3, 4, 5, 7, and 10 years, which cover both the short end and the intermediate-to-long segment of the term structure.

For the credit block, corporate yields are taken from ChinaBond’s fitted zero-coupon curves by domestic rating. Guided by trading activity and to keep migration/default modeling tractable, we focus on four investment-grade buckets—AAA, AA+, AA, and AA-. We retain maturities of 1, 2, 3, 4, 5, 7, and 10 years; beyond 10 years, turnover is thinner and fitted curves are more sensitive to illiquidity and issue-specific noise.

We also construct matched-maturity spreads to separate credit and discounting components. The CDB spread is defined as the CDB zero-coupon yield minus the CGB zero-coupon yield at the same maturity. The corporate spread is defined as the corporate zero-coupon yield (by rating) minus the CDB zero-coupon yield at the same maturity. All series are aligned at the weekly frequency and on the same maturity grid.

Table~\ref{tab:ts_stats_corp_yields} reports summary statistics for corporate bond yields across the four rating buckets and maturities up to 10 years. Average yields increase with both maturity and credit risk, and lower-rated yields are modestly more volatile. Autocorrelations at one- and twelve-week horizons are close to unity, indicating very persistent dynamics. This motivates a flexible multi-factor specification that allows for slow-moving components and time variation in dynamics, including regime changes, rather than imposing a restrictive single-regime structure.

\begin{figure}[!t]
    \centering 
    \includegraphics[width=0.75\linewidth]{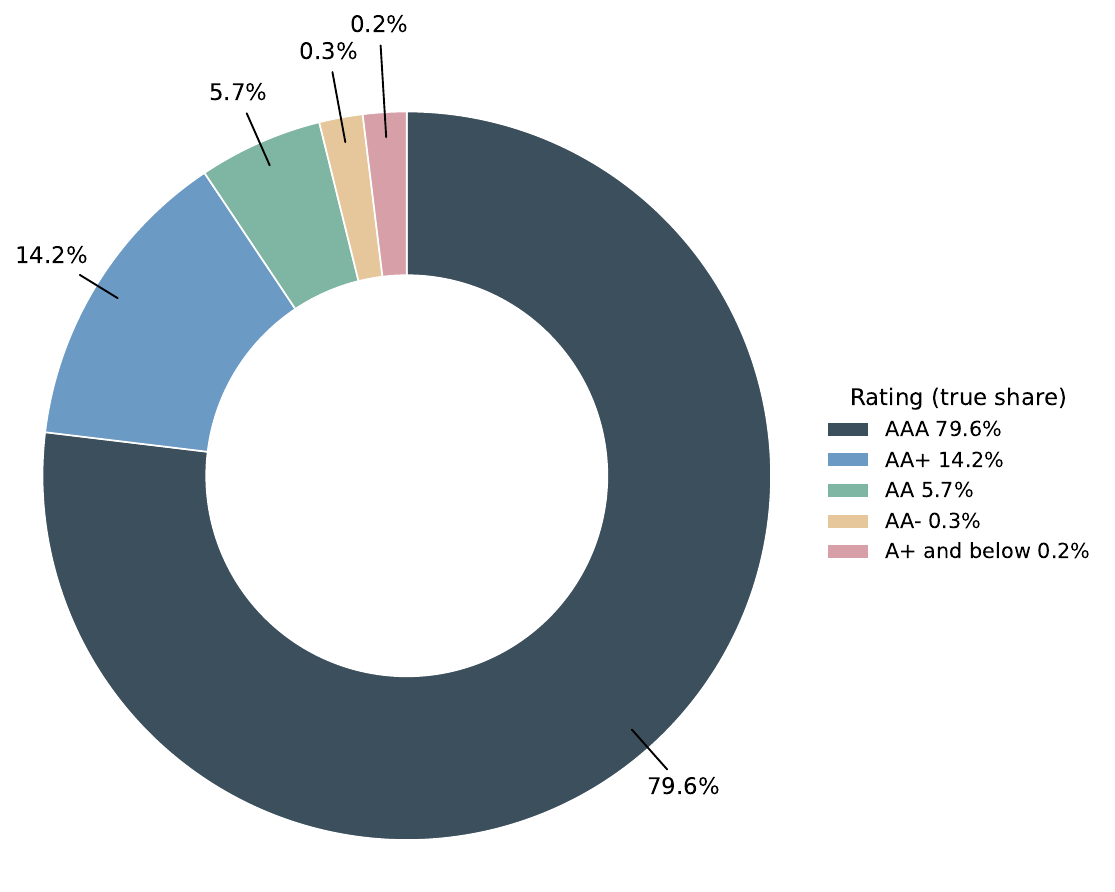}
    \caption{Rating distribution of China credit bond secondary-market trading volume (by notional amount), 2024 (source: Wind).}
    \label{fig:rating_share_2024}
    \vspace{-0.4em}
\end{figure}

\begin{table}[!t]
    \centering
    \caption{Time-series properties of weekly corporate zero-coupon yields by rating (maturities up to 10 years)}
    \label{tab:ts_stats_corp_yields}
    \footnotesize
    \setlength{\tabcolsep}{4pt}
    \renewcommand{\arraystretch}{1.1}
    \begin{tabular*}{\linewidth}{@{\extracolsep{\fill}} llrrrrrr}
        \toprule
        Rating & Maturity & Mean & Std. & AR(1) & AR(12) & ADF $p$-value & KPSS $p$-value \\
        \midrule
        AAA & 1Y  & 3.16 & 0.91 & 0.994 & 0.891 & 0.393 & 0.010 \\
            & 2Y  & 3.35 & 0.91 & 0.994 & 0.906 & 0.470 & 0.010 \\
            & 3Y  & 3.47 & 0.90 & 0.994 & 0.903 & 0.391 & 0.010 \\
            & 4Y  & 3.62 & 0.89 & 0.995 & 0.913 & 0.427 & 0.010 \\
            & 5Y  & 3.71 & 0.89 & 0.995 & 0.916 & 0.430 & 0.010 \\
            & 7Y  & 3.88 & 0.90 & 0.997 & 0.932 & 0.500 & 0.010 \\
            & 10Y & 4.03 & 0.90 & 0.997 & 0.940 & 0.622 & 0.010 \\
        \midrule
        AA+ & 1Y  & 3.38 & 1.00 & 0.994 & 0.898 & 0.385 & 0.010 \\
            & 2Y  & 3.59 & 1.00 & 0.995 & 0.910 & 0.262 & 0.010 \\
            & 3Y  & 3.73 & 1.00 & 0.995 & 0.908 & 0.303 & 0.010 \\
            & 4Y  & 3.92 & 1.00 & 0.995 & 0.919 & 0.333 & 0.010 \\
            & 5Y  & 4.04 & 1.00 & 0.996 & 0.920 & 0.506 & 0.010 \\
            & 7Y  & 4.29 & 1.03 & 0.997 & 0.936 & 0.464 & 0.010 \\
            & 10Y & 4.48 & 1.05 & 0.998 & 0.947 & 0.601 & 0.010 \\
        \midrule
        AA  & 1Y  & 3.60 & 1.06 & 0.994 & 0.902 & 0.400 & 0.010 \\
            & 2Y  & 3.86 & 1.07 & 0.996 & 0.916 & 0.326 & 0.010 \\
            & 3Y  & 4.06 & 1.05 & 0.996 & 0.915 & 0.382 & 0.010 \\
            & 4Y  & 4.33 & 1.06 & 0.996 & 0.927 & 0.609 & 0.010 \\
            & 5Y  & 4.48 & 1.06 & 0.996 & 0.931 & 0.469 & 0.010 \\
            & 7Y  & 4.77 & 1.11 & 0.998 & 0.943 & 0.703 & 0.010 \\
            & 10Y & 4.99 & 1.15 & 0.998 & 0.952 & 0.692 & 0.010 \\
        \midrule
        AA- & 1Y  & 5.18 & 1.01 & 0.994 & 0.862 & 0.798 & 0.010 \\
            & 2Y  & 5.56 & 1.05 & 0.995 & 0.898 & 0.899 & 0.010 \\
            & 3Y  & 5.82 & 1.06 & 0.996 & 0.902 & 0.925 & 0.010 \\
            & 4Y  & 6.13 & 1.09 & 0.996 & 0.919 & 0.885 & 0.010 \\
            & 5Y  & 6.30 & 1.10 & 0.996 & 0.920 & 0.874 & 0.010 \\
            & 7Y  & 6.69 & 1.04 & 0.997 & 0.937 & 0.980 & 0.010 \\
            & 10Y & 6.91 & 1.08 & 0.998 & 0.947 & 0.972 & 0.010 \\
        \bottomrule
    \end{tabular*}

    \vspace{0.3em}
    \raggedright\footnotesize
    Notes: This table reports summary statistics for weekly corporate zero-coupon yields from ChinaBond across four rating buckets (AAA, AA+, AA, and AA-) and maturities from 1 to 10 years. Mean and Std.\ denote the sample mean (in percent) and standard deviation. AR(1) and AR(12) are sample autocorrelations at lags 1 and 12 weeks. The ADF $p$-value corresponds to the augmented Dickey--Fuller test with an intercept (lag length selected by an information criterion). The KPSS $p$-value corresponds to the KPSS test of level stationarity with an intercept.
\end{table}

\subsection{Baseline State Specification and Measurement Mapping}
\label{subsec:baseline_state_spec}

The baseline latent-state vector in Section \ref{sec:model} is not chosen by mechanically applying statistical factor-selection criteria to the full panel of yields and spreads.
Instead, it is a structural specification designed to match the pricing decomposition of China’s bond market and the block-recursive identification strategy of the model.
The empirical objective is to separate three objects that are distinct in valuation:
(i) the common benchmark discount-rate component underlying the CGB and CDB curves,
(ii) the systematic CDB--CGB spread, and
(iii) the common credit-compensation component in corporate yields measured relative to CDB.
This leads to the four-state baseline
\[
X_t = (X_{1,t},X_{2,t},X_{3,t},X_{4,t})^\top,
\]
organized into a three-state rate block and a one-state credit block.

On the rates side, the data indicate that CGB and CDB curves share a large common term-structure component, but are not interchangeable.
A one-factor rate specification would be too restrictive for weekly zero-coupon curves from 1 to 10 years because it would primarily generate parallel shifts and would not provide enough flexibility for persistent slope variation.
We therefore use two common rate states, \(X_{1,t}\) and \(X_{2,t}\), as the minimal benchmark block for discounting dynamics.
These two states are intended to summarize the persistent level- and slope-type movements that jointly shape the CGB and CDB curves over the sample.

The third rate state, \(X_{3,t}\), is introduced for a different reason.
The matched-maturity CDB--CGB spread is persistent and economically meaningful in China’s onshore market, reflecting liquidity, collateral, and regulatory-demand effects rather than conventional corporate-type default risk.
Treating CGB and CDB as a single curve would therefore contaminate measured corporate spreads with movements in the policy-bank spread.
The role of \(X_{3,t}\) is to isolate this convenience-yield component and to keep the benchmark discount curve used for corporate valuation aligned with the CDB curve.

On the credit side, we adopt one latent systematic credit factor, \(X_{4,t}\), in the baseline specification.
This design reflects the fact that the corporate panels used in the paper are organized around investment-grade rating buckets and are priced in a rating-based reduced-form framework.
Cross-sectional heterogeneity across AAA, AA\(+\), AA, and AA\(-\) is handled primarily through rating-specific yields, migration/default inputs, and the matrix-exponential pricing structure, rather than by assigning an independent latent factor to each rating bucket.
The baseline credit state is therefore intended to capture the common time-varying compensation for credit risk beyond the CDB benchmark, while leaving finer cross-sectional or maturity-specific deviations to the rating structure and measurement residuals.

This specification is deliberately parsimonious.
The goal of the structural model is not to reproduce every statistical fluctuation in unrestricted yield panels, but to isolate a small set of priced states that remain economically interpretable and estimable in a regime-switching nonlinear state-space system.
In particular, enlarging the number of latent states would quickly increase the parameter dimension of the RS--GCIR model and weaken identification in the block-recursive filter, especially once separate rate and credit regimes are allowed for.
The four-state design should therefore be interpreted as a baseline structural decomposition of discounting, spread dynamics, and systematic credit compensation, rather than as a literal factor-count result from a purely statistical dimension-reduction exercise.

Table~\ref{tab:panel_state_mapping} summarizes the mapping from observable panels to the baseline state vector.
The next subsection turns from state-vector design to regime-number choice and asks whether the data support a parsimonious two-state regime structure.

\begin{table}[!t]
    \centering
    \caption{Mapping from observable panels to the baseline state specification.}
    \label{tab:panel_state_mapping}
    \footnotesize
    \renewcommand{\arraystretch}{1.2}
    \setlength{\tabcolsep}{6pt}
    \begin{threeparttable}
    \begin{tabular*}{\linewidth}{@{\extracolsep{\fill}} lccc}
        \toprule
        Observable panel & Baseline state(s) & Economic object & Role in the model \\
        \midrule
        CGB yields                & $X_{1,t},X_{2,t}$ & Common rate dynamics     & Benchmark discounting \\
        CDB--CGB spread           & $X_{3,t}$         & Convenience       & Separates CDB from CGB \\
        Corporate--CDB spreads    & $X_{4,t}$         & Systematic credit factor & Common credit compensation \\
        \bottomrule
    \end{tabular*}
    \begin{tablenotes}[flushleft]
        \footnotesize
        \item Notes: This table summarizes the economic mapping underlying the baseline latent-state vector
        \(X_t=(X_{1,t},X_{2,t},X_{3,t},X_{4,t})^\top\).
        The first two states capture the common benchmark term-structure movements needed for discounting.
        The third state isolates the persistent CDB--CGB spread, which reflects non-credit components such as liquidity,
        collateral, and regulatory demand.
        The fourth state captures the common time-varying credit component in corporate yields measured relative to CDB.
        The table is intended as a design summary rather than a statistical factor-selection result.
    \end{tablenotes}
    \end{threeparttable}
\end{table}

\subsection{Regime diagnostics and the two-state choice}

\label{sec:empirical}
\begin{figure}[!t]
    \includegraphics[width=0.95\linewidth]{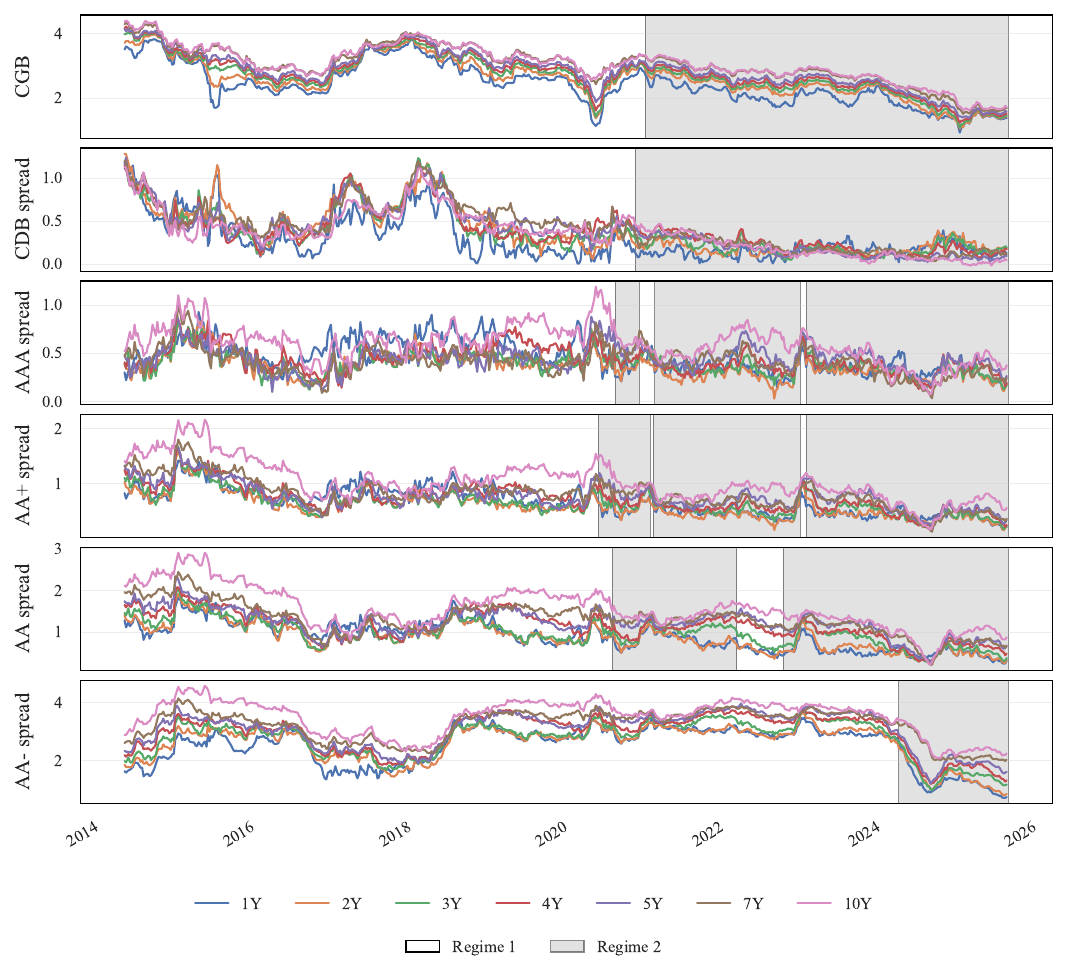}
    \centering
    \caption{Term structures of spreads for CGB, CDB, and investment-grade corporate buckets (AAA, AA\(+\), AA, AA\(-\)).
    The CDB spread is defined as the CDB yield minus the matched-maturity CGB yield.
    The corporate spread is defined as the corporate yield minus the matched-maturity CDB yield.
    Shaded areas mark the yield-based Gaussian-HMM regime sequence \(\widehat{s}^{\,\mathrm{HMM}}_{t}\in\{L,H\}\), where \(L\) (\(H\)) denotes the low-rate (high-rate) environment identified from weekly yields; see Appendix~\ref{app:hmm_diagnostics}.}
    \label{fig:spreads_plot}
\end{figure}

\begin{table}[!t]
    \centering
    \caption{Overall fit of Gaussian HMMs by bond segment (weekly yields).}
    \label{tab:gauss_hmm_fit_by_segment}
    \footnotesize
    \renewcommand{\arraystretch}{1.2}
    \setlength{\tabcolsep}{4pt}
    \begin{threeparttable}
    \begin{tabular*}{\linewidth}{@{\extracolsep{\fill}} lccc}
        \toprule
        Segment & $K = 1$ & $K = 2$ & $K = 3$ \\
        \midrule
        \multicolumn{4}{l}{\textbf{\textit{Log-likelihood}}} \\
        CGB  & $2318.95$ & $3189.72$ & $2934.08$ \\
        CDB  & $1891.06$ & $2822.21$ & $2786.99$ \\
        AAA  & $3041.25$ & $3798.26$ & $3765.61$ \\
        AA+  & $2565.46$ & $3390.36$ & $3194.05$ \\
        AA   & $1569.03$ & $2333.37$ & $2298.04$ \\
        AA-  & $998.22$  & $1713.10$ & $1747.06$ \\
        \midrule
        \multicolumn{4}{l}{\textbf{\textit{AIC}}} \\
        CGB  & $-4567.90$ & $-6233.43$ & $-5642.16$ \\
        CDB  & $-3712.13$ & $-5498.43$ & $-5347.98$ \\
        AAA  & $-6012.50$ & $-7450.52$ & $-7305.23$ \\
        AA+  & $-5060.93$ & $-6634.72$ & $-6162.10$ \\
        AA   & $-3068.06$ & $-4520.74$ & $-4370.08$ \\
        AA-  & $-1926.43$ & $-3280.20$ & $-3268.12$ \\
        \midrule
        \multicolumn{4}{l}{\textbf{\textit{BIC}}} \\
        CGB  & $-4417.06$ & $-5918.81$ & $-5155.14$ \\
        CDB  & $-3561.28$ & $-5183.80$ & $-4860.96$ \\
        AAA  & $-5861.65$ & $-7135.90$ & $-6818.21$ \\
        AA+  & $-4910.08$ & $-6320.10$ & $-5675.08$ \\
        AA   & $-2917.21$ & $-4206.12$ & $-3883.06$ \\
        AA-  & $-1775.59$ & $-2965.58$ & $-2781.10$ \\
        \bottomrule
    \end{tabular*}
    \begin{tablenotes}[flushleft]
        \footnotesize
        \item Notes: The table reports the maximized log-likelihood and the Akaike (AIC) and Bayesian (BIC) information criteria for Gaussian HMMs with $K\in\{1,2,3\}$ regimes, estimated separately for each yield segment.
        The sample length is $T=550$ weekly observations for every segment. Lower (more negative) AIC/BIC indicates a better fit after penalizing model complexity.
    \end{tablenotes}
    \end{threeparttable}
\end{table}

Before estimating the continuous-time RS--CIR model, we assess whether weekly yields exhibit discrete and persistent shifts that are consistent with a finite number of regimes.
We use Gaussian hidden Markov models (HMMs) as a reduced-form diagnostic.
The HMM is intentionally non-structural and does not impose no-arbitrage restrictions or CIR dynamics.

A Gaussian HMM combines two ingredients.
First, conditional on an unobserved regime, the data follow a Gaussian distribution with regime-specific mean and dispersion.
Second, regimes evolve as a Markov chain, which favors persistent episodes rather than isolated outliers.
This structure makes the model well suited to capture low-frequency changes in the level and variability of yields.
It also provides an operational way to compare different values of \(K\).
As \(K\) increases, the likelihood mechanically improves, but additional regimes are penalized by information criteria.
Hence, AIC and BIC offer a disciplined diagnostic for how many regimes are supported by the data. Following \citet{Hamilton1989}, we assume an unobserved Markov chain \(S_t \in \{1,\dots,K\}\).
Conditional on \(S_t=j\), the observed yield series has a Gaussian emission distribution,
\[
   y_t \mid (S_t=j) \sim \mathcal{N}(\mu_j,\Sigma_j), \qquad j = 1,\dots,K,
\]
where \(\mu_j\) summarizes the regime-specific mean level and \(\Sigma_j\) summarizes within-regime dispersion.
Regime transitions are governed by a time-homogeneous Markov matrix \(Q=(q_{ij})\), with
\(q_{ij}=\mathbb{P}(S_t=j\mid S_{t-1}=i)\).
Parameters are estimated by maximum likelihood using the forward--backward recursions and the EM/Baum--Welch algorithm
\citep{Hamilton1989,Gray1996ModelingTC,kim2017state}.

Figure~\ref{fig:spreads_plot} reports term structures of credit spreads and overlays the yield-based HMM state sequence using shaded areas.
We work with spreads rather than raw yields because spreads isolate risk premia from the common discount-rate component and allow a cleaner comparison across segments.
Specifically, the CDB spread is constructed as the CDB zero-coupon yield minus the matched-maturity CGB yield.
The corporate spread is constructed as the rating-specific corporate zero-coupon yield minus the matched-maturity CDB yield. This definition treats CDB as the intermediate benchmark and focuses the corporate spread on compensation beyond policy-bank funding conditions. The shaded classification yields long contiguous episodes, consistent with persistent regime dynamics rather than short-lived noise.
Across these episodes, spreads display systematic shifts in level and, in some panels, changes in slope across maturities.
The pattern is more pronounced for lower ratings, which exhibit higher spread levels and larger cross-maturity dispersion. At the same time, the figure shows transient widening and idiosyncratic movements that do not line up exactly with regime boundaries. This is expected because the HMM is designed to capture low-frequency shifts rather than event-driven shocks. The same regime-switching signature is visible in the yield panels reported in Appendix Figure~\ref{fig:app_hmm_yields}.
Taken together, the spread-based and yield-based evidence supports using the HMM as diagnostic guidance and motivates a parsimonious two-state structure in the structural RS--CIR model.

Table~\ref{tab:gauss_hmm_fit_by_segment} reports the maximized log-likelihood and the corresponding AIC and BIC for Gaussian HMMs with \(K\in\{1,2,3\}\), estimated separately by segment. Across the rate block (CGB and CDB) and the investment-grade buckets, moving from \(K=1\) to \(K=2\) yields a large improvement in fit, with substantially more favorable AIC and BIC values. In contrast, increasing to \(K=3\) delivers only limited incremental gains once model complexity is penalized; even for AA\(-\), where the log-likelihood rises under \(K=3\), both information criteria still favor \(K=2\). Appendix Table~\ref{tab:app_hmm_moments_yields} provides complementary support based on regime moments: under \(K=2\), the two regimes exhibit clearly separated mean levels and materially different dispersion proxies, whereas under \(K=3\) two of the three states often have very similar mean levels and dispersion proxies, making them difficult to distinguish in economic terms. Taken together, the information criteria in Table~\ref{tab:gauss_hmm_fit_by_segment} and the moment patterns in Appendix Table~\ref{tab:app_hmm_moments_yields} support a parsimonious two-regime specification as the most reliable benchmark.

\subsection{Rating migration data and risk-neutral calibration}
\label{subsec:rating_migration_Q}

This subsection develops the rating-migration inputs required for matrix-exponential corporate pricing under the risk-neutral measure $\mathbb Q$. The construction has three components: (i) a China-specific historical one-year transition matrix under $\mathbb P$ based on public migration disclosures, with tail-rating aggregation; (ii) a parsimonious one-year transition matrix under $\mathbb Q$ calibrated to spread-implied default probabilities following the rating-based reduced-form approach of \citet{lando2004}; and (iii) a continuous-time generator representation in which only default intensities are adjusted while the cross-rating migration network is held fixed, consistent with the decomposition idea in \citet{feldhutter2008decomposing}.

\paragraph{Historical transition matrix under $\mathbb P$ and state aggregation.}
China’s onshore credit market is highly concentrated in the AAA/AA+/AA segment, while ratings below $AA^{-}$ are relatively thin and noisy in both issuance and secondary trading. To improve statistical stability and align the migration system with the model state space, we construct a China-specific one-year transition matrix under $\mathbb P$ from the migration tables disclosed by \citet{CSPengyuan2025MigrationMatrix2024} and aggregate lower ratings into a coarse state space
\[
\mathcal S=\{\mathrm{AAA},\,\mathrm{AA}^+,\,\mathrm{AA},\,\mathrm{AA}^-,\,\mathrm{SG},\,\mathrm{D}\},
\]
where $\mathrm{SG}$ pools all domestic ratings below $AA^{-}$ and $\mathrm{D}$ denotes absorbing default.

Let $\mathcal S^{f}$ denote the fine-notch rating set reported in the disclosure. For $i,j\in\mathcal S^{f}$, let $N_{ij}$ be the number of obligors that start the year in fine rating $i$ and end the year in fine rating $j$, and define $N_i=\sum_j N_{ij}$. The fine-notch one-year transition probabilities are
\begin{equation}\label{eq:P_fine}
p_{ij}(1)=\frac{N_{ij}}{N_i},\qquad i,j\in\mathcal S^{f},
\end{equation}
with absorbing default imposed by $p_{DD}(1)=1$ and $p_{Dj}(1)=0$ for $j\neq D$.

The coarse transition matrix is obtained by pooling fine-notch transitions by destination and origin buckets. For any $I,J\in\mathcal S$, define
\begin{equation}\label{eq:P_coarse_simple}
\bar P^{\mathbb P}_{IJ}(1)
=\frac{\sum_{i\in I}\sum_{j\in J}N_{ij}}{\sum_{i\in I}N_i},
\qquad I,J\in\mathcal S,
\end{equation}
where $i\in I$ (respectively $j\in J$) means that fine rating $i$ (respectively $j$) belongs to coarse bucket $I$ (respectively $J$). Because disclosed tables may involve rounding, we apply a final row normalization for non-default rows:
\begin{equation}\label{eq:P_row_norm_simple}
P^{\mathbb P}_{IJ}(1)
=\frac{\bar P^{\mathbb P}_{IJ}(1)}{\sum_{L\in\mathcal S}\bar P^{\mathbb P}_{IL}(1)},
\qquad I\neq D,
\end{equation}
while preserving the absorbing default row by construction. This aggregation reduces noise in the lower-rating tail, avoids weakly identified high-dimensional migration systems, and matches the focus on actively traded investment-grade buckets.

\paragraph{Risk-neutral calibration from market curves.}
We next construct a one-year transition matrix under $\mathbb Q$ using spread-implied default probabilities from rating zero-coupon curves. Let
\[
s_i(T)=y_i(T)-y^{\mathrm{ref}}(T),
\]
denote the matched-maturity spread of rating $i\in\mathcal R:=\mathcal S\setminus\{D\}$ over a reference discounting curve $y^{\mathrm{ref}}(T)$ (baseline: the fitted CDB zero-coupon curve). Under a reduced-form approximation with constant recovery $\delta\in[0,1)$, the spread-implied cumulative default probability is approximated by
\begin{equation}\label{eq:qimp_def_simple}
q_i^{\mathrm{imp}}(T)\approx \frac{1-\exp\!\big(-s_i(T)T\big)}{1-\delta},\qquad i\in\mathcal R.
\end{equation}
When necessary, $q_i^{\mathrm{imp}}(T)$ is truncated to $[0,1]$. As usual, this object is interpreted as default-equivalent and may absorb non-credit spread components (e.g., liquidity or convenience effects).

To translate $\{q_i^{\mathrm{imp}}(T)\}$ into a tractable migration system, we start from $P^{\mathbb P}(1)$ and apply a row-wise distortion that changes only default probabilities while preserving the relative proportions of non-default transitions within each row. For each $i\in\mathcal R$, let $\pi_i\in(0,1)$ denote the one-year risk-neutral default probability and set
\begin{equation}\label{eq:PQ_default_simple}
P^{\mathbb Q}_{iD}(1)=\pi_i,\qquad
P^{\mathbb Q}_{ij}(1)=
\frac{1-\pi_i}{1-P^{\mathbb P}_{iD}(1)}\,P^{\mathbb P}_{ij}(1),
\quad j\in\mathcal S,\ j\neq D.
\end{equation}
The default row remains absorbing: $P^{\mathbb Q}_{DD}(1)=1$ and $P^{\mathbb Q}_{Dj}(1)=0$ for $j\neq D$.

The vector $\pi=(\pi_i)_{i\in\mathcal R}$ is estimated by matching the model-implied default probabilities on an integer-year horizon grid:
\begin{equation}\label{eq:pi_match_ls}
\widehat\pi\in
\arg\min_{\pi\in(0,1)^K}
\sum_{t=1}^{t_{\max}}\sum_{i\in\mathcal R}\omega_{it}
\left(\big[(P^{\mathbb Q}(1))^t\big]_{iD}-q_i^{\mathrm{imp}}(t)\right)^2,
\qquad \omega_{it}\ge 0,
\end{equation}
where $K=|\mathcal R|$ is the number of non-default rating buckets. Table~\ref{tab:P_transition_matrix} reports the resulting one-year transition matrix under $\mathbb Q$.

\begin{table}[!t]
\centering
\begin{threeparttable}
\caption{One-year risk-neutral rating transition matrix \(P^{\mathbb{Q}}(1)\) (rounded to four decimals).}
\label{tab:P_transition_matrix}
\setlength{\tabcolsep}{6pt}
\renewcommand{\arraystretch}{1.05}
\begin{tabularx}{\linewidth}{l *{6}{>{\centering\arraybackslash}X}}
\toprule
& AAA & AA+ & AA & AA- & SG & D \\
\midrule
AAA & 0.9978 & 0.0000 & 0.0000 & 0.0000 & 0.0000 & 0.0022 \\
AA+ & 0.0064 & 0.9719 & 0.0174 & 0.0013 & 0.0002 & 0.0028 \\
AA  & 0.0000 & 0.0060 & 0.9717 & 0.0187 & 0.0002 & 0.0034 \\
AA- & 0.0000 & 0.0007 & 0.0071 & 0.9798 & 0.0033 & 0.0091 \\
SG  & 0.0000 & 0.0000 & 0.0064 & 0.0208 & 0.8935 & 0.0792 \\
D   & 0.0000 & 0.0000 & 0.0000 & 0.0000 & 0.0000 & 1.0000 \\
\bottomrule
\end{tabularx}
\begin{tablenotes}[flushleft]
\footnotesize
\item Notes: States are \(\{\mathrm{AAA},\mathrm{AA}^+,\mathrm{AA},\mathrm{AA}^-,\mathrm{SG},\mathrm{D}\}\). \(\mathrm{SG}\) pools all domestic ratings below \(AA^-\), and \(\mathrm{D}\) is absorbing. Rows sum to one up to rounding.
\end{tablenotes}
\end{threeparttable}
\end{table}

\paragraph{Continuous-time generator and default-intensity adjustment.}
Matrix-exponential pricing is implemented with a continuous-time generator on the rating state space. Let the issuer rating be indexed as in Section~2.3 by $\eta_t\in\{1,\dots,K\}$, where $\{1,\dots,K-1\}$ are non-default ratings and $K$ is the absorbing default state. Define the non-default set $\mathcal R:=\{1,\dots,K-1\}$.

For $i,j\in\mathcal R$ with $i\neq j$, let $\lambda_{ij}\ge 0$ denote the risk-neutral migration intensity from rating $i$ to rating $j$, and define the total non-default migration exit intensity
\begin{equation}\label{eq:lambda_row_sum_def}
\lambda_i=\sum_{j\in\mathcal R,\ j\neq i}\lambda_{ij},\qquad i\in\mathcal R.
\end{equation}
Collect these intensities into the non-default migration matrix $\Lambda\in\mathbb R^{(K-1)\times(K-1)}$,
\begin{equation}\label{eq:Lambda_def}
\Lambda_{ij}=\lambda_{ij}\ (i\neq j),\qquad
\Lambda_{ii}=-\lambda_i,
\end{equation}
so that $\Lambda\mathbf 1=\mathbf 0$.

Let $\nu=(\nu_i)_{i\in\mathcal R}\in\mathbb R_+^{K-1}$ denote the baseline default-intensity vector implied by the embedded one-year matrix $P^{\mathbb Q}(1)$.\footnote{The embedding is obtained from a matrix-logarithm step with regularization/projection when needed to enforce a valid conservative generator with absorbing default.}
The corresponding $K\times K$ generator is
\begin{equation}\label{eq:Q_generator_baseline}
\mathcal Q(\nu)=
\begin{pmatrix}
\Lambda-\mathrm{diag}(\nu) & \nu\\
\mathbf 0^\top & 0
\end{pmatrix},
\end{equation}
where the upper-left block is $(K-1)\times(K-1)$, the default column $\nu$ is $(K-1)\times 1$, and the last row corresponds to the absorbing default state.

In the final calibration step, the off-diagonal migration network $\{\lambda_{ij}\}_{i\neq j}$ is kept fixed, and only the default intensities are adjusted:
\begin{equation}\label{eq:delta_nu_def}
\Delta\nu=(\Delta\nu_i)_{i\in\mathcal R}\in\mathbb R_+^{K-1},\qquad
\nu \mapsto \nu+\Delta\nu .
\end{equation}
Under this restriction, the update affects only the default column and the diagonal of the non-default block:
\begin{equation}\label{eq:nu_and_diag_update}
\mathcal Q(\nu+\Delta\nu)=
\begin{pmatrix}
\Lambda-\mathrm{diag}(\nu+\Delta\nu) & \nu+\Delta\nu\\
\mathbf 0^\top & 0
\end{pmatrix}.
\end{equation}
Equivalently, for each $i\in\mathcal R$,
\[
-\lambda_i-\nu_i
\;\longmapsto\;
-\lambda_i-\nu_i-\Delta\nu_i,
\]
while all off-diagonal migration intensities $\lambda_{ij}$ ($j\neq i$) remain unchanged.

The adjustment preserves the generator constraints because
\begin{equation}\label{eq:row_sum_zero_adjusted}
\big(\Lambda-\mathrm{diag}(\nu+\Delta\nu)\big)\mathbf 1+(\nu+\Delta\nu)=\mathbf 0,
\end{equation}
and all off-diagonal entries remain nonnegative. Hence $\mathcal Q(\nu+\Delta\nu)$ is a valid conservative generator and $\exp\!\big(\mathcal Q(\nu+\Delta\nu)t\big)$ is a valid transition matrix for all $t\ge 0$.

Recall the joint regime $S_t=(s_t^{r},s_t^{c})$ and the nonnegative credit-risk driver $\mu_t(S_t)$ defined in \eqref{eq:mu_driver_joint_state}. Under fractional recovery of market value, the loss-adjusted generator is
\begin{equation}\label{eq:A_t_loss_adjusted_final}
A_t^{\mathbb Q}(S_t)
=\mu_t(S_t)\,\mathcal Q(\nu+\Delta\nu).
\end{equation}
The vector $\Delta\nu$ is estimated conditional on the previously estimated term-structure and state-dynamics parameters by minimizing the relative RMSE between observed and model-implied yields. If needed for numerical stability, we additionally impose shape restrictions (e.g., monotonicity across rating buckets).

\smallskip
\noindent
Table~\ref{tab:P_transition_matrix} reports the pre-adjustment benchmark $P^{\mathbb Q}(1)$ used to initialize the continuous-time generator. The final pricing generator incorporates the subsequent $\Delta\nu$-adjustment.

\section{Estimation}
\label{sec:estimation}
We estimate the RS--GCIR model on a weekly panel of China government bond yields (CGB; central government),
China Development Bank (CDB) bond yields, and rating-sorted corporate bond yields. The model admits a
nonlinear regime-switching state-space representation with latent continuous factors $X_t$ and two latent
Markov components: a benchmark-rate regime $s_t^{r}$ and a systemic credit regime $s_t^{c}$. For
state-contingent valuation, we use the joint index $S_t=(s_t^{r},s_t^{c})$ as a bookkeeping device, while
maintaining separable regime transition intensities as in Section~\ref{subsec:theory_block_learning}. We
implement a block-recursive (two-step) estimation strategy: the rate block is filtered first using only
rate-block information to recover regime-conditional filtered states
$\{\widehat X_{1:3,t}^{(s^{r})}\}_{s^{r}\in\mathcal S^{r}}$ and filtered regime probabilities
$\pi^{r}_{t\mid t}(\cdot)$, and the credit block is then filtered conditional on these first-stage outputs
(generated regressors \citep{Pagan1984,murphy2002estimation}. Throughout, we normalize the observation
interval to one unit (one week), i.e., $\Delta=1$, and use the $t-1/t$ notation for prediction and updating.

In the rate block, the observation vector $Y_t^{r}$ (CGB curve and CDB yields) is linked to the rate
factors and the rate regime through
\begin{equation}
Y_t^{r}=h^{r}\bigl(X_{1:3,t},s_t^{r}\bigr)+\varepsilon_t^{r},
\qquad
\varepsilon_t^{r}\mid s_t^{r}\sim\mathcal{N}\!\left(0,R^{r}_{s_t^{r}}\right).
\label{eq:obs_rate_est}
\end{equation}
In the credit block, the observation vector $Y_t^{c}$ (rating-sorted corporate yields/spreads) is generated by
a nonlinear measurement map that depends on the credit state and regime and on the injected rate-block outputs.
Under the joint-regime conditional valuation route (Section~\ref{subsec:pricing_setup}), we write
\begin{equation}
Y_t^{c}
=
h^{c}\!\Bigl(X_{4,t},s_t^{c};\,\widehat X_{1:3,t}^{(s_t^{r})},\,s_t^{r}\Bigr)+\varepsilon_t^{c},
\qquad
\varepsilon_t^{c}\mid s_t^{c}\sim\mathcal{N}\!\left(0,R^{c}_{s_t^{c}}\right),
\label{eq:obs_credit_est}
\end{equation}
where $\widehat X_{1:3,t}^{(s^{r})}$ denotes the regime-conditional filtered rate state from the first-stage
filter. The dependence on $s_t^{r}$ reflects that corporate pricing objects (discounting and the pass-through
coefficient entering the intensity driver) are evaluated under the rate regime. The injected information set
$\widehat{\mathcal R}_t$ (defined in Section~\ref{subsec:block_recursive}) provides a finite-dimensional summary
of the rate-block posterior used in the second-stage recursion.

The crucial econometric feature is that $h^{c}(\cdot)$ is intrinsically nonlinear. Corporate bond prices are
computed under $\mathbb Q$ by combining (i) regime-dependent discounting from the RS--GCIR rate block with (ii)
rating-contingent cash-flow weights generated by a risk-neutral rating generator through the matrix-exponential
operator, and observed yields are obtained only after a price-to-yield transformation. This matrix-based migration
channel (and the subsequent yield inversion) induces substantial curvature in the measurement equation, especially
when joint-regime uncertainty enters the pricing kernel. We therefore evaluate the likelihood with a regime-switching
unscented Kalman filter (RS--UKF), which propagates sigma points through $h^{r}$ and $h^{c}$ rather than relying on first-order Taylor linearization. 

\subsection{Two-step likelihood evaluation.}
Let $Q^{r}$ and $Q^{c}$ denote the continuous-time generators for the rate and credit regimes, and let
\(
P^{m}(\Delta)=\exp(Q^{m}\Delta)
\)
be the $\Delta$-horizon transition matrix for $m\in\{r,c\}$.
With weekly data we set $\Delta=1$ and write $P^{m}\equiv P^{m}(1)$ with elements $p^{m}_{s's}=[P^{m}]_{s's}$.
For each regime $j$, the RS--GCIR dynamics imply a non-Gaussian conditional distribution for the factors.
Following the standard UKF construction, we approximate the one-step conditional distribution by moment matching,
\begin{equation}
X_t \mid (X_{t-1}, s_t=j)
\;\approx\;
\mathcal N\!\Big(m_j(X_{t-1}),\,V_j(X_{t-1})\Big),
\label{eq:state_moment_match_est}
\end{equation}
where $m_j(\cdot)$ and $V_j(\cdot)$ are implied by the RS--GCIR dynamics under the one-week discretization
(with a frozen-regime approximation within the week). Given parameters, the RS--UKF recursion predicts regime
probabilities, computes the regime-conditional predictive distribution of the observations via the unscented
transform, and evaluates the one-step predictive density as a finite mixture across regimes.

\medskip
\noindent\textit{Stage 1 (rate block): rate-regime filtering and likelihood.}
Given \eqref{eq:obs_rate_est}, define the regime-conditional one-step predictive likelihood
\[
\ell_t^{\,r}(s^{r})
\;\equiv\;
f\!\bigl(Y_t^{r}\mid \mathcal I_{t-1}^{r},\,s_t^{r}=s^{r}\bigr),
\qquad s^{r}\in\mathcal S^{r},
\]
evaluated numerically by the RS--UKF (Appendix~\ref{app:rsukf}). The predictive and filtered rate-regime
probabilities update by
\[
\pi_{t\mid t-1}^{r}(s^{r})
=
\sum_{s'\in\mathcal S^{r}}\pi_{t-1\mid t-1}^{r}(s')\,p^{r}_{s's^{r}},
\qquad
\pi_{t\mid t}^{r}(s^{r})
=
\frac{\ell_t^{\,r}(s^{r})\,\pi_{t\mid t-1}^{r}(s^{r})}
{\sum_{\tilde s^{r}\in\mathcal S^{r}}\ell_t^{\,r}(\tilde s^{r})\,\pi_{t\mid t-1}^{r}(\tilde s^{r})}.
\]
The implied one-step predictive density is the regime mixture
\begin{equation}
p\!\left(Y_t^{r}\mid \mathcal I_{t-1}^{r}\right)
=
\sum_{s^{r}\in\mathcal S^{r}}
\pi_{t\mid t-1}^{r}(s^{r})\,\ell_t^{\,r}(s^{r}),
\label{eq:pred_density_rate}
\end{equation}
and the first-stage log-likelihood is
\begin{equation}
\mathcal L^{r}(\Theta^{r})
=
\sum_{t}
\log p\!\left(Y_t^{r}\mid \mathcal I_{t-1}^{r};\Theta^{r}\right).
\label{eq:loglik_rate}
\end{equation}
The first-stage outputs $\bigl(\widehat X_{1:3,t}^{(s^{r})},\pi_{t\mid t}^{r}(s^{r})\bigr)_{s^{r}\in\mathcal S^{r}}$
form the injected information set $\widehat{\mathcal R}_{t}$ used in the second stage.

\medskip
\noindent\textit{Stage 2 (credit block): conditional credit-regime filtering and likelihood.}
Let $\ell_t^{\,c}(s^{c}\mid s^{r};\widehat{\mathcal R}_{t})$ denote the joint-regime conditional one-step
predictive likelihood contribution,
\[
\ell_t^{\,c}\!\left(s^{c}\mid s^{r};\widehat{\mathcal R}_{t}\right)
\;\equiv\;
f\!\bigl(Y_t^{c}\mid \mathcal I_{t-1}^{c},\,\widehat{\mathcal R}_{t},\,s_t^{c}=s^{c},\,s_t^{r}=s^{r}\bigr),
\qquad s^{c}\in\mathcal S^{c},
\]
evaluated numerically by propagating sigma points through $h^{c}$ in \eqref{eq:obs_credit_est}.
For each fixed $s^{r}\in\mathcal S^{r}$, the conditional credit-regime probabilities evolve with transition matrix
$P^{c}$ as
\[
\pi_{t\mid t-1}^{c}\!\left(s^{c}\mid s^{r}\right)
=
\sum_{s'\in\mathcal S^{c}}
\pi_{t-1\mid t-1}^{c}\!\left(s'\mid s^{r}\right)\,
p^{c}_{s's^{c}},
\]
and update to the conditional filtered probabilities
\[
\pi_{t\mid t}^{c}\!\left(s^{c}\mid s^{r};\widehat{\mathcal R}_{t}\right)
=
\frac{\ell_t^{\,c}\!\left(s^{c}\mid s^{r};\widehat{\mathcal R}_{t}\right)\,
\pi_{t\mid t-1}^{c}\!\left(s^{c}\mid s^{r}\right)}
{\sum_{\tilde s^{c}\in\mathcal S^{c}}
\ell_t^{\,c}\!\left(\tilde s^{c}\mid s^{r};\widehat{\mathcal R}_{t}\right)\,
\pi_{t\mid t-1}^{c}\!\left(\tilde s^{c}\mid s^{r}\right)}.
\]
Given the already-updated rate posterior at date $t$, the credit-block one-step predictive density conditional on
$\widehat{\mathcal R}_{t}$ is
\begin{equation}
p\!\left(Y_t^{c}\mid \mathcal I_{t-1}^{c},\widehat{\mathcal R}_{t}\right)
=
\sum_{s^{r}\in\mathcal S^{r}}
\pi_{t\mid t}^{r}(s^{r})
\sum_{s^{c}\in\mathcal S^{c}}
\pi_{t\mid t-1}^{c}\!\left(s^{c}\mid s^{r}\right)\,
\ell_t^{\,c}\!\left(s^{c}\mid s^{r};\widehat{\mathcal R}_{t}\right),
\label{eq:pred_density_credit}
\end{equation}
and the second-stage log-likelihood is
\begin{equation}
\mathcal L^{c}(\Theta^{c}\mid \widehat{\mathcal R})
=
\sum_{t}
\log p\!\left(Y_t^{c}\mid \mathcal I_{t-1}^{c},\widehat{\mathcal R}_{t};\Theta^{c}\right).
\label{eq:loglik_credit}
\end{equation}
Finally, the marginal credit-regime probability (useful for diagnostics) is obtained by mixing over the rate regime,
\begin{equation}
\pi_{t\mid t}^{c}(s^{c})
=
\sum_{s^{r}\in\mathcal S^{r}}
\pi_{t\mid t}^{r}(s^{r})\,
\pi_{t\mid t}^{c}\!\left(s^{c}\mid s^{r};\widehat{\mathcal R}_{t}\right).
\end{equation}

We estimate $\Theta^{r}$ by maximizing \eqref{eq:loglik_rate}. Conditional on the injected first-stage outputs
$\{\widehat{\mathcal R}_{t}\}_{t}$ (treated as generated regressors), we estimate $\Theta^{c}$ by maximizing
\eqref{eq:loglik_credit}. We implement a collapsing (moment-matching) step within the regime-mixture recursion to
stabilize the filter and carry forward a tractable Gaussian summary. Missing observations are handled by removing
the corresponding entries of $Y_t^{r}$ or $Y_t^{c}$ and the associated rows/columns of the predictive moments at
date $t$.

\subsection{Standard errors.}
Because the RS--UKF evaluates a numerically approximated likelihood in a nonlinear regime-switching
state-space model, we treat the resulting estimator as quasi-maximum likelihood (QMLE) and report robust
(sandwich) standard errors. Let the parameter vector be $\Theta=(\Theta^{r},\Theta^{c})$ and denote the
(per-week) log-likelihood contribution by
\(
\ell_t(\Theta)=\log p(Y_t^{r}\mid\mathcal I_{t-1}^{r};\Theta^{r})
+\log p(Y_t^{c}\mid\mathcal I_{t-1}^{c},\widehat{\mathcal R}_{t}(\Theta^{r});\Theta^{c}),
\)
where the second term depends on generated regressors $\widehat{\mathcal R}_{t}$ from the first-step filter.

Let $g_t(\Theta)=\partial \ell_t(\Theta)/\partial \Theta$ be the score and define the usual sandwich ingredients
\[
H(\Theta)=-\sum_{t}\frac{\partial^{2}\ell_t(\Theta)}{\partial\Theta\,\partial\Theta^{\top}},
\qquad
J(\Theta)=\sum_{t} g_t(\Theta)\,g_t(\Theta)^{\top}.
\]
We obtain $g_t(\Theta)$ and $H(\Theta)$ numerically using finite differences of the UKF-evaluated log-likelihood.
The robust covariance matrix is then
\[
\widehat{\mathrm{Var}}(\widehat\Theta)=H(\widehat\Theta)^{-1}\,J(\widehat\Theta)\,H(\widehat\Theta)^{-1}.
\]
This construction automatically incorporates the Murphy--Topel adjustment (generated-regressor) through the
dependence of $\ell_t(\Theta)$ on $\widehat{\mathcal R}_{t}(\Theta^{r})$ \citep{Pagan1984,murphy2002estimation}.
As an additional robustness check (e.g., under serial dependence or filter approximation error), we also report
standard errors from a block bootstrap over weeks.

\section{Empirical Results}

\subsection{Model performance}\label{subsec:emp_model_performance}

We evaluate in-sample fit using weekly zero-coupon yields for the rate block (CGB and CDB) and the rating buckets (AAA, AA$+$, AA, and AA$-$) over maturities $\tau\in\{1,2,3,4,5,7,10\}$ years.
For each segment and maturity, the pricing error is
\[
\epsilon_t(\tau)=y_t(\tau)-\hat y_t(\tau),
\]
where $y_t(\tau)$ is the observed yield and $\hat y_t(\tau)$ is the model-implied yield.
We report the sample mean and standard deviation of $\epsilon_t(\tau)$ in basis points (bp), i.e., $10{,}000\times \epsilon_t(\tau)$.
To summarize relative fit while keeping units comparable across maturities, we also report
\[
\mathrm{RRMSE}(\tau)=
\frac{\sqrt{\frac{1}{T}\sum_{t=1}^{T}\epsilon_t(\tau)^2}}{\bar y(\tau)},
\qquad
\bar y(\tau)=\frac{1}{T}\sum_{t=1}^{T} y_t(\tau),
\]
which is unitless.
The ``Average'' column is the simple average across $\tau\in\{1,2,3,4,5,7,10\}$.
Both models are evaluated on the same sample and the same observation set, so the statistics are directly comparable.

Table~\ref{tab:rsukf_pricing_error_stats} reports performance under the baseline RS--UKF specification, while Table~\ref{tab:ukf_pricing_error_stats} reports the single-regime UKF benchmark.
Across all segments, allowing for regime switching improves fit in a consistent direction.
Using the ``Average'' RRMSE as a summary, the RS--UKF reduces RRMSE from $0.058$ to $0.046$ for CGB, from $0.051$ to $0.045$ for CDB, and from $0.062$ to $0.047$ for the AAA bucket.
Improvements are also visible for AA$+$ (from $0.064$ to $0.051$) and AA (from $0.066$ to $0.058$), indicating that the gains are not limited to the rate block.

Beyond the relative metric, the mean errors show that regime switching primarily reduces persistent biases at short and intermediate maturities in the corporate buckets.
For example, under the single-regime UKF the 1-year mean errors are around $-10$ to $-14$ bp for AAA/AA$+$/AA, whereas under RS--UKF they move substantially closer to zero.
At the long end, both models display a modest positive bias for the 10-year CGB and CDB yields, suggesting that residual term-premium or curve-fitting effects remain more difficult to absorb with a parsimonious factor structure.
This pattern motivates reporting both absolute (bp) and relative (RRMSE) measures: the relative metric is informative across maturities, while the bp errors are essential for cross-rating interpretation.

Finally, we use linear-filter benchmarks as a robustness check.
Appendix~\ref{app:kf_comparison} reports analogous statistics for a Extended Kalman-filter (EKF) implementation and its regime-switching counterpart (RS--EKF).
Allowing for regime switching within an Extended KF can improve fit in several segments and maturities, but the remaining errors highlight that the measurement equation inherits nonlinearity from defaultable-bond valuation and from the price-to-yield transformation.
This evidence supports our main estimation strategy: a nonlinear filter combined with regime-conditional pricing and probability-weighted aggregation at the price level.
Figures~\ref{fig:cgb_level_slope_curv} and \ref{fig:cdb_level_slope_curv} provide a graphical diagnostic that complements the pricing-error evidence reported in Tables~\ref{tab:rsukf_pricing_error_stats} and \ref{tab:ukf_pricing_error_stats}. To summarize economically salient movements of each fitted term structure in a low-dimensional yet interpretable way, we consider three standard proxies: the short-end level $L_t \equiv y_t(1\text{Y})$, the slope $S_t \equiv y_t(10\text{Y})-y_t(1\text{Y})$, and the intermediate-maturity curvature $C_t \equiv 2y_t(5\text{Y})-y_t(1\text{Y})-y_t(10\text{Y})$. These components capture the dominant rotations and hump-shaped deformations of yield curves, and they map naturally to (i) the prevailing discount-rate environment and (ii) the maturity profile of risk compensation that enters reduced-form valuation through discounting and credit intensities \citep{lando1998cox,feldhutter2008decomposing}.

Figures~\ref{fig:cgb_level_slope_curv}--\ref{fig:cdb_level_slope_curv} show that the rate block reproduces the joint dynamics of $(L_t,S_t,C_t)$ with high fidelity over the full sample. This matters because, in our two-block design, the rate block delivers the operational discount curve for valuing corporate cash flows, whereas the credit block prices corporate bonds through an intensity-based migration--default mechanism. In the spirit of the term-structure decompositions in \citet{feldhutter2008decomposing}, corporate yields reflect joint variation in discounting and credit compensation; consequently, disciplined inference on the discounting leg is a prerequisite for attributing spread movements to credit intensities rather than to shifts in the benchmark rate environment. A subtle but informative asymmetry is that the CGB fit is marginally weaker than the CDB fit in episodes featuring localized deformations (particularly in curvature) beyond a pure level--slope rotation. Economically, this pattern is consistent with time-varying convenience yields and safe-asset premia embedded in sovereign valuations, which need not be shared one-for-one by policy-bank bonds \citep{KrishnamurthyVissingJorgensen2012,ChenChenHeLiuXie2023}. Methodologically, our specification absorbs such spread-like components via the dedicated CDB--CGB factor, thereby improving the CDB representation while keeping the CGB curve intentionally parsimonious. Importantly, this feature strengthens (rather than weakens) the pricing interpretation: using CDB as the discounting reference avoids mechanically loading sovereign convenience premia into measured corporate credit compensation.

\begin{figure*}[!t]
    \centering
    \includegraphics[width=0.98\textwidth]{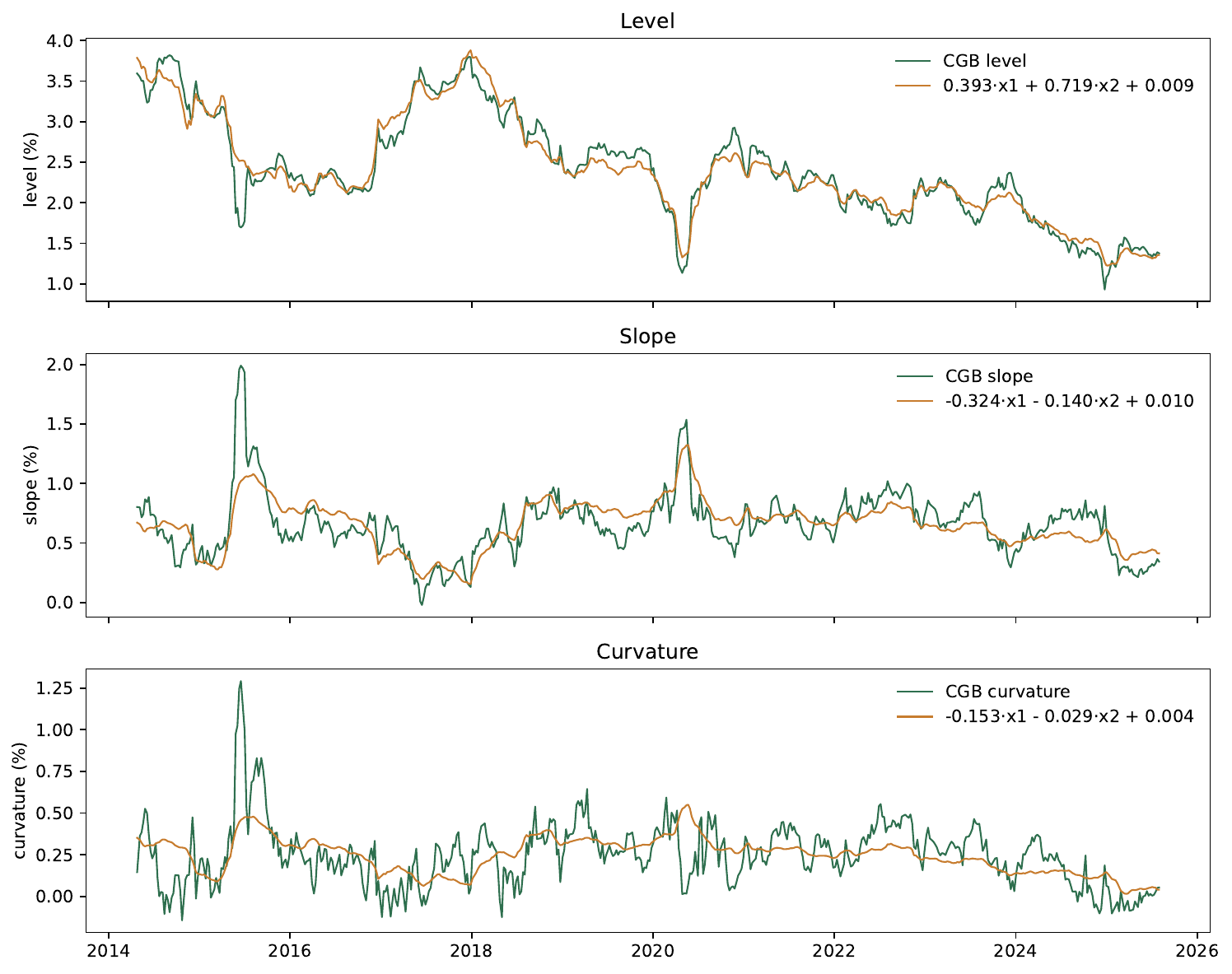}
    \caption{CGB term-structure diagnostics: observed vs.\ factor-implied level, slope, and curvature.
    The top panel reports the \emph{level} proxy $L_t \equiv y_t(1\text{Y})$.
    The middle panel reports the \emph{slope} proxy $S_t \equiv y_t(10\text{Y})-y_t(1\text{Y})$.
    The bottom panel reports the \emph{curvature} proxy $C_t \equiv 2y_t(5\text{Y})-y_t(1\text{Y})-y_t(10\text{Y})$.
    The green line is the data-based proxy computed from the observed zero-coupon yields $y_t(\tau)$, and the orange line is the factor-implied proxy obtained from the OLS projection.}
    \label{fig:cgb_level_slope_curv}
\end{figure*}

\begin{figure*}[!t]
    \centering
    \includegraphics[width=0.98\textwidth]{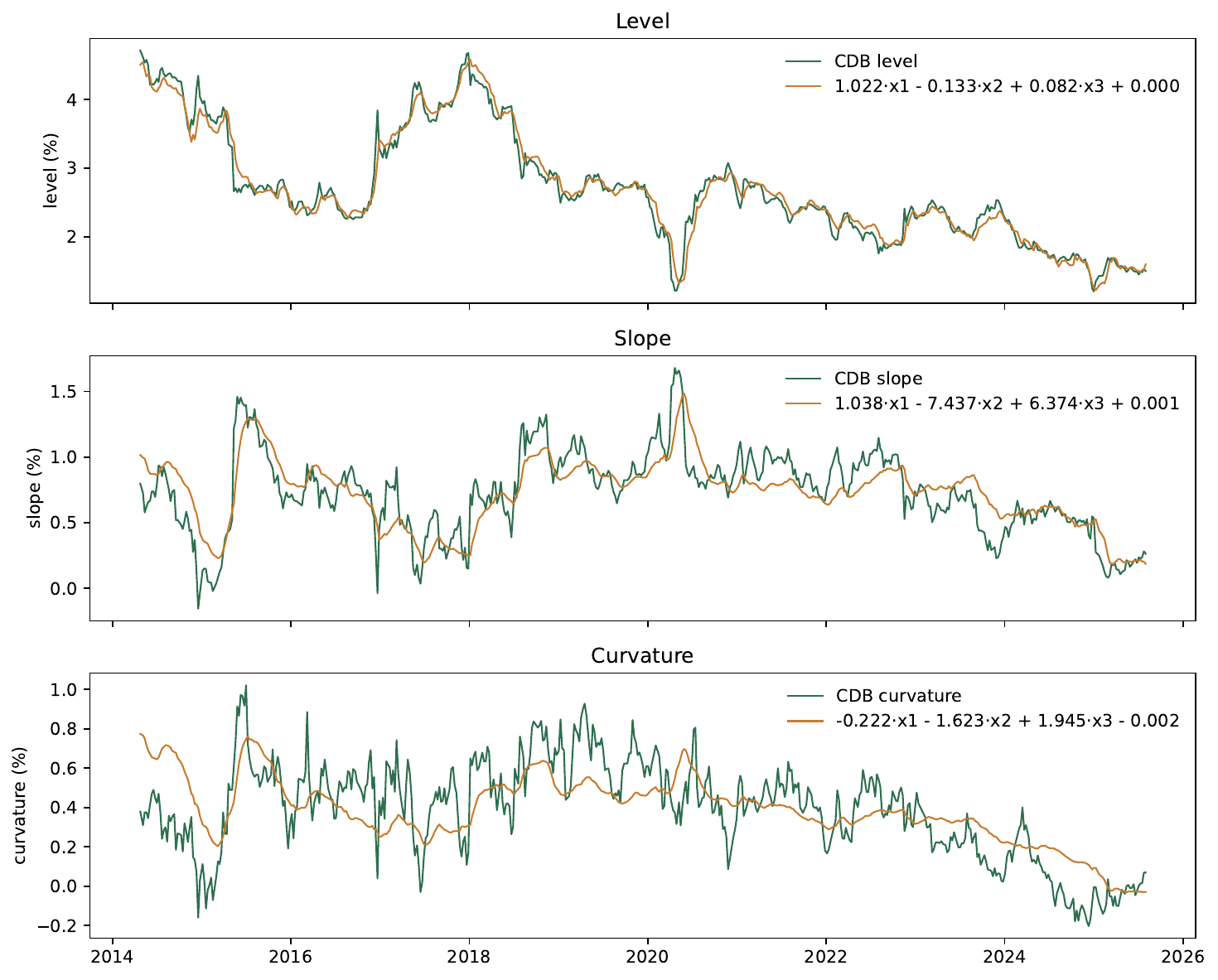}
    \caption{CDB term-structure diagnostics: observed vs.\ factor-implied level, slope, and curvature.
    The top panel reports the \emph{level} proxy $L_t \equiv y_t^{\text{CDB}}(1\text{Y})$.
    The middle panel reports the \emph{slope} proxy $S_t \equiv y_t^{\text{CDB}}(10\text{Y})-y_t^{\text{CDB}}(1\text{Y})$.
    The bottom panel reports the \emph{curvature} proxy $C_t \equiv 2y_t^{\text{CDB}}(5\text{Y})-y_t^{\text{CDB}}(1\text{Y})-y_t^{\text{CDB}}(10\text{Y})$.
    The green line is the data-based proxy computed from the observed CDB zero-coupon yields $y_t^{\text{CDB}}(\tau)$, and the orange line is the factor-implied proxy obtained from the OLS projection.}
    \label{fig:cdb_level_slope_curv}
\end{figure*}

\begin{table*}[!t]
\centering
\caption{Pricing error statistics under the RS--UKF model (up to 10Y).
The pricing error is $\epsilon_t = y_t - \hat{y}_t$, where $\hat{y}_t$ is the model-implied yield and $y_t$ is the observed yield.
Mean and standard deviation are reported in basis points (bp), i.e., $10{,}000 \times \epsilon_t$.
RRMSE is unitless.}
\label{tab:rsukf_pricing_error_stats}
\footnotesize
\setlength{\tabcolsep}{4pt}
\renewcommand{\arraystretch}{1.08}

\begin{threeparttable}
\begin{tabular*}{\linewidth}{@{\extracolsep{\fill}} lrrrrrrrr}
\toprule
& $\epsilon_{1}$ & $\epsilon_{2}$ & $\epsilon_{3}$ & $\epsilon_{4}$ & $\epsilon_{5}$ & $\epsilon_{7}$ & $\epsilon_{10}$ & Average \\
\midrule

\textit{CGB} \\
Mean       &   2.63 &  -2.70 &   2.27 &   2.09 &   3.15 &  -4.64 &  14.05 &   2.41 \\
St.~dev.   &  15.25 &  12.80 &  14.02 &  14.71 &  13.78 &  13.96 &  14.17 &  14.10 \\
RRMSE          &  0.059 &  0.044 &  0.046 &  0.042 &  0.040 &  0.039 &  0.054 &  0.046 \\
\addlinespace

\textit{CDB} \\
Mean       &   3.52 &  -1.75 &   2.87 &   2.29 &   2.99 &  -5.31 &  13.09 &   2.53 \\
St.~dev.   &  14.94 &  11.80 &  12.82 &  13.39 &  12.49 &  12.32 &  12.61 &  12.91 \\
RRMSE          &  0.062 &  0.043 &  0.045 &  0.041 &  0.040 &  0.037 &  0.050 &  0.045 \\
\addlinespace

\textit{AAA} \\
Mean       &  -1.36 &   3.00 &   6.37 &   1.42 &   1.40 &  -4.82 &  -9.77 &  -0.54 \\
St.~dev.  &  12.29 &  14.09 &  15.61 &  14.37 &  14.82 &  11.41 &  12.41 &  13.57 \\
RRMSE          &  0.042 &  0.052 &  0.056 &  0.049 &  0.048 &  0.041 &  0.044 &  0.047 \\
\addlinespace

\textit{AA+} \\
Mean       &  -4.61 &   0.42 &   3.84 &  -1.62 &  -1.49 &  -7.49 &  -3.86 &  -2.12 \\
St.~dev.   &  18.89 &  21.17 &  20.83 &  18.05 &  17.35 &  18.26 &  24.98 &  19.93 \\
RRMSE          &  0.052 &  0.061 &  0.060 &  0.048 &  0.043 &  0.041 &  0.051 &  0.051 \\
\addlinespace

\textit{AA} \\
Mean       &  -6.43 &  -1.30 &   2.02 &  -5.64 &  -3.63 &  -3.81 &  10.94 &  -1.12 \\
St.~dev.   &  25.64 &  26.03 &  22.97 &  18.92 &  19.09 &  23.85 &  33.14 &  24.23 \\
RRMSE          &  0.069 &  0.069 &  0.059 &  0.045 &  0.044 &  0.046 &  0.074 &  0.058 \\
\addlinespace

\textit{AA-} \\
Mean       &  -4.44 &  -2.99 &   1.22 &  -5.89 &  -3.62 & -11.70 &  -4.25 &  -4.52 \\
St.~dev.   &  34.99 &  23.44 &  23.20 &  23.03 &  22.78 &  18.61 &  20.73 &  23.83 \\
RRMSE          &  0.076 &  0.046 &  0.038 &  0.037 &  0.036 &  0.035 &  0.031 &  0.043 \\
\bottomrule
\end{tabular*}

\vspace{0.25em}
\raggedright\footnotesize
Notes:Mean and standard deviation are computed from yield pricing errors and converted to basis points by multiplying by 10{,}000.
\end{threeparttable}
\end{table*}

\begin{table*}[!t]
\centering
\caption{Pricing error statistics under the UKF model (single-regime, up to 10Y).
The pricing error is $\epsilon_t = y_t - \hat{y}_t$, where $\hat{y}_t$ is the model-implied yield and $y_t$ is the observed yield.
Mean and standard deviation are reported in basis points (bp), i.e., $10{,}000 \times \epsilon_t$.
RRMSE is unitless.}
\label{tab:ukf_pricing_error_stats}
\footnotesize
\setlength{\tabcolsep}{4pt}
\renewcommand{\arraystretch}{1.08}

\begin{threeparttable}
\begin{tabular*}{\linewidth}{@{\extracolsep{\fill}} lrrrrrrrr}
\toprule
& $\epsilon_{1}$ & $\epsilon_{2}$ & $\epsilon_{3}$ & $\epsilon_{4}$ & $\epsilon_{5}$ & $\epsilon_{7}$ & $\epsilon_{10}$ & Average \\
\midrule

\textit{CGB} \\
Mean      &  8.92 & -0.74 &  2.48 &  1.67 &  2.64 & -4.57 & 15.54 &  3.71 \\
St.~dev.  & 17.07 & 13.49 & 14.62 & 15.11 & 13.83 & 13.91 & 14.33 & 14.62 \\
RRMSE     & 0.086 & 0.051 & 0.051 & 0.047 & 0.046 & 0.042 & 0.068 & 0.056 \\
\addlinespace

\textit{CDB} \\
Mean      &  6.57 & -4.13 & -1.15 & -1.85 & -0.62 & -7.10 & 14.18 &  0.84 \\
St.~dev.  & 15.04 & 11.91 & 11.85 & 11.96 & 11.63 & 11.51 & 13.13 & 12.43 \\
RRMSE     & 0.066 & 0.043 & 0.039 & 0.036 & 0.036 & 0.037 & 0.053 & 0.044 \\
\addlinespace

\textit{AAA} \\
Mean      & -8.71 & -2.45 &  3.05 &  0.37 &  2.85 &  2.46 &  7.88 &  0.78 \\
St.~dev.  & 19.21 & 19.73 & 20.18 & 18.25 & 18.10 & 14.32 & 13.99 & 17.68 \\
RRMSE     & 0.075 & 0.078 & 0.078 & 0.071 & 0.067 & 0.053 & 0.048 & 0.067 \\
\addlinespace

\textit{AA+} \\
Mean      & -10.85 & -4.13 &  1.07 & -2.66 & -0.75 & -2.95 &  6.75 & -1.93 \\
St.~dev.  & 27.17 & 26.47 & 24.80 & 21.11 & 19.46 & 18.82 & 24.12 & 23.14 \\
RRMSE     & 0.084 & 0.083 & 0.076 & 0.064 & 0.055 & 0.046 & 0.058 & 0.067 \\
\addlinespace

\textit{AA} \\
Mean      & -11.35 & -4.66 & -0.01 & -6.69 & -3.86 & -2.58 & 14.77 & -2.05 \\
St.~dev.  & 34.22 & 30.70 & 25.90 & 20.51 & 19.96 & 23.56 & 32.05 & 26.70 \\
RRMSE     & 0.097 & 0.084 & 0.069 & 0.053 & 0.049 & 0.048 & 0.078 & 0.068 \\
\addlinespace

\textit{AA-} \\
Mean      & -0.33 &  0.25 &  3.13 & -5.58 & -4.80 & -14.85 & -7.12 & -4.19 \\
St.~dev.  & 31.98 & 22.85 & 22.90 & 23.91 & 25.09 & 21.89 & 27.35 & 25.14 \\
RRMSE     & 0.074 & 0.050 & 0.040 & 0.040 & 0.041 & 0.041 & 0.042 & 0.047 \\
\addlinespace

\bottomrule
\end{tabular*}

\vspace{0.25em}
\raggedright\footnotesize
Notes:Mean and standard deviation are computed from yield pricing errors and converted to basis points by multiplying by 10{,}000.
\end{threeparttable}
\end{table*}

\subsection{Parameter Analysis}\label{subsec:parameters}
Table~\ref{tab:rsukf_params} reports regime-dependent estimates for the RS--UKF model with two rate regimes
$s_t^{r}\in\{H,L\}$ and two credit regimes $s_t^{c}\in\{E,C\}$.
The rate factors $(X_{1},X_{2},X_{3})$ evolve conditional on $s_t^{r}$, while the credit factor $X_{4}$ evolves conditional on $s_t^{c}$.
Parameters $(\kappa,\theta,\alpha,\beta)$ govern the physical dynamics used for filtering, while the market-price-of-risk terms $\lambda$
govern the mapping into risk-neutral dynamics used for valuation.
This $P\!\to\!Q$ mapping is discipline-specific and is given by \eqref{eq:rn-map-rsg}, which clarifies the role of $\lambda$ in Table~\ref{tab:rsukf_params}.
To keep the main text focused on economically interpretable implications, we emphasize three features:
(i) regime differences in discounting dynamics in the rate block;
(ii) regime differences in residual credit dynamics in the credit block; and
(iii) the state-dependent pass-through coefficients $c(s_t^{c}\!\mid s_t^{r})$ that govern how rate-level movements enter the priced
migration/default channel under the risk-neutral measure.

The rate-block estimates support a two-state discounting environment in which the conditional behavior of the level component
$(X_1,X_2)$ and the policy-bank spread component $X_3$ differs across $H$ and $L$ (Table~\ref{tab:rsukf_params}).
Because $(X_1,X_2)$ enter the discount curve directly, regime differences in persistence and long-run configuration translate into
distinct term-structure responses to shocks: one regime features a faster-adjusting level component with a different long-run anchor,
while the other is more persistent, allowing low-frequency shifts in the benchmark curve to be captured without imposing a single stationary law.
The spread factor $X_3$ affects corporate yields through discounting but not through default compensation; its regime dependence therefore
captures state-dependent variation in the intermediate benchmark (CDB relative to CGB) and helps separate funding/discounting movements from
credit compensation. These regime contrasts are economically material and precisely estimated, as reflected in the reported standard errors in
Table~\ref{tab:rsukf_params}. Overall, the rate-block parameter patterns are consistent with the regime diagnostics in
Figure~\ref{fig:rate_credit_joint_regime}, which highlight persistent swings in the rate-regime belief.

The credit factor $X_4$ is filtered conditional on the fitted rate block and is designed to capture residual credit conditions beyond discounting.
The estimates imply that the priced credit environment can switch between an expansion state $E$ and a contraction state $C$ with different
typical levels and conditional dispersion of $X_4$ (Table~\ref{tab:rsukf_params}).
Importantly, the labels $(E,C)$ are assigned for interpretability using the filtered credit-regime probabilities and the observed behavior of
corporate yields/spreads, rather than by mechanically ordering a single parameter.
Moreover, since valuation is governed by risk-neutral dynamics, the economic implications for pricing are determined by the joint mapping
from physical parameters into risk-neutral parameters via $\lambda$ in \eqref{eq:rn-map-rsg} (rather than by physical parameters in isolation),
and the corresponding risk-neutral contrasts are the relevant objects for valuation.

A central pricing object is the state-dependent pass-through matrix $\mathbf{C}=\{c(s_t^{c}\!\mid s_t^{r})\}$.
Under the joint-regime valuation route, $\mathbf{C}$ governs how rate-level movements enter the risk-neutral migration/default driver
(see \eqref{eq:mu_driver_joint_state}): it measures the \emph{within-joint-regime marginal loading} of priced credit risk on the
quasi-risk-free level component $(X_1+X_2)$, conditional on the credit factor.
In contrast, the spread factor $X_3$ affects corporate yields only through discounting and does not enter the default-compensation channel.
We emphasize that $\mathbf{C}$ is a reduced-form pricing object under the risk-neutral measure: it summarizes regime-contingent co-movement
between priced credit risk and rate factors, rather than a structural causal effect of policy.

\[
\mathbf{C}
=
\begin{pmatrix}
c(E\mid H) & c(E\mid L)\\
c(C\mid H) & c(C\mid L)
\end{pmatrix}
=
\begin{pmatrix}
3.1592 & -0.2821\\
-0.1052 & 0.6013
\end{pmatrix}.
\]

The estimated coefficients in $\mathbf{C}$ display pronounced state dependence in sign and magnitude (Table~\ref{tab:rsukf_params}).
In the expansion credit regime $E$, the pass-through changes sign across rate regimes, implying that the marginal pricing impact of a given
rate-level movement on default compensation can differ depending on the prevailing rate environment. In the contraction regime $C$,
the pass-through remains regime dependent, with the mapping being comparatively muted under one rate regime and stronger under the other.
These sign patterns provide a disciplined way to separate discounting effects from default-compensation effects across joint regimes:
rate-level movements always affect discounting, but they affect default compensation only through $\mathbf{C}$, while $X_3$ operates purely
through the discount curve.

Finally, negative entries in $\mathbf{C}$ do not imply negative default intensities. Admissibility is ensured by the construction of
migration/default intensities from the common driver: when mapping the driver into intensities we impose non-negativity (and, in computation,
apply a small lower bound to guard against numerical violations), so priced intensities remain well-defined in all regimes.
Figure~\ref{fig:rate_credit_joint_regime} provides a visual check that the inferred joint regimes correspond to persistent episodes,
supporting identification of regime-contingent loadings.

Table~\ref{tab:rsukf_params} also reports $(\nu_1,\ldots,\nu_5)$, which are risk-neutral adjustment terms applied to the empirical default
intensity (``vega'' adjustments) in the migration/default block. They provide additional flexibility for matching observed credit spreads
while preserving the model-implied cross-maturity consistency of the hazard curve.

\begin{table*}[!t]
\centering
\caption{Regime-switching parameter estimates (RS--UKF, $K=2$).
The first three factors ($X_1$--$X_3$) follow the rate regimes $\{H,L\}$, while the credit factor ($X_4$) follows the credit regimes $\{E,S\}$.
Robust standard errors are reported in parentheses.}
\label{tab:rsukf_params}
\footnotesize
\setlength{\tabcolsep}{4pt}
\renewcommand{\arraystretch}{1.10}

\begin{threeparttable}

\begin{tabular*}{\linewidth}{@{\extracolsep{\fill}} llccccc}
\toprule
Factor & Regime & $\kappa$ & $\theta$ & $\alpha$ & $\beta$ & $\lambda$ \\
\midrule

\multirow{2}{*}{$X_{1}$} & H &
\begin{tabular}{@{}c@{}} 1.489 \\ \((2.89\times 10^{-6})\) \end{tabular} &
\begin{tabular}{@{}c@{}} 0.006250 \\ (0.000155) \end{tabular} &
\begin{tabular}{@{}c@{}} 0.000133 \\ \((2.42\times 10^{-5})\) \end{tabular} &
\begin{tabular}{@{}c@{}} 0.000126 \\ \((3.21\times 10^{-5})\) \end{tabular} &
\begin{tabular}{@{}c@{}} -41.098 \\ \((0.01130)\) \end{tabular} \\
 & L &
\begin{tabular}{@{}c@{}} 0.5349 \\ \((7.95\times 10^{-5})\) \end{tabular} &
\begin{tabular}{@{}c@{}} 0.009634 \\ \((7.21\times 10^{-5})\) \end{tabular} &
\begin{tabular}{@{}c@{}} 0.000126 \\ \((3.81\times 10^{-6})\) \end{tabular} &
\begin{tabular}{@{}c@{}} 0.000123 \\ \((4.24\times 10^{-5})\) \end{tabular} &
\begin{tabular}{@{}c@{}} -32.901 \\ \((0.0120)\) \end{tabular} \\
\addlinespace

\multirow{2}{*}{$X_{2}$} & H &
\begin{tabular}{@{}c@{}} 0.02852 \\ (0.002592) \end{tabular} &
\begin{tabular}{@{}c@{}} 0.03341 \\ (0.01868) \end{tabular} &
\begin{tabular}{@{}c@{}} $8.78\times 10^{-5}$ \\ (0.003436) \end{tabular} &
\begin{tabular}{@{}c@{}} $3.67\times 10^{-4}$ \\ ($1.20\times 10^{-4}$) \end{tabular} &
\begin{tabular}{@{}c@{}} -15.4113 \\ (0.2500) \end{tabular} \\
& L &
\begin{tabular}{@{}c@{}} 0.05949 \\ (0.004193) \end{tabular} &
\begin{tabular}{@{}c@{}} 0.01382 \\ (0.001960) \end{tabular} &
\begin{tabular}{@{}c@{}} $2.05\times 10^{-5}$ \\ (0.00044) \end{tabular} &
\begin{tabular}{@{}c@{}} 0.003206 \\ (0.037966) \end{tabular} &
\begin{tabular}{@{}c@{}} -9.1011 \\ (1.5096) \end{tabular} \\

\multirow{2}{*}{$X_{3}$} & H &
\begin{tabular}{@{}c@{}} 1.0272 \\ (0.02231) \end{tabular} &
\begin{tabular}{@{}c@{}} $0.00017$ \\ (0.00020) \end{tabular} &
\begin{tabular}{@{}c@{}} $3.00\times 10^{-5}$ \\ ($5.54\times 10^{-6}$) \end{tabular} &
\begin{tabular}{@{}c@{}} 0.004623 \\ ($1.93\times 10^{-5}$) \end{tabular} &
\begin{tabular}{@{}c@{}} -0.1537 \\ (0.00740) \end{tabular} \\
 & L &
\begin{tabular}{@{}c@{}} 0.1740 \\ ($7.53\times 10^{-5}$) \end{tabular} &
\begin{tabular}{@{}c@{}} $0.00014$ \\ (0.00025) \end{tabular} &
\begin{tabular}{@{}c@{}} $1.00\times 10^{-4}$ \\ (0.000060) \end{tabular} &
\begin{tabular}{@{}c@{}} 0.02824 \\ (0.00800) \end{tabular} &
\begin{tabular}{@{}c@{}} -0.03754 \\ (0.0134) \end{tabular} \\
\addlinespace

\multirow{2}{*}{$X_{4}$} & E &
\begin{tabular}{@{}c@{}} 0.4824 \\ (0.0340) \end{tabular} &
\begin{tabular}{@{}c@{}} 0.0862 \\ (0.3848) \end{tabular} &
\begin{tabular}{@{}c@{}} 0.1992 \\ (0.0657) \end{tabular} &
\begin{tabular}{@{}c@{}} 0.1321 \\ (0.0143) \end{tabular} &
\begin{tabular}{@{}c@{}} -3.2500 \\ (0.3576) \end{tabular} \\
 & C &
\begin{tabular}{@{}c@{}} 0.4282 \\ (0.0510) \end{tabular} &
\begin{tabular}{@{}c@{}} 0.0504 \\ (0.3712) \end{tabular} &
\begin{tabular}{@{}c@{}} 0.1276 \\ (0.0923) \end{tabular} &
\begin{tabular}{@{}c@{}} 0.0882 \\ (0.0130) \end{tabular} &
\begin{tabular}{@{}c@{}} -2.7564 \\ (0.3381) \end{tabular} \\

\textit{Other parameters}
\end{tabular*}

\vspace{0.55em}

\begin{tabular*}{\linewidth}{@{\extracolsep{\fill}} lccccc}
\toprule
\multicolumn{6}{l}{}\\
\addlinespace[-0.25em]
& $c(E\mid H)$ & $c(E\mid L)$ & $c(S\mid H)$ & $c(S\mid L)$ &  \\
& \begin{tabular}{@{}c@{}} 3.1592 \\ (0.5653) \end{tabular}
& \begin{tabular}{@{}c@{}} -0.2821 \\ (0.0924) \end{tabular}
& \begin{tabular}{@{}c@{}} -0.1052 \\ (0.0319) \end{tabular}
& \begin{tabular}{@{}c@{}} 0.6013 \\ (0.1304) \end{tabular}
&  \\

\addlinespace[0.2em]
& $\Delta\nu_{1}$ & $\Delta\nu_{2}$ & $\Delta\nu_{3}$ & $\Delta\nu_{4}$ & $\Delta\nu_{5}$ \\
\addlinespace[-0.25em]
& \begin{tabular}{@{}c@{}}
0.000220 \\
(0.000145)
\end{tabular}
& \begin{tabular}{@{}c@{}}
$5.36\times 10^{-6}$ \\
(0.000465)
\end{tabular}
& \begin{tabular}{@{}c@{}}
0.000105 \\
(0.000310)
\end{tabular}
& \begin{tabular}{@{}c@{}}
0.000257 \\
(0.000222)
\end{tabular}
& \begin{tabular}{@{}c@{}}
0.016223 \\
(0.000966)\end{tabular} \\

\bottomrule
\end{tabular*}

\vspace{0.25em}
\raggedright\footnotesize
Notes: $H/L$ denote the two rate regimes for $X_1$--$X_3$, and $E/S$ denote the two credit regimes for $X_4$ (labels can be swapped if desired).
\end{threeparttable}
\end{table*}

\subsection{Joint regime dynamics}\label{subsec:regimes}

Figure~\ref{fig:rate_credit_joint_regime} provides a compact view of the joint rate--credit regime
dynamics implied by the RS--UKF state-space system. Panel~A plots representative rate-block
yields (CGB 3Y and CDB 3Y), Panel~B plots representative corporate yields (3Y) by rating, and
Panel~C reports the filtered regime probabilities. Throughout, the primary object is the pair of beliefs
\(\pi_t^{r}\) and \(\pi_t^{c}\)---not a hard regime label---and the figure should be interpreted in terms of the
evolution and persistence of these probabilities. For visual guidance only, the background shading
in Panels~A--B converts probabilities into a binary display: solid shading marks dates for which the
low-rate probability is above a 0.5 threshold, and hatched shading analogously marks dates
for which the contracted credit probability is above the same threshold; the underlying probability paths
are shown in Panel~C.%
\footnote{In the model, the credit regime is labeled as expansion--contraction; the figure legend may use
``contracted'' to refer to the contraction state. The interpretation is invariant to relabeling.}

Two empirical regularities emerge that map cleanly into the paper's identification and valuation design.
First, movements in \(\pi_t^{r}\) are relatively persistent and align with low-frequency shifts in the level and
dynamics of the benchmark curve: episodes with \(\pi_t^{r}\) close to one coincide with a sustained decline
in both CGB and CDB yields, and tend to be associated with a more stable intermediate benchmark
(CDB relative to CGB). This is the behavior the rate block is designed to capture, as the rate regime is
filtered from rate-block information only and is governed by a continuous-time Markov generator mapped
to weekly transition probabilities, which mechanically induces persistence.%

Second, \(\pi_t^{c}\) is more episodic and displays a conditional character: increases in the probability of the
contracted (credit-contraction) state tend to occur when corporate yields (and, equivalently, corporate
spreads relative to the CDB discount curve) rise in a sustained fashion and when cross-rating dispersion
becomes more pronounced, with stronger visibility for lower ratings. Importantly, this behavior is not
imposed by construction. Under our block-recursive learning restriction, the rate regime is inferred from
rate-block information, while the credit regime is inferred from credit-block information conditional on
a finite-dimensional posterior summary of the first-stage rate filter. As a consequence, the model admits
both synchronized configurations (joint increases in \(\pi_t^{r}\) and \(\pi_t^{c}\)) and asynchronous configurations
in which the rate environment changes while credit remains largely in the expansion state (or vice versa).

These probability objects play a direct role in valuation. Because defaultable-bond pricing is nonlinear in
states through an exponential(-affine) kernel, state-contingent prices are computed regime-by-regime and
then mixed at the price level using filtered regime probabilities, rather than averaging drivers or
states prior to valuation. Under block-recursive learning, the joint weights are constructed without joint
filtering by combining the first-stage rate-regime posterior with the second-stage credit-regime posterior
conditional on the rate regime. This price-level mixing step is essential for avoiding Jensen-type
distortions and clarifies why Figure~\ref{fig:rate_credit_joint_regime} should be read probabilistically.

Finally, any apparent timing differences between the reduced-form Gaussian HMM classification used as a
diagnostic elsewhere and the structural filtered probabilities here are expected and economically
uninformative on their own. The HMM provides a discrete, likelihood-based segmentation of moments,
whereas \(\pi_t^{r}\) and \(\pi_t^{c}\) are structural beliefs produced by a regime-switching state-space filter with
(i) CTMC-implied persistence and (ii) nonlinear measurement through pricing and price-to-yield
transformation. It is therefore natural for structural probabilities to adjust gradually and for thresholded
labels (used only for shading) to lag incipient probability movements. Consistent with this interpretation,
regime boundaries in Panels~A--B should not be treated as event identifiers; the empirical content lies in
the persistence of \((\pi_t^{r},\pi_t^{c})\) and in the extent of synchronization between rate and credit states,
which directly matches the model's one-way learning design and state-dependent pricing mechanism.

\begin{figure}[!t]
\centering
\includegraphics[width=\linewidth,height=0.86\textheight,keepaspectratio]{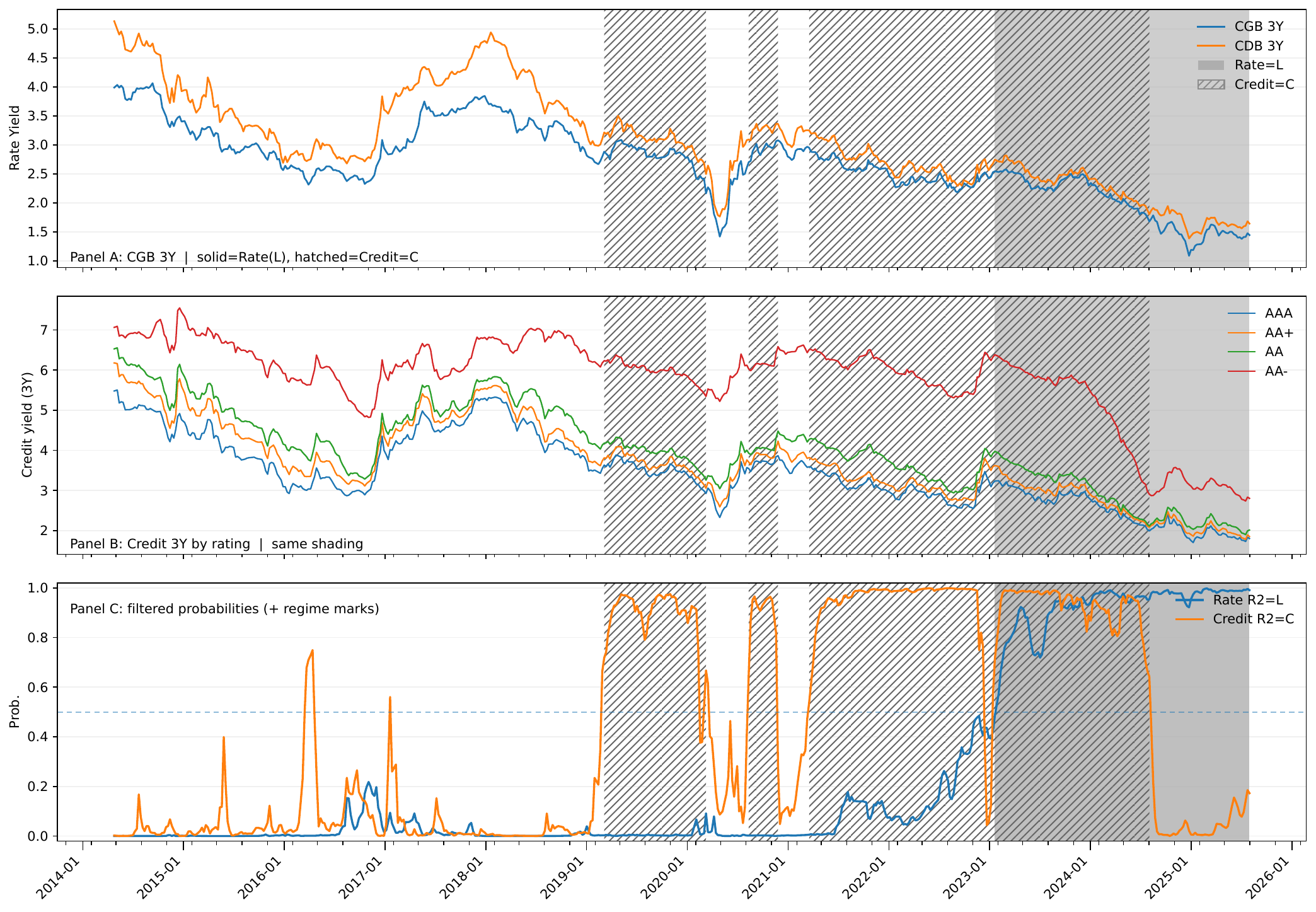}
\caption{Joint rate--credit regimes and filtered probabilities.
Panel~A plots the 3-year CGB yield and the 3-year CDB yield.
Panel~B plots 3-year corporate yields by rating (AAA, AA$+$, AA, AA$-$).
Panel~C reports the filtered regime probabilities \(\pi_t^{r}\equiv \Pr(s_t^{r}=L\mid I_t)\) and
\(\pi_t^{c}\equiv \Pr(s_t^{c}=C\mid I_t)\), where \(L\) denotes the low-rate state and \(C\) the credit-contraction
(contracted) state.
For readability, background shading in Panels~A--B provides a thresholded visualization of the joint
configuration: solid shading indicates dates with \(\pi_t^{r}>0.5\), and hatched shading indicates dates with
\(\pi_t^{c}>0.5\). The underlying probabilities are shown in Panel~C and should be treated as the primary
object of interest.}
\label{fig:rate_credit_joint_regime}
\vspace{-0.3em}
\end{figure}

\subsection{Spread decomposition}\label{subsec:decomp_spread}

This section recasts the rate--credit information used above into a decomposition that is
(i) algebraically exact and (ii) economically interpretable. Let \(y_t^{\mathrm{CGB}}(\tau)\) denote the
zero-coupon China government bond (CGB) yield at maturity \(\tau\), \(y_t^{\mathrm{CDB}}(\tau)\) the matched-maturity
policy-bank (CDB) yield, and \(y_{r,t}^{\mathrm{corp}}(\tau)\) the matched-maturity corporate zero-coupon yield for
rating bucket \(r\in\{\mathrm{AAA},\mathrm{AA{+}},\mathrm{AA}\}\).\footnote{Issuance and secondary-market activity in China’s onshore
credit market are heavily concentrated in the AAA/AA+/AA segment, while AA$-$ and below form a relatively thin
segment and are often viewed as a transition toward a speculative-grade universe. Our spread decomposition therefore
focuses on the most liquid investment-grade buckets to mitigate mechanical noise from illiquid tails.}

We work with spreads rather than raw yields to separate risk premia from the common discount-rate component.
Specifically, we define
\begin{align}
s_t^{\mathrm{CDB}}(\tau) &\equiv y_t^{\mathrm{CDB}}(\tau)-y_t^{\mathrm{CGB}}(\tau), \label{eq:cdb_wedge_def}\\
s_{r,t}^{\mathrm{corp}}(\tau) &\equiv y_{r,t}^{\mathrm{corp}}(\tau)-y_t^{\mathrm{CDB}}(\tau). \label{eq:corp_spread_def}
\end{align}
By construction, corporate yields satisfy the additive identity
\begin{equation}
y_{r,t}^{\mathrm{corp}}(\tau)
=
\underbrace{y_t^{\mathrm{CGB}}(\tau)}_{\text{sovereign discounting}}
+
\underbrace{s_t^{\mathrm{CDB}}(\tau)}_{\text{policy-bank spread}}
+
\underbrace{s_{r,t}^{\mathrm{corp}}(\tau)}_{\text{corporate compensation over CDB}}.
\label{eq:yield_decomp_identity}
\end{equation}
Equation~\eqref{eq:yield_decomp_identity} clarifies the empirical role of the two-block regime structure. The rate block
summarizes persistent shifts in discounting and funding conditions and is naturally reflected in the policy-bank spread
\(s_t^{\mathrm{CDB}}(\tau)\). The credit block isolates the residual compensation embedded in corporate yields beyond the
policy-bank benchmark, summarized by \(s_{r,t}^{\mathrm{corp}}(\tau)\). This organization follows the logic of separating
discounting from credit-related components in reduced-form term-structure decompositions, adapted to the China setting
where CDB provides a natural intermediate benchmark.

To summarize cross-maturity differences across regimes without imposing a hard classification, we report
probability-weighted regime means. For any series \(x_t(\tau)\), define
\begin{equation}
\bar{x}^{A}(\tau)
\equiv
\frac{\sum_{t=1}^{T} w_t^{A}\, x_t(\tau)}{\sum_{t=1}^{T} w_t^{A}},
\qquad
w_t^{A}\in[0,1],
\label{eq:prob_weighted_mean}
\end{equation}
where the weights are filtered regime probabilities. For the rate block, \(w_t^{L}=\pi_t^{r}\) and \(w_t^{H}=1-\pi_t^{r}\);
for the credit block, \(w_t^{S}=\pi_t^{c}\) and \(w_t^{E}=1-\pi_t^{c}\). The summaries in \eqref{eq:prob_weighted_mean}
are invariant to state relabeling as long as probabilities are mapped consistently.

\begin{figure}[!t]
\centering
\includegraphics[width=0.98\linewidth]{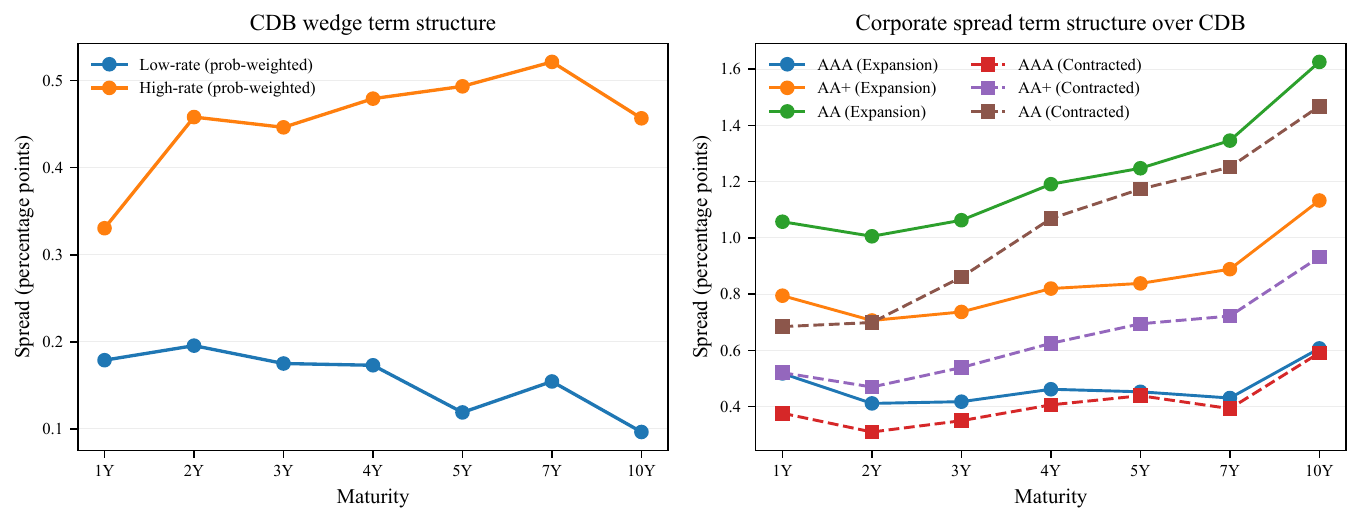}
\caption{Regime-conditional spread term structures (probability-weighted).
Panel~A plots the policy-bank spread $s_t^{\mathrm{CDB}}(\tau)=y_t^{\mathrm{CDB}}(\tau)-y_t^{\mathrm{CGB}}(\tau)$ under the two \emph{rate} regimes,
computed as probability-weighted means using $\pi_t^{r}$ and $1-\pi_t^{r}$.
Panel~B plots the corporate spread over CDB, $s_{r,t}^{\mathrm{corp}}(\tau)=y_{r,t}^{\mathrm{corp}}(\tau)-y_t^{\mathrm{CDB}}(\tau)$, for $r\in\{\mathrm{AAA},\mathrm{AA{+}},\mathrm{AA}\}$ under the two \emph{credit} regimes,
computed as probability-weighted means using $\pi_t^{c}$ and $1-\pi_t^{c}$.
All objects are in percentage points.}
\label{fig:fig55a_term_struct}
\vspace{-0.3em}
\end{figure}

Figure~\ref{fig:fig55a_term_struct} reports the implied regime-conditional term structures. Panel~A shows that the policy-bank
spread \(s_t^{\mathrm{CDB}}(\tau)\) is systematically different across rate states: the probability-weighted spread is materially
lower in the low-rate state and higher in the high-rate state, with differences that are visible across maturities. Panel~B
summarizes \(s_{r,t}^{\mathrm{corp}}(\tau)\), the corporate compensation over the CDB benchmark. Two robust features stand out.
First, the cross-sectional ordering \(s_{\mathrm{AA},t}^{\mathrm{corp}}(\tau) > s_{\mathrm{AA+},t}^{\mathrm{corp}}(\tau) >
s_{\mathrm{AAA},t}^{\mathrm{corp}}(\tau)\) holds at each maturity under both credit states, indicating that rating buckets retain
distinct compensation even after netting out the policy-bank component. Second, the contracted-credit state is associated
with higher corporate compensation, with the difference becoming more pronounced toward the long end, consistent with
credit conditions affecting longer-horizon migration/default compensation more strongly than short-maturity pricing.

Figure~\ref{fig:fig55b_timeseries} complements the cross-sectional evidence with a time-series view at \(\tau\in\{3,5,10\}\) years.
The top three panels plot the policy-bank spread \(s_t^{\mathrm{CDB}}(\tau)\) together with the AA corporate spread over CDB
\(s_{\mathrm{AA},t}^{\mathrm{corp}}(\tau)\) at each maturity, while the bottom panel reports the filtered probabilities \(\pi_t^{r}\)
and \(\pi_t^{c}\). The figure highlights two points relevant for interpretation. First, the policy-bank spread varies in the expected
direction with the rate-regime probability: the spread tends to be lower when \(\pi_t^{r}\) assigns high probability to the low-rate
state, and higher when the low-rate probability is low. Second, the corporate compensation component co-moves more closely
with the credit-regime probability: sustained elevations in \(\pi_t^{c}\) are associated with persistently higher corporate spreads
over the CDB benchmark, particularly at longer maturities.

Finally, the regime probabilities are not intended for event dating. Because they are filtered beliefs that combine a persistent
transition law with nonlinear measurement, they may adjust gradually and need not coincide one-for-one with short-lived spread
spikes. Accordingly, the empirical content of Figure~\ref{fig:fig55b_timeseries} lies in the persistence of \((\pi_t^{r},\pi_t^{c})\)
and in how the decomposition in \eqref{eq:yield_decomp_identity} attributes movements in corporate yields to either the policy-bank
spread (funding) or corporate compensation over CDB (credit-risk channel).

\begin{figure}[!t]
\centering
\includegraphics[width=0.98\linewidth]{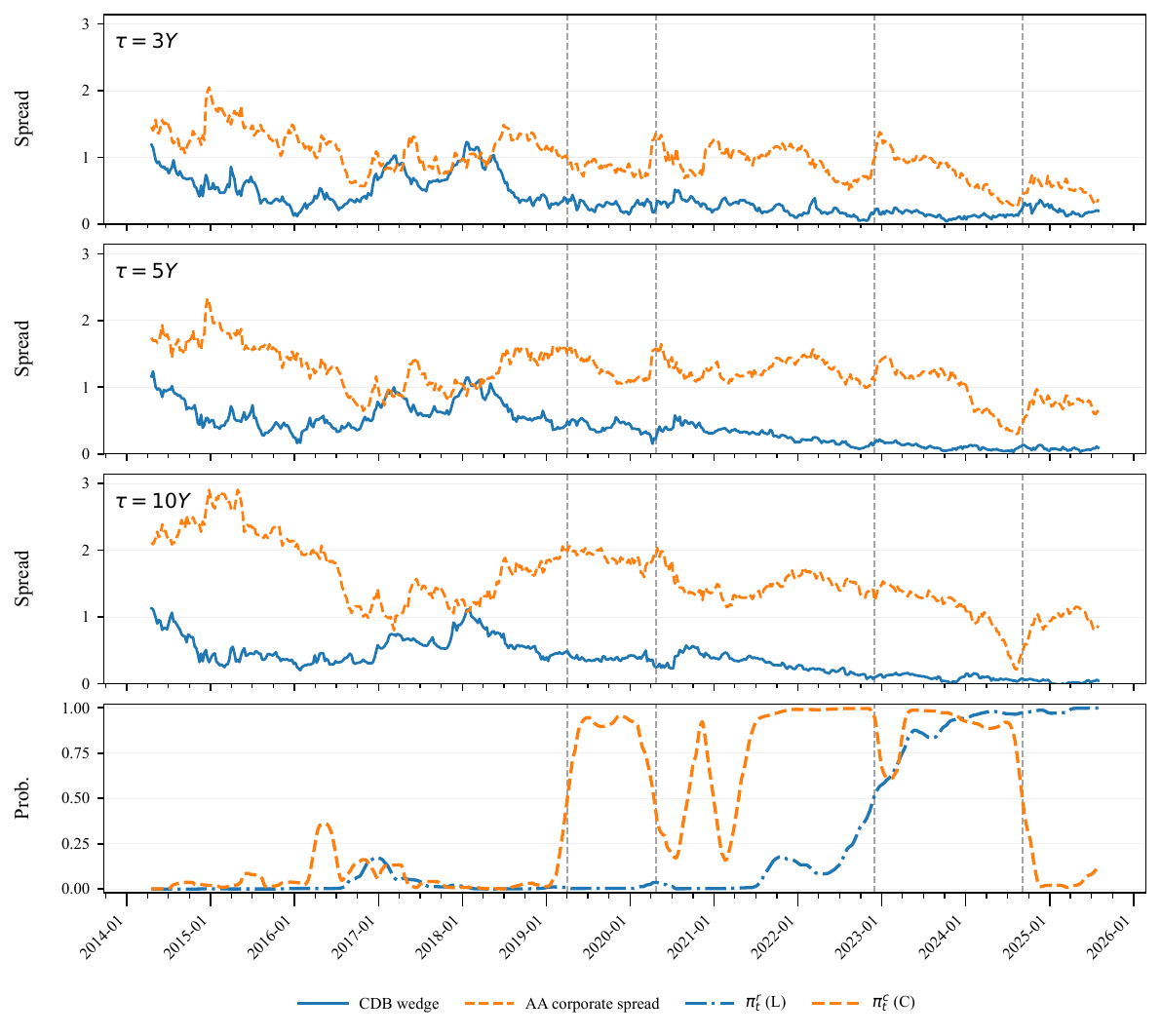}
\caption{Time-series spread decomposition and filtered regime beliefs at key maturities.
For each maturity $\tau\in\{3,5,10\}$ years, the left axis plots the policy-bank spread
$s_t^{\mathrm{CDB}}(\tau)\equiv y_t^{\mathrm{CDB}}(\tau)-y_t^{\mathrm{CGB}}(\tau)$ and a representative corporate spread over CDB,
$s_{\mathrm{AA},t}^{\mathrm{corp}}(\tau)\equiv y_{\mathrm{AA},t}^{\mathrm{corp}}(\tau)-y_t^{\mathrm{CDB}}(\tau)$ (see \eqref{eq:cdb_wedge_def}--\eqref{eq:corp_spread_def}).
The bottom panel reports the filtered regime probabilities on the right axis,
$\pi_t^{r}=\Pr(s_t^{r}=L\mid I_t)$ (low-rate state) and $\pi_t^{c}=\Pr(s_t^{c}=C\mid I_t)$ (credit-contraction state). All spreads are in percentage points, sampled at a weekly frequency.
Vertical reference lines mark calendar time for readability only and carry no structural interpretation.}
\label{fig:fig55b_timeseries}
\vspace{-0.3em}
\end{figure}

\subsection{Characteristic-interacted factor regressions}
\label{subsec:ipca_routeB}

To complement our structural regime-switching pricing results, we provide a concise, reduced-form validation that the investment-grade corporate spread surface exhibits (i) low-dimensional common movements and (ii) systematic cross-sectional heterogeneity across rating and maturity. Specifically, we implement an IPCA-style instrumented-beta regression in which each rating--maturity node is treated as an ``asset'' whose exposure to factor innovations varies with observable characteristics. Our specification follows the instrumented-beta idea by \citet{KellyPruittSu2019}, where observable characteristics proxy for conditional factor loadings. 

Let $i=(r,\tau)$ index rating bucket $r\in\{\text{AAA},\text{AA+},\text{AA}\}$ and maturity $\tau\in\{1,2,3,5,7,10\}$ years. The dependent variable is the weekly change in the matched-maturity corporate spread,
\begin{equation}
y_{i,t} \equiv \Delta s_{r,t}(\tau).
\end{equation}
We relate $y_{i,t}$ to innovations in our model-implied factors through a characteristic-driven loading,
\begin{equation}
y_{i,t}=\beta_{i,t-1}^{\top}\Delta X_t+\varepsilon_{i,t},
\qquad 
\beta_{i,t-1}=z_{i,t-1}^{\top}\Gamma,
\label{eq:ipca_routeB}
\end{equation}
so that $y_{i,t}=z_{i,t-1}^{\top}\Gamma\,\Delta X_t+\varepsilon_{i,t}$.
The characteristics $z_{i,t-1}$ include a constant, a quadratic maturity profile $(\tau,\tau^2)$, rating dummies (AA+ and AA; AAA as the base group), and the standardized lagged spread level at the node (capturing state dependence). Estimation is pooled OLS on the interaction regressors $(z_{i,t-1}\otimes \Delta X_t)$ with standard errors clustered by week. We report two specifications: a credit-focused specification using $\Delta X_t=(\Delta x_{4,t},\Delta u_t)$ (Spec-1) and a full-factor specification adding rate-side innovations (Spec-2).

Table~\ref{tab:ipca_routeB_summary} summarizes the key findings. Panel A indicates that the credit-focused innovations already explain a meaningful share of weekly spread changes ($R^2=0.227$), while allowing the full set of innovations substantially raises explanatory power ($R^2=0.748$). Panel B shows that the implied exposure to the credit factor, $\beta_{x4}$, is monotonically ordered by rating (AA $>$ AA+ $>$ AAA) at all maturities, consistent with lower credit quality being more sensitive to aggregate credit shocks. In addition, $\beta_{x4}$ declines with maturity within each rating bucket, suggesting that short-to-intermediate maturities carry stronger cyclical credit sensitivity in the investment-grade segment. Finally, Panel C provides direct evidence of characteristic-driven heterogeneity: the AA+ and AA dummies load positively into $\Gamma_{x4}$, and the lagged spread level increases the loading on $x_4$ (state dependence), while the lagged spread loads negatively into $\Gamma_u$, consistent with $u_t$ capturing an aggregate credit-conditioning component that intensifies in contracted states. Overall, this IPCA-style exercise corroborates the factor-based structure embedded in our regime-switching credit pricing model and the central role of $(x_4,u)$ in organizing the investment-grade spread surface.

\begin{table*}[!t]
\centering
\caption{Supplementary IPCA-style evidence (Route B) for the investment-grade corporate spread surface (AAA/AA+/AA; maturities 1--10Y). 
Panel A reports pooled $R^2$ under the credit-only specification (Spec-1: $\Delta x_4,\Delta u$) and the full-factor specification (Spec-2: $\Delta x_1,\Delta x_2,\Delta x_3,\Delta x_4,\Delta u$). 
Panel B reports the average implied exposure to the credit factor, $\beta_{x4}$, by rating and maturity (Spec-1). 
Panel C reports selected entries of the feature-to-beta mapping $\beta_{i,t-1}=z_{i,t-1}'\Gamma$ as coefficient (t-stat) under Spec-1.}
\label{tab:ipca_routeB_summary}
\footnotesize
\setlength{\tabcolsep}{4pt}
\renewcommand{\arraystretch}{1.08}

\begin{threeparttable}
\begin{tabular*}{\linewidth}{@{\extracolsep{\fill}} lrrrrrrr}
\toprule
& 1Y & 2Y & 3Y & 5Y & 7Y & 10Y & Average \\
\midrule

\multicolumn{8}{l}{\textit{Panel A: Model fit (pooled $R^2$; $N_{\text{weeks}}=549$, $N_{\text{nodes}}=18$, $N_{\text{obs}}=9{,}882$)}} \\
Spec-1: $\Delta x_4,\Delta u$ 
& \multicolumn{6}{c}{} & 0.227 \\
Spec-2: $\Delta x_1,\Delta x_2,\Delta x_3,\Delta x_4,\Delta u$ 
& \multicolumn{6}{c}{} & 0.748 \\
\addlinespace
\midrule
\multicolumn{8}{l}{\textit{Panel B: Average implied exposure to the credit factor $\beta_{x4}$ (Spec-1)}} \\
AAA  & 0.585 & 0.589 & 0.561 & 0.452 & 0.260 & -0.186 & 0.377 \\
AA+  & 0.809 & 0.817 & 0.794 & 0.698 & 0.520 &  0.082 & 0.620 \\
AA   & 0.966 & 0.983 & 0.970 & 0.896 & 0.728 &  0.295 & 0.806 \\
\addlinespace
\midrule
\multicolumn{8}{l}{\textit{Panel C: Selected $\Gamma$ entries in $\beta_{i,t-1}=z_{i,t-1}'\Gamma$ (Spec-1), coefficient (t-stat)}} \\
Lagged spread (std.) $\rightarrow \Gamma_{x4}$ 
& \multicolumn{6}{c}{} & 0.214 (1.30) \\
AA+ dummy $\rightarrow \Gamma_{x4}$ 
& \multicolumn{6}{c}{} & 0.182 (1.48) \\
AA dummy $\rightarrow \Gamma_{x4}$ 
& \multicolumn{6}{c}{} & 0.296 (1.77) \\
Lagged spread (std.) $\rightarrow \Gamma_{u}$ 
& \multicolumn{6}{c}{} & $-0.364$ ($-2.17$) \\
\bottomrule
\end{tabular*}

\vspace{0.25em}
\raggedright\footnotesize
Notes: The dependent variable is the weekly change in the matched-maturity corporate spread at node $i=(r,\tau)$.
Implied exposures in Panel B are time-averages of $\beta_{x4,i,t}$ obtained from Spec-1. 
Panel C reports selected coefficients from the instrumented-beta mapping under Spec-1.
Standard errors are clustered by week.
\end{threeparttable}
\end{table*}

\section{Conclusion}
\label{sec:conclusion}
This paper studies the joint dynamics and valuation of sovereign and investment-grade corporate bond
curves in China through the lens of a continuous-time regime-switching affine framework. We develop a
sovereign--credit model in which (i) the risk-free short rate and the effective discounting environment are
driven by multiple generalized-CIR term-structure factors, (ii) corporate default risk is modelled in a
rating-based reduced-form setup, and (iii) a common intensity factor governs rating migrations and
default while transmitting sovereign conditions to corporate credit through regime-dependent pass-through.
A key methodological feature is a block-recursive estimation strategy that filters the rate block first and
then conditions the credit block on rate-block posterior summaries. For valuation, we compute
regime-conditional prices and mix them at the price level using filtered regime probabilities, which
preserves the exponential-affine structure within regimes and avoids Jensen-type distortions.

Empirically, using weekly ChinaBond zero-coupon yield curves from 2014 to 2025, we find that a
parsimonious two-regime specification provides an economically interpretable representation of
time variation in discounting and credit compensation and improves the joint fit of the CGB, CDB, and
rating-sorted corporate curves. The inferred regimes align with persistent shifts in the macro--financial
environment and generate a transparent decomposition of corporate yields into a policy-bank spread and
a corporate credit-spread component. These results highlight that modelling regime changes in the
sovereign discounting environment is essential for explaining policy-bank pricing and for understanding
how sovereign conditions transmit to high-grade credit spreads in China.

Several extensions are natural. First, incorporating time-varying risk premia and allowing regime
dependence in market prices of risk would sharpen the economic interpretation of the inferred regimes
under the pricing measure. Second, extending the credit block to explicitly accommodate liquidity
premia and segmentation across issuers (e.g., SOEs vs. private firms) could help disentangle pure default
compensation from market frictions. Finally, integrating richer, China-specific migration information and
contract-period dynamics may further improve the model’s ability to jointly price cross-section and time-series
patterns of credit spreads at longer maturities. Overall, the proposed regime-switching affine sovereign--credit
framework provides a tractable foundation for structural interpretation and quantitative pricing of onshore
credit risk under evolving macro--financial conditions.


\bibliographystyle{plainnat}   
\bibliography{reference}
\appendix

\section*{Appendix}
\addcontentsline{toc}{section}{Appendix}

\section{Affine Transform for the Generalized CIR Factor}\label{app:gcir}

This appendix records the closed-form affine transform used in our pricing and filtering blocks.  
We consider a one–factor generalized CIR (square-root) diffusion under the risk–neutral
measure $\mathbb{Q}$,
\begin{equation}\label{eq:gcir}
    dX_t=\kappa(\theta-X_t)\,dt+\sqrt{\alpha+\beta X_t}\,dW_t ,
\end{equation}
where $\kappa>0$ is the speed of mean reversion toward the level $\theta\ge 0$, and
$\alpha\ge 0$, $\beta>0$ control the variance loading.  
When $\alpha=0$ and $\beta=\sigma^2$ this reduces to the standard CIR model; if $\alpha>0$ the state
space can be chosen so that $\alpha+\beta x\ge 0$ (e.g., $x\ge -\alpha/\beta$).  
For the pure CIR case ($\alpha=0$), the classical Feller condition $2\kappa\theta\ge\beta$ prevents
the boundary $0$ from being attained.

For constants $c_1,c_2\in\mathbb{R}$, define the expectation

\begin{equation}
\mathbb{E}\left[e^{-\int_t^T c_1 X_u d u} e^{c_2 X_T} \mid \mathcal{F}_t\right]=e^{A(t, T)-B(t, T) X_t}
\end{equation}

\noindent The functions $A(t,T)$ and $B(t,T)$ solve Riccati ODEs and admit closed forms.  
Let
\[
    \gamma \coloneqq \sqrt{\kappa^{2}+2\beta c_{1}},\qquad
    \text{(real whenever $c_1>-\kappa^{2}/(2\beta)$).}
\]
Then the explicit expressions are
\begin{equation}
A(t, T)=\frac{\alpha c_1}{\beta}(T-t)-\frac{\alpha}{\beta} \frac{\left(2 c_1-2 \kappa c_2-\beta c_2^2\right)\left(e^{\gamma(T-t)}-1\right)}{2 \gamma+\left(\gamma+\kappa+\beta c_2\right)\left(e^{\gamma(T-t)}-1\right)}+\frac{2(\beta \theta+\alpha) \kappa}{\beta^2} \ln \left(\frac{2 \gamma e^{\frac{(\gamma+\kappa)(T-t)}{2}}}{2 \gamma+\left(\gamma+\kappa+\beta c_2\right)\left(e^{\gamma(T-t)}-1\right)}\right)
\end{equation}

and

\begin{equation}
B(t, T)=\frac{2 c_1\left(e^{\gamma(T-t)}-1\right)+(\gamma-\kappa) c_2 e^{\gamma(T-t)}+(\gamma+\kappa) c_2}{2 \gamma+\left(\gamma+\kappa+\beta c_2\right)\left(e^{\gamma(T-t)}-1\right)}
\end{equation}


\subsection{Solution: Martingale Argument and Riccati ODEs}

\paragraph{Step 1. Build a martingale and identify terminal conditions.}
To obtain the affine transform, we look for $A(t,T)$ and $B(t,T)$ such that
\begin{equation}\label{eq:martingale-Mt-def}
M_t \coloneqq 
\exp\!\Big(-\!\int_{0}^{t} c_{1}X_u\,du\Big)\,
\exp\!\big(A(t,T)-B(t,T)X_t\big)
\end{equation}
is a $\mathbb{Q}$–martingale. At maturity $t=T$,
\begin{equation}\label{eq:MT}
M_T=\exp\!\Big(-\!\int_{0}^{T} c_{1}X_u\,du\Big)\,
\exp\!\big(A(T,T)-B(T,T)X_T\big).
\end{equation}
Impose terminal conditions
\begin{equation}\label{eq:terminal}
A(T,T)=0,\qquad B(T,T)=c_{2},
\end{equation}
then

\begin{equation}\label{eq:MT-terminal}
M_T=\exp\!\Big(-\!\int_{0}^{T} c_{1}X_u\,du\Big)\,e^{c_{2}X_T}.
\end{equation}

If $(M_t)_{t\le T}$ is a martingale, then $\mathbb{E}[M_T\mid\mathcal{F}_t]=M_t$ and, since
$\exp(-\!\int_{0}^{t} c_{1}X_u\,du)$ is $\mathcal{F}_t$–measurable,
\begin{align}
\mathbb{E}\!\left[\exp\!\Big(-\!\int_{0}^{T} c_{1}X_u\,du\Big)e^{c_{2}X_T}\,\Big|\,\mathcal{F}_t\right]
&=\exp\!\Big(-\!\int_{0}^{t} c_{1}X_u\,du\Big)\, \exp\!\big(A(t,T)-B(t,T)X_t\big) \notag\\
\Longleftrightarrow\qquad
\mathbb{E}\!\left[\exp\!\Big(-\!\int_{t}^{T} c_{1}X_u\,du\Big)e^{c_{2}X_T}\,\Big|\,\mathcal{F}_t\right]
&=\exp\!\big(A(t,T)-B(t,T)X_t\big).
\end{align}

\paragraph{Step 2. Apply It\^o’s formula to obtain ODEs for $A$ and $B$.}
Let
\[
Y_t\coloneqq \int_{0}^{t} c_{1}X_u\,du
\quad\Rightarrow\quad
dY_t=c_{1}X_t\,dt,\qquad
dY_t\,dY_t=0,\qquad
dY_t\,dX_t=0.
\]
Under the generalized CIR dynamics,
\[
dX_t=\kappa(\theta-X_t)\,dt+\sqrt{\alpha+\beta X_t}\,dW_t,
\qquad
dX_t\,dX_t=(\alpha+\beta X_t)\,dt.
\]
Define \(f(t,u,v)=\exp(-u+A(t,T)-B(t,T)v)\), so that \(M_t=f(t,Y_t,X_t)\). Its derivatives are
\[
\frac{\partial f}{\partial t}=(A'(t,T)-B'(t,T)v)f,\quad
\frac{\partial f}{\partial u}=-f,\quad
\frac{\partial f}{\partial v}=-B(t,T)f,\quad
\frac{\partial^{2}f}{\partial v^{2}}=B^{2}(t,T)f,
\]
and all derivatives involving $u$ beyond first order (or mixed with $v$) vanish in It\^o terms because
$dY_t\,dY_t=dY_t\,dX_t=0$. It\^o’s lemma gives
\begin{align}
dM_t
&= df(t,Y_t,X_t) \notag\\
&=\frac{\partial f}{\partial t}\,dt+\frac{\partial f}{\partial u}\,dY_t+\frac{\partial f}{\partial v}\,dX_t
 +\frac{1}{2}\frac{\partial^{2}f}{\partial v^{2}}\,dX_t\,dX_t \notag\\
&=\Big[A'(t,T)-B'(t,T)X_t - c_{1}X_t - B(t,T)\kappa(\theta-X_t)
 +\tfrac{1}{2}B^{2}(t,T)\big(\alpha+\beta X_t\big)\Big]M_t\,dt \notag\\
&\quad -\,B(t,T)M_t\sqrt{\alpha+\beta X_t}\,dW_t.
\label{eq:dMt}
\hfill\tag{10}
\end{align}
For $M_t$ to be a martingale, the drift in \eqref{eq:dMt} must vanish for all $t$. Since the drift is affine in $X_t$,
both the constant term and the $X_t$–coefficient must be zero:
\begin{equation}
\Big[A'(t,T)-\kappa\theta\,B(t,T)+\tfrac{1}{2}\alpha\,B^{2}(t,T)\Big]
+\Big[-B'(t,T)-c_{1}+\kappa B(t,T)+\tfrac{1}{2}\beta\,B^{2}(t,T)\Big]X_t=0.
\hfill\tag{11}
\end{equation}
Therefore $A$ and $B$ satisfy the Riccati system

\begin{equation}
\left\{\begin{array}{l}
A^{\prime}(t, T)-B(t, T) \kappa \theta+\frac{\alpha}{2} B^2(t, T)=0 \\
B^{\prime}(t, T)+c_1-\kappa B(t, T)-\frac{\beta}{2} B^2(t, T)=0
\end{array}\right.
\end{equation}

\begin{equation}
A(T,T)=0,\; B(T,T)=c_{2}.
\end{equation}

\subsection{Closed–form solution of \(B(t,T)\)}\label{app:solve-B}

We start from the Riccati system obtained in the previous subsection.  
Let \(y(t)\coloneqq B(t,T)\). Then
\begin{equation}\label{eq:riccati-AB}
\left\{
\begin{aligned}
A'(t,T) &= \kappa\theta\,y(t)\;-\;\frac{\alpha}{2}\,y^2(t),\\[2pt]
y'(t)   &= -\,c_1+\kappa\,y(t)\;+\;\frac{\beta}{2}\,y^2(t),\\[2pt]
A(T,T)  &= 0,\qquad y(T)=c_2 .
\end{aligned}\right.
\end{equation}

To eliminate the quadratic term,
we multiply the first ODE by \(\beta\) and the second by \(\alpha\):
\begin{equation*}
\left\{
\begin{aligned}
\beta A'(t,T)&=\beta\kappa\theta\,y(t)\;-\;\frac{\alpha\beta}{2}\,y^2(t),\\
\alpha y'(t) &= -\alpha c_1+\alpha\kappa y(t)\;+\;\frac{\alpha\beta}{2}\,y^2(t).
\end{aligned}\right.
\end{equation*}
Adding them cancels the square term:

\begin{equation} \label{eq:13}
\beta A'(t,T)+\alpha y'(t)=\kappa(\beta\theta+\alpha)\,y(t)-\alpha c_1 .
\end{equation}
Integrating \((\ref{eq:13})\) from \(t\) to \(T\) yields
\begin{align}
\beta\!\int_t^T\! A'(s,T)\,ds
	&=\;-\alpha\!\int_t^T\! y'(s)\,ds
		+(\beta\theta+\alpha)\kappa\!\int_t^T\! y(s)\,ds
		-\alpha c_1\!\int_t^T\!ds , \\[2pt]
\Rightarrow\quad
-\beta A(t,T)
	&=\;-\alpha\big(c_2-y(t)\big)
		+(\beta\theta+\alpha)\kappa\!\int_t^T\! B(s,T)\,ds
		-\alpha c_1 (T-t),
\end{align}
where \(A(T,T)=0\) and \(y(T)=c_2\) are used. Hence
\begin{equation}
A(t,T)=\frac{\alpha}{\beta}\big(c_2-B(t,T)\big)
-\frac{(\beta\theta+\alpha)\kappa}{\beta}\int_t^T\! B(s,T)\,ds
+\frac{\alpha c_1}{\beta}(T-t) .
\end{equation}

\paragraph{Solving the Riccati for \(y(t)=B(t,T)\).}
Consider the scalar Riccati with terminal value \(y(T)=c_2\):
\begin{equation}
\left\{
\begin{aligned}
y'(t)&=-\,c_1+\kappa y(t)+\frac{\beta}{2}y^2(t),\\
y(T)&=c_2 .
\end{aligned}
\right.
\end{equation}
It is separable:
\begin{equation}
\frac{dy(t)}{dt}=-\,c_1+\kappa y(t)+\frac{\beta}{2}y^2(t)
\qquad\Rightarrow\qquad
\frac{dy(t)}{c_1-\kappa y(t)-\frac{\beta}{2}y^2(t)}=-dt .
\end{equation}

\begin{equation}
\Rightarrow \quad \int \frac{d y(t)}{c_1-\kappa y(t)-\frac{\beta}{2} y^2(t)}=-\int d t+C_1=-t+C_1
\end{equation}

\paragraph{Completing the square.}
The denominator on the left can be written explicitly as a perfect square:
\begin{align}
c_1-\kappa y(t)-\frac{\beta}{2}y^2(t)
&=c_1-\frac{\beta}{2}\!\left(y^2(t)+\frac{2\kappa}{\beta}y(t)\right) \notag\\
&=c_1-\frac{\beta}{2}\!\left[\Big(y(t)+\frac{\kappa}{\beta}\Big)^2
-\Big(\frac{\kappa}{\beta}\Big)^2\right] \notag\\
&=c_1+\frac{\kappa^2}{2\beta}-\frac{\beta}{2}\!\left(y(t)+\frac{\kappa}{\beta}\right)^2 \notag\\
&=\frac{1}{2\beta}\!\left[\,2\beta c_1+\kappa^2-\big(\beta y(t)+\kappa\big)^2\,\right].
\end{align}
Assume \(2\beta c_1+\kappa^2>0\) and set
\begin{equation}
\gamma\;\coloneqq\;\sqrt{2\beta c_1+\kappa^2}\;>0 .
\end{equation}
and then

\begin{equation} \label{eq:22}
\begin{gathered}
c_1-\kappa y(t)-\frac{\beta}{2} y^2(t)=\frac{1}{2 \beta}\left[\gamma^2-(\beta y(t)+\kappa)^2\right] \\
-t+C_1=\int \frac{d y(t)}{c_1-\kappa y(t)-\frac{\beta}{2} y^2(t)}=\int \frac{2 \beta d y(t)}{\gamma^2-(\beta y(t)+\kappa)^2} \\
\Rightarrow \int \frac{d y(t)}{\gamma^2-(\beta y(t)+\kappa)^2}=-\frac{1}{2 \beta} t+\frac{1}{2 \beta} C_1
\end{gathered}
\end{equation}

\paragraph{Partial fractions and integration.}
Using
\[
\frac{1}{\gamma^2-z^2}=\frac{1}{2\gamma}\!\left(\frac{1}{\gamma-z}+\frac{1}{\gamma+z}\right),
\quad z\equiv\beta y(t)+\kappa,
\]
multiplying \((\ref{eq:22})\) by \(\gamma/\beta\) gives
\begin{align}
-\frac{\gamma}{\beta}t+\frac{\gamma}{\beta}C_1
	&=\int\frac{2\gamma\,dy(t)}{\gamma^2-\big(\beta y(t)+\kappa\big)^2}
	=\int\frac{dy(t)}{\gamma-(\beta y(t)+\kappa)}
	+\int\frac{dy(t)}{\gamma+(\beta y(t)+\kappa)} \notag\\
	&=-\frac{1}{\beta}\ln\!\big(\gamma-\beta y(t)-\kappa\big)
		+\frac{1}{\beta}\ln\!\big(\gamma+\beta y(t)+\kappa\big).
\end{align}
Hence
\begin{equation}
\ln\!\left(\frac{\gamma+\beta y(t)+\kappa}{\gamma-\beta y(t)-\kappa}\right)
=-\gamma t+C_2,\qquad C_2\equiv\gamma C_1 .
\notag
\end{equation}
Imposing \(y(T)=c_2\) determines \(C_2\):
\begin{equation}
\ln\!\left(\frac{\gamma+\beta c_2+\kappa}{\gamma-\beta c_2-\kappa}\right)
=-\gamma T+C_2
\;\Rightarrow\;
C_2=\ln\!\left(\frac{\gamma+\beta c_2+\kappa}{\gamma-\beta c_2-\kappa}\right)+\gamma T .
\end{equation}
Therefore, for any \(t\le T\),
\begin{equation}
\ln\!\left(\frac{\gamma+\beta y(t)+\kappa}{\gamma-\beta y(t)-\kappa}\right)
=\gamma(T-t)
+\ln\!\left(\frac{\gamma+\beta c_2+\kappa}{\gamma-\beta c_2-\kappa}\right).
\end{equation}

After taking the exponential on the both sides,
\begin{equation}
\begin{aligned}
& \frac{\gamma+\beta y(t)+\kappa}{\gamma-\beta y(t)-\kappa}=e^{\gamma(T-t)} \frac{\gamma+\beta c_2+\kappa}{\gamma-\beta c_2-\kappa} \\
\Rightarrow & (\gamma+\kappa)+\beta y(t)=e^{\gamma(T-t)} \frac{\gamma+\beta c_2+\kappa}{\gamma-\beta c_2-\kappa}(\gamma-\beta y(t)-\kappa)=(\gamma-\kappa) \frac{\gamma+\beta c_2+\kappa}{\gamma-\beta c_2-\kappa} e^{\gamma(T-t)}-\beta y(t) \frac{\gamma+\beta c_2+\kappa}{\gamma-\beta c_2-\kappa} e^{\gamma(T-t)} \\
\Rightarrow & \left(\beta+\beta \frac{\gamma+\kappa+\beta c_2}{\gamma-\kappa-\beta c_2} e^{\gamma(T-t)}\right) y(t)=\frac{(\gamma-\kappa)\left(\gamma+\kappa+\beta c_2\right)}{\gamma-\kappa-\beta c_2} e^{\gamma(T-t)}-(\gamma+\kappa) \\
\Rightarrow & y(t)=\frac{\frac{(\gamma-\kappa)\left(\gamma+\kappa+\beta c_2\right)}{\gamma-\kappa-\beta c_2} e^{\gamma(T-t)}-(\gamma+\kappa)}{\beta\left(1+\frac{\gamma+\kappa+\beta c_2}{\gamma-\kappa-\beta c_2} e^{\gamma(T-t)}\right)}
\end{aligned}
\end{equation}
Multiplying numerator and denominator of by \(\gamma-\kappa-\beta c_2\) and expanding gives
\begin{equation}
\begin{aligned}
y(t) & =\frac{(\gamma-\kappa)\left(\gamma+\kappa+\beta c_2\right) e^{\gamma(T-t)}-(\gamma+\kappa)\left(\gamma-\kappa-\beta c_2\right)}{\beta\left[\left(\gamma-\kappa-\beta c_2\right)+\left(\gamma+\kappa+\beta c_2\right) e^{\gamma(T-t)}\right]} \\
& =\frac{(\gamma-\kappa)\left(\gamma+\kappa+\beta c_2\right)+(\gamma-\kappa)\left(\gamma+\kappa+\beta c_2\right)\left(e^{\gamma(T-t)}-1\right)-(\gamma+\kappa)\left(\gamma-\kappa-\beta c_2\right)}{\beta\left[\left(\gamma-\kappa-\beta c_2+\gamma+\kappa+\beta c_2\right)+\left(\gamma+\kappa+\beta c_2\right)\left(e^{\gamma(T-t)}-1\right)\right]} \\
& =\frac{(\gamma-\kappa+\gamma+\kappa) \beta c_2+(\gamma-\kappa)\left(\gamma+\kappa+\beta c_2\right)\left(e^{\gamma(T-t)}-1\right)}{\beta\left[2 \gamma+\left(\gamma+\kappa+\beta c_2\right)\left(e^{\gamma(T-t)}-1\right)\right]} \\
& =\frac{2 \gamma \beta c_2+(\gamma-\kappa)\left(\gamma+\kappa+\beta c_2\right)\left(e^{\gamma(T-t)}-1\right)}{2 \gamma \beta+\beta\left(\gamma+\kappa+\beta c_2\right)\left(e^{\gamma(T-t)}-1\right)}
\end{aligned}
\end{equation}
Dividing numerator and denominator in \(2\gamma\beta\) yields a symmetric form:
\begin{equation}
y(t)=\frac{c_2+\frac{(\gamma-\kappa)\left(\gamma+\kappa+\beta c_2\right)}{2 \gamma \beta}\left(e^{\gamma(T-t)}-1\right)}{1+\frac{\gamma+\kappa+\beta c_2}{2 \gamma}\left(e^{\gamma(T-t)}-1\right)}
\end{equation}

\paragraph{Closed form for \(B(t,T)\).}
Recalling \(y(t)=B(t,T)\), we obtain the analytic solution
\begin{equation}\label{eq:B-closed-two}
B(t,T)=
\frac{\,c_2+\dfrac{(\gamma-\kappa)(\gamma+\kappa+\beta c_2)}{2\gamma\beta}\big(e^{\gamma(T-t)}-1\big)}
{\,1+\dfrac{\gamma+\kappa+\beta c_2}{2\gamma}\big(e^{\gamma(T-t)}-1\big)}
\end{equation}
Equivalently—using \(\gamma^2=\kappa^2+2\beta c_1\) so that
\(2c_1=(\gamma^2-\kappa^2)/\beta\)—\eqref{eq:B-closed-two} can be written in the widely used
\(c_1\)-explicit form
\[
B(t,T)=
\frac{\,2c_1\big(e^{\gamma(T-t)}-1\big)+(\gamma-\kappa)c_2 e^{\gamma(T-t)}+(\gamma+\kappa)c_2}
{\,2\gamma+\big(\gamma+\kappa+\beta c_2\big)\big(e^{\gamma(T-t)}-1\big)}
\]
Both formulas are algebraically identical and satisfy the terminal condition \(B(T,T)=c_2\).



\subsection{Deriving the closed form of \(A(t,T)\)}

We use the linear identity between \(A\) and \(B\) obtained from the martingale construction:
\begin{equation}
\label{eq:AB}
A(t,T)= -\frac{\alpha}{\beta}\big(B(t,T)-c_2\big)
       -\frac{(\beta\theta+\alpha)\kappa}{\beta}\int_t^{T} B(s,T)\,ds
       +\frac{\alpha c_1}{\beta}\,(T-t).
\end{equation}

\paragraph{Step 1: Rewrite \(B(t,T)-c_2\) correctly.}
Let
\[
\gamma=\sqrt{\kappa^2+2\beta c_1},\qquad
D=\frac{\gamma+\kappa+\beta c_2}{2\gamma}.
\]
From the closed form of \(B\),
\[
B(t,T)=
c_2+\frac{(\gamma-\kappa)(\gamma+\kappa+\beta c_2)}{2\gamma\beta}\,
      \frac{e^{\gamma(T-t)}-1}{1+D\big(e^{\gamma(T-t)}-1\big)} ,
\]
we have
\begin{align}
B(t,T)-c_2
&=\frac{\dfrac{\gamma-\kappa}{\beta}\,\dfrac{\gamma+\kappa+\beta c_2}{2\gamma}\,
        \big(e^{\gamma(T-t)}-1\big)}
        {1+D\big(e^{\gamma(T-t)}-1\big)} .
\end{align}
\emph{Remark.} The erroneous appearance of \((\gamma-\kappa-\beta c_2)\) in some drafts
comes from moving \(-c_2\) to the numerator without first taking the common denominator.

\paragraph{Step 2: Evaluate \(\displaystyle \int_t^T B(s,T)\,ds\).}
Using the same representation for \(B\) and the change of variables
\(u=T-s\), \(w=e^{\gamma u}\) (so \(du=\tfrac{1}{\gamma w}\,dw\) and \(w\in[1,w_0]\) with \(w_0=e^{\gamma(T-t)}\)),
\begin{align}
\int_t^T B(s,T)\,ds
&=\int_0^{T-t}\!\left[c_2+\frac{\gamma-\kappa}{\beta}\,
\frac{D\big(e^{\gamma u}-1\big)}{1+D\big(e^{\gamma u}-1\big)}\right]du\\
&=\frac{1}{\gamma}\int_{1}^{w_0}\!
\left[\frac{c_2+\frac{\gamma-\kappa}{\beta}\,\frac{Dw}{1+Dw}}{w}
      -\frac{\gamma-\kappa}{\beta}\,\frac{D}{(1+Dw)w}\right]dw .
\end{align}
Use the identity \(\displaystyle \frac{1}{w(1+Dw)}=\frac{1}{w}-\frac{D}{1+Dw}\) to obtain
\begin{align}
\int_t^T B(s,T)\,ds
&=\frac{c_2}{\gamma}\int_{1}^{w_0}\frac{dw}{w}
   -\frac{\gamma-\kappa-\beta c_2}{\gamma\beta}\int_{1}^{w_0}\frac{dw}{1+Dw}
   +\frac{\gamma-\kappa}{\gamma\beta}\int_{1}^{w_0}\frac{dw}{w+1}\\
&=\frac{2}{\beta}\,\ln\!\big(1+D\,w_0\big)
  -\frac{\gamma+\kappa}{\gamma\beta}\,\ln\!\big(w_0+1\big),
\end{align}
where \(w_0=e^{\gamma(T-t)}\).

\paragraph{Step 3: Assemble \(A(t,T)\).}
Substituting the expressions from Steps 1–2 into the linear identity (\ref{eq:AB}) gives
\begin{equation}
\begin{aligned}
A(t,T)
&=\frac{\alpha c_1}{\beta}\,(T-t)
 -\frac{\alpha}{\beta}\,
  \frac{\left(2c_1-2\kappa c_2-\beta c_2^2\right)\left(e^{\gamma(T-t)}-1\right)}
       {2\gamma+\left(\gamma+\kappa+\beta c_2\right)\left(e^{\gamma(T-t)}-1\right)} \\
&\quad +\frac{2(\beta\theta+\alpha)\kappa}{\beta^{2}}
 \ln\!\left(
 \frac{2\gamma\,e^{\frac{(\gamma+\kappa)(T-t)}{2}}}
      {2\gamma+\left(\gamma+\kappa+\beta c_2\right)\left(e^{\gamma(T-t)}-1\right)}
 \right).
\end{aligned}
\end{equation}

\paragraph{For completeness: closed form of \(B(t,T)\).}
An equivalent expression that will be used elsewhere is
\begin{equation}
B(t,T)=
\frac{\,2c_{1}\big(e^{\gamma(T-t)}-1\big)+(\gamma-\kappa)c_{2}e^{\gamma(T-t)}+(\gamma+\kappa)c_{2}\,}
     {\,2\gamma+\big(\gamma+\kappa+\beta c_{2}\big)\big(e^{\gamma(T-t)}-1\big)} .
\end{equation}
Both formulas satisfy \(A(T,T)=0\) and \(B(T,T)=c_2\).

\subsection{General Solution}

In this section we present the general closed-form solution for the extended G--CIR model
under the risk–neutral measure $\mathbb Q$, and then give an equivalent representation
in terms of $\mathbb P$-measure primitives with a constant market price of risk $\lambda$.

\paragraph{Convention.}
We denote risk–neutral parameters by primes $(\kappa',\theta',\alpha,\beta)$.
When one starts from physical-measure primitives $(\kappa,\theta,\alpha,\beta)$ with
a constant market price of risk $\lambda$, the two sets are linked by the mapping below.

\paragraph{Measure mapping (constant market price of risk).}
Starting from $\mathbb P$-measure dynamics
\[
dx_t=\kappa(\theta-x_t)\,dt+\sqrt{\alpha+\beta x_t}\,dW_t^{\mathbb P},
\]
with $dW_t^{\mathbb Q}=dW_t^{\mathbb P}+\lambda\sqrt{\alpha+\beta x_t}\,dt$, the
$\mathbb Q$-measure parameters are
\begin{equation}
\label{eq:Q-from-P-lambda}
\kappa'=\kappa+\beta\lambda,\qquad
\theta'=\frac{\kappa\theta-\alpha\lambda}{\kappa+\beta\lambda}.
\end{equation}
Conversely, given $(\kappa',\theta',\alpha,\beta)$ and $\lambda$,
\begin{equation}
\label{eq:P-from-Q-lambda}
\kappa=\kappa'-\beta\lambda,\qquad
\theta=\frac{\theta'\kappa'+\alpha\lambda}{\kappa'-\beta\lambda}.
\end{equation}
The identity
\begin{equation}
\label{eq:log-invariance-lambda}
(\beta\theta'+\alpha)\,\kappa'=(\beta\theta+\alpha)\,\kappa
\end{equation}
implies that the coefficient of the logarithmic term in $A(\cdot)$ is measure-invariant.

\paragraph{General affine solution under $\mathbb Q$.}
For generic Laplace loadings $(c_1,c_2)$, set $\tau\coloneqq T-t$ and
\[
\gamma'\equiv \sqrt{\kappa'^2+2\beta\,c_1}.
\]
The transform
$\mathbb E_t^{\mathbb Q}\!\big[e^{-\int_t^{t+\tau} c_1 x_u\,du}\,e^{c_2 x_{t+\tau}}\big]
=\exp\!\big(A^{\mathbb Q}(t,T)-B^{\mathbb Q}(t,T)\,x_t\big)$
has the closed forms
\begin{align}
B^{\mathbb Q}(t,T)
&=\frac{\,2c_{1}\big(e^{\gamma'\tau}-1\big)+(\gamma'-\kappa')c_{2}\,e^{\gamma'\tau}+(\gamma'+\kappa')c_{2}\,}
       {\,2\gamma'+\big(\gamma'+\kappa'+\beta c_{2}\big)\big(e^{\gamma'\tau}-1\big)}, \\
A^{\mathbb Q}(t,T)
&=\frac{\alpha c_{1}}{\beta}\,\tau
-\frac{\alpha}{\beta}\,
\frac{\left(2c_{1}-2\kappa' c_{2}-\beta c_{2}^{2}\right)\left(e^{\gamma'\tau}-1\right)}
     {\,2\gamma'+\left(\gamma'+\kappa'+\beta c_{2}\right)\left(e^{\gamma'\tau}-1\right)} \notag\\
&\quad +\frac{2(\beta\theta'+\alpha)\kappa'}{\beta^{2}}
\ln\!\left(
\frac{2\gamma'\,e^{\frac{(\gamma'+\kappa')\tau}{2}}}
     {\,2\gamma'+\left(\gamma'+\kappa'+\beta c_{2}\right)\left(e^{\gamma'\tau}-1\right)}
\right).
\end{align}

\paragraph{Equivalent $\mathbb P$-parameter representation (constant $\lambda$).}
Substituting \eqref{eq:Q-from-P-lambda} into the above and using
\eqref{eq:log-invariance-lambda}, one obtains the same affine coefficients written
directly in terms of $(\kappa,\theta,\alpha,\beta)$ and $\lambda$.
Define
\[
\gamma^\ast\equiv \sqrt{(\kappa+\beta\lambda)^2+2\beta\,c_1}.
\]
Then
\begin{align}
B^{\mathbb P}(t,T;\lambda)
&=\frac{\,2c_{1}\big(e^{\gamma^\ast\tau}-1\big)
+\big(\gamma^\ast-\kappa-\beta\lambda\big)c_{2}\,e^{\gamma^\ast\tau}
+\big(\gamma^\ast+\kappa+\beta\lambda\big)c_{2}\,}
{\,2\gamma^\ast+\big(\gamma^\ast+\kappa+\beta\lambda+\beta c_{2}\big)\big(e^{\gamma^\ast\tau}-1\big)}, \\
A^{\mathbb P}(t,T;\lambda)
&=\frac{\alpha c_{1}}{\beta}\,\tau
-\frac{\alpha}{\beta}\,
\frac{\left(2c_{1}-2(\kappa+\beta\lambda) c_{2}-\beta c_{2}^{2}\right)\left(e^{\gamma^\ast\tau}-1\right)}
     {\,2\gamma^\ast+\left(\gamma^\ast+\kappa+\beta\lambda+\beta c_{2}\right)\left(e^{\gamma^\ast\tau}-1\right)} \notag\\
&\quad +\frac{2(\beta\theta+\alpha)\kappa}{\beta^{2}}
\ln\!\left(
\frac{2\gamma^\ast\,e^{\frac{(\gamma^\ast+\kappa+\beta\lambda)\tau}{2}}}
     {\,2\gamma^\ast+\left(\gamma^\ast+\kappa+\beta\lambda+\beta c_{2}\right)\left(e^{\gamma^\ast\tau}-1\right)}
\right).
\end{align}
By construction $A^{\mathbb P}(t,T;\lambda)=A^{\mathbb Q}(t,T)$ and
$B^{\mathbb P}(t,T;\lambda)=B^{\mathbb Q}(t,T)$ for the same $(c_1,c_2)$.

\paragraph{Admissibility.}
All loadings $c_1$ are taken in the Riccati admissibility domain so that
$\gamma'$ and $\gamma^\ast$ are real. Standard parameter restrictions
ensuring nonnegativity of $x_t$ apply to the extended G--CIR setting.

\subsection{Special Cases Solutions}

We start from the physical-measure (\( \mathbb P \)) dynamics of the generalized CIR factor:
\begin{equation}\label{eq:P-SDE}
dx_t \;=\; \kappa(\theta-x_t)\,dt \;+\; \sqrt{\alpha+\beta x_t}\,dW_t^{\mathbb P},
\qquad \alpha\ge 0,\;\beta>0,\;\kappa>0,\;\theta\ge 0.
\end{equation}
Pricing is done under the risk-neutral measure \( \mathbb Q \) for the transform
\[
\mathbb E_t^{\mathbb Q}\!\left[\exp\!\Big(-\!\int_t^{t+\tau} x_u\,du\Big)\right]
= \exp\!\big(A(\tau)-B(\tau)\,x_t\big),
\]
with \(A,B\) solving the Riccati system derived earlier.

\paragraph{Change of measure.}
Assume a \emph{constant} market price of risk level \( \lambda\in\mathbb R \), and take the
\emph{risk-price process}
\[
\Lambda_t \;\coloneqq\; \lambda\,\sqrt{\alpha+\beta x_t}.
\]
Define the Radon--Nikodym density
\[
\frac{d\mathbb Q}{d\mathbb P}\Big|_{\mathcal F_t}
=\exp\!\left(-\int_0^t \Lambda_s\,dW_s^{\mathbb P}-\frac12\int_0^t \Lambda_s^2\,ds\right),
\quad
dW_t^{\mathbb Q}=dW_t^{\mathbb P}+\Lambda_t\,dt .
\]
Applying Girsanov to \eqref{eq:P-SDE}, the \( \mathbb Q \)-dynamics become
\begin{equation}\label{eq:Q-SDE-raw}
dx_t \;=\; \Big(\kappa(\theta-x_t)-\Lambda_t\,\sqrt{\alpha+\beta x_t}\Big)dt
\;+\; \sqrt{\alpha+\beta x_t}\,dW_t^{\mathbb Q}
\;=\; \big[\kappa(\theta-x_t)-\lambda(\alpha+\beta x_t)\big]dt
\;+\; \sqrt{\alpha+\beta x_t}\,dW_t^{\mathbb Q}.
\end{equation}
The drift in \eqref{eq:Q-SDE-raw} can be written in mean-reverting form
\begin{equation}\label{eq:Q-kp-thetap}
dx_t \;=\; \kappa'(\theta'-x_t)\,dt \;+\; \sqrt{\alpha+\beta x_t}\,dW_t^{\mathbb Q},
\qquad
\kappa' \coloneqq \kappa+\beta\lambda,\qquad
\theta' \coloneqq \frac{\kappa\theta-\alpha\lambda}{\kappa+\beta\lambda}.
\end{equation}
Indeed,
\[
\kappa'(\theta'-x)
=(\kappa+\beta\lambda)\!\left(\frac{\kappa\theta-\alpha\lambda}{\kappa+\beta\lambda}-x\right)
=\kappa\theta-\alpha\lambda-\kappa x-\beta\lambda x
=\kappa(\theta-x)-\lambda(\alpha+\beta x).
\]

\paragraph{Affine transform under \( \mathbb Q \).}
With \eqref{eq:Q-kp-thetap}, the transform keeps exactly the same functional form as in the
general case, with parameters \( (\kappa',\theta') \) and
\(
\gamma' \coloneqq \sqrt{\kappa'^2+2\beta c_1}
\)
(where \(c_1\) is the Laplace loading; for discounting we will set \(c_1=1\), \(c_2=0\)).
A useful identity used below is
\begin{equation}\label{eq:coeff-cancel}
(\beta\theta'+\alpha)\,\kappa'
=\Big(\beta\,\frac{\kappa\theta-\alpha\lambda}{\kappa'}+\alpha\Big)\kappa'
=\beta\kappa\theta-\beta\alpha\lambda+\alpha\kappa+\alpha\beta\lambda
= (\beta\theta+\alpha)\kappa,
\end{equation}
which shows that the coefficient multiplying the log term is \emph{invariant} to the
change of measure when the market price of risk is chosen as above.

\bigskip

\paragraph{Case 1: Classical CIR Model}

Setting \(\alpha=0\), \(\beta=\sigma^{2}\), \(c_{1}=1\), \(c_{2}=0\) gives
\[
dX_t=\kappa(\theta-X_t)\,dt+\sigma\sqrt{X_t}\,dW_t,
\qquad
\gamma=\sqrt{\kappa^{2}+2\sigma^{2}}.
\]
The affine transform for the discount factor
\[
\mathbb{E}\!\left[\exp\left(-\int_t^{T}X_u\,du\right)\middle|\mathcal{F}_t\right]
=\exp\!\big(A(t,T)-B(t,T)X_t\big)
\]
reduces to the classical CIR expressions:
\begin{align}
B(t,T)
&=\frac{ 2\big(e^{\gamma (T-t)}-1\big) }
       { 2\gamma+(\gamma+\kappa)\big(e^{\gamma (T-t)}-1\big) }, \\
A(t,T)
&=\frac{2\kappa\theta}{\beta}\,
\ln\!\left(\frac{2\gamma\,e^{\frac{(\gamma+\kappa)(T-t)}{2}}}
                 {\,2\gamma+(\gamma+\kappa)\big(e^{\gamma (T-t)}-1\big)}\right).
\end{align}

\paragraph{Case 2 (Zero market price of risk; $\mathbb Q$-dynamics coincide with $\mathbb P$-dynamics).}
Start from the physical-measure dynamics
\begin{equation}
dx_t=\kappa(\theta-x_t)\,dt+\sqrt{\alpha+\beta x_t}\,dW_t^{\mathbb P}.
\end{equation}
Let the (constant) market price of risk be $\lambda$. Under the change of measure with
risk-price process $\Lambda_t=\lambda\sqrt{\alpha+\beta x_t}$,
\[
\frac{d\mathbb Q}{d\mathbb P}\Big|_{\mathcal F_t}
=\exp\!\left(-\int_0^t \Lambda_s\,dW_s^{\mathbb P}-\frac12\int_0^t \Lambda_s^2\,ds\right),
\qquad
dW_t^{\mathbb Q}=dW_t^{\mathbb P}+\Lambda_t\,dt .
\]
When $\lambda=0$ we have $\Lambda_t\equiv 0$, hence
\[
\frac{d\mathbb Q}{d\mathbb P}\Big|_{\mathcal F_t}=1,\qquad dW_t^{\mathbb Q}=dW_t^{\mathbb P},
\]
so the $\mathbb Q$-dynamics coincide with those under $\mathbb P$:
\begin{equation}
dx_t=\kappa(\theta-x_t)\,dt+\sqrt{\alpha+\beta x_t}\,dW_t^{\mathbb Q}.
\end{equation}
For the discounting transform (take $c_1=1$, $c_2=0$) and
$\gamma=\sqrt{\kappa^2+2\beta}$, the affine coefficients are
\begin{align}
B(\tau)&=\frac{2\big(e^{\gamma\tau}-1\big)}
               {2\gamma+(\gamma+\kappa)\big(e^{\gamma\tau}-1\big)},\\
A(\tau)&=\frac{\alpha}{\beta}\,\tau
-\frac{2\alpha}{\beta}\,
 \frac{e^{\gamma\tau}-1}{2\gamma+(\gamma+\kappa)\big(e^{\gamma\tau}-1\big)}
+\frac{2(\beta\theta+\alpha)\kappa}{\beta^{2}}
 \ln\!\left(\frac{2\gamma\,e^{\frac{(\gamma+\kappa)\tau}{2}}}
 {2\gamma+(\gamma+\kappa)\big(e^{\gamma\tau}-1\big)}\right).
\end{align}
Case 2 is the special case of Case 3 with $\lambda=0$, so that
$\kappa'=\kappa$ and $\theta'=\theta$. Consequently, the affine
coefficients in Case 3 reduce exactly to those in Case 2 when $\lambda=0$.

\bigskip
\noindent\textbf{Case 3 (constant market price of risk \( \lambda\)).}

Under the physical measure \(\mathbb P\),
\[
dx_t=\kappa(\theta-x_t)\,dt+\sqrt{\alpha+\beta x_t}\,dW_t^{\mathbb P}.
\]
With a constant market price of risk \(\lambda\) and
\(dW_t^{\mathbb Q}=dW_t^{\mathbb P}+\lambda\sqrt{\alpha+\beta x_t}\,dt\),
the \(\mathbb Q\)-dynamics are
\[
dx_t=\big[\kappa(\theta-x_t)-\lambda(\alpha+\beta x_t)\big]dt
     +\sqrt{\alpha+\beta x_t}\,dW_t^{\mathbb Q}
   =\kappa'(\theta'-x_t)\,dt+\sqrt{\alpha+\beta x_t}\,dW_t^{\mathbb Q},
\]
where
\[
\kappa'=\kappa+\beta\lambda,\qquad
\theta'=\frac{\kappa\theta-\alpha\lambda}{\kappa+\beta\lambda}.
\]

For the affine transform
\(\mathbb E_t^{\mathbb Q}\!\big[e^{-\int_t^{t+\tau}c_1 x_u\,du}\,e^{c_2 x_{t+\tau}}\big]
=\exp(A(\tau)-B(\tau)x_t)\),
let \(\gamma'=\sqrt{\kappa'^2+2\beta c_1}\). Then
\begin{align}
B(\tau)
&=\frac{2c_1\!\left(e^{\gamma'\tau}-1\right)+(\gamma'-\kappa')c_2 e^{\gamma'\tau}+(\gamma'+\kappa')c_2}
       {2\gamma'+(\gamma'+\kappa'+\beta c_2)\left(e^{\gamma'\tau}-1\right)},\\
A(\tau)
&=\frac{\alpha c_1}{\beta}\,\tau
-\frac{\alpha}{\beta}\,
 \frac{\left(2c_1-2\kappa' c_2-\beta c_2^2\right)\left(e^{\gamma'\tau}-1\right)}
      {2\gamma'+(\gamma'+\kappa'+\beta c_2)\left(e^{\gamma'\tau}-1\right)}
+\frac{2(\beta\theta'+\alpha)\kappa'}{\beta^2}
 \ln\!\left(
 \frac{2\gamma' e^{\frac{(\gamma'+\kappa')\tau}{2}}}
      {2\gamma'+(\gamma'+\kappa'+\beta c_2)\left(e^{\gamma'\tau}-1\right)}
 \right).
\end{align}

In the discounting case (\(c_1=1,\;c_2=0\)) with \(\gamma'=\sqrt{\kappa'^2+2\beta}\),
\begin{align}
B(\tau)
&=\frac{2\left(e^{\gamma'\tau}-1\right)}
       {2\gamma'+(\gamma'+\kappa')\left(e^{\gamma'\tau}-1\right)},\\
A(\tau)
&=\frac{\alpha}{\beta}\,\tau
-\frac{2\alpha}{\beta}\,
 \frac{e^{\gamma'\tau}-1}{2\gamma'+(\gamma'+\kappa')\left(e^{\gamma'\tau}-1\right)}
+\frac{2(\beta\theta'+\alpha)\kappa'}{\beta^2}
 \ln\!\left(
 \frac{2\gamma' e^{\frac{(\gamma'+\kappa')\tau}{2}}}
      {2\gamma'+(\gamma'+\kappa')\left(e^{\gamma'\tau}-1\right)}
 \right).
\end{align}

Thus, the closed-form solutions of $A(\tau)$ and $B(\tau)$ expressed by parameters under $\mathbb P$ are:

$$
\begin{gathered}
A(\tau)=\frac{\alpha}{\beta} \tau-\frac{2 \alpha}{\beta} \frac{e^{\gamma \tau}-1}{2 \gamma+(\gamma+\kappa+\beta \lambda)\left(e^{\gamma \tau}-1\right)}+\frac{2(\beta \theta+\alpha) \kappa}{\beta^2} \ln \left(\frac{2 \gamma e^{\frac{(\gamma+\kappa+\beta \lambda) \tau}{2}}}{2 \gamma+(\gamma+\kappa+\beta \lambda)\left(e^{\gamma \tau}-1\right)}\right) \\
B(\tau)=\frac{2\left(e^{\gamma \tau}-1\right)}{2 \gamma+(\gamma+\kappa+\beta \lambda)\left(e^{\gamma \tau}-1\right)}
\end{gathered}
$$

where $\gamma=\sqrt{(\kappa+\beta \lambda)^2+2 \beta}$.

\section{Pricing defaultable zero-coupon bonds in the Cox--Markov model}
\label{sec:cox-pricing-onepiece}

\paragraph{Setup.}
Work on $(\Omega,\mathcal F,(\mathcal F_t)_{t\ge0},\mathbb Q)$.
Let $X=(X_t)_{t\ge0}$ be an $\mathcal F_t$-adapted state process and let the short rate be
$r_t=R(X_t)$. Let $\eta=(\eta_t)_{t\ge0}$ be a continuous-time Markov chain on the rating set
$\{1,\ldots,K\}$ with the absorbing default state $K$.
We work under the \emph{Cox--Markov} assumption: conditional on the path of $X$,
$\eta$ is a (possibly time-inhomogeneous) Markov chain with $\mathcal F_t$-adapted generator
$A_X(u)\in\mathbb R^{K\times K}$ on $[t,T]$.
Let $P_X(s,t)=[P_X(s,t)_{im}]$ denote the conditional transition matrix given $X$ on $[s,t]$; it solves
the backward Kolmogorov equation
\begin{equation}\label{eq:BK}
\partial_s P_X(s,t)=-A_X(s)\,P_X(s,t),\qquad P_X(t,t)=I_K .
\end{equation}
Define the default time $\tau=\inf\{u\ge0:\eta_u=K\}$ and the enlarged filtration
$\mathcal G_t:=\mathcal F_t\vee\sigma(\eta_u: u\le t)$.

\paragraph{Proposition (pricing identity).}
Assume zero recovery. For a firm with current rating $\eta_t=i\in\{1,\ldots,K-1\}$,
the time-$t$ price of a defaultable zero-coupon bond maturing at $T$ is
\begin{equation}\label{eq:v-goal}
v^{\,i}(t,T)=
\mathbb E_t^{\mathbb Q}\!\left[
\exp\!\Big(-\!\int_t^T R(X_u)\,du\Big)\,\big(1-P_X(t,T)_{iK}\big)
\ \Big|\ \mathcal G_t
\right].
\end{equation}

\begin{proof}
Let $B_t=\exp(\int_0^t r_u\,du)$ be the money-market account under $\mathbb Q$.
With zero recovery, the payoff is $1_{\{\tau>T\}}$. Hence,
\[
v^{\,i}(t,T)=B_t\,\mathbb E_t^{\mathbb Q}\!\left[\left.\frac{1_{\{\tau>T\}}}{B_T}\right|\mathcal G_t\right]
=\mathbb E_t^{\mathbb Q}\!\left[
1_{\{\tau>T\}}\exp\!\Big(-\!\int_t^T R(X_u)\,du\Big)\ \Big|\ \mathcal G_t
\right].
\]
By the tower property and the Cox--Markov property,
\[
\mathbb E_t^{\mathbb Q}\!\left[1_{\{\tau>T\}}\mid \mathcal F_T\vee\mathcal G_t\right]
=
\mathbb P_t^{\mathbb Q}\!\left(\tau>T\mid \mathcal F_T\vee\mathcal G_t\right)
=1-P_X(t,T)_{iK},
\]
where $P_X(t,T)_{iK}$ is the conditional default probability on $[t,T]$ given the path of $X$.
Since $\exp(-\int_t^T R(X_u)\,du)$ is $\mathcal F_T$-measurable, substituting into the previous display
yields \eqref{eq:v-goal}.
\end{proof}

\paragraph{Diagonal propagator under a common-eigenvector generator.}
Fix $[s,t]\subseteq [0,\infty)$ and condition on a path of $X$ on $[s,t]$.
Assume the conditional generator admits a common-eigenvector representation
\begin{equation}\label{eq:AX-diag}
A_X(u)=B\,\mu(X_u)\,B^{-1},\qquad
\mu(X_u):=\mathrm{diag}\!\big(\mu_1(X_u),\ldots,\mu_{K-1}(X_u),0\big),
\end{equation}
where $B$ is constant and invertible and state $K$ is absorbing (hence the last diagonal entry is $0$).
Define $\widetilde P(s,t):=B^{-1}P_X(s,t)B$. Substituting into \eqref{eq:BK} gives
\begin{equation}\label{eq:tilde-ODE}
\partial_s \widetilde P(s,t)=-\mu(X_s)\,\widetilde P(s,t),\qquad \widetilde P(t,t)=I_K.
\end{equation}
Because $\mu(X_s)$ is diagonal, $\widetilde P(s,t)$ remains diagonal and the coordinates satisfy
\[
\partial_s \widetilde P_{jj}(s,t)=-\mu_j(X_s)\widetilde P_{jj}(s,t),\qquad \widetilde P_{jj}(t,t)=1.
\]
Integrating yields the diagonal propagator
\begin{equation}\label{eq:EX}
E_X(s,t):=\widetilde P(s,t)
=\mathrm{diag}\!\Big(
e^{\int_s^{t}\mu_1(X_u)\,du},\ldots,e^{\int_s^{t}\mu_{K-1}(X_u)\,du},\,1\Big),
\end{equation}
and hence
\begin{equation}\label{eq:P-factor}
P_X(s,t)=B\,E_X(s,t)\,B^{-1}.
\end{equation}
Differentiating \eqref{eq:P-factor} and using \eqref{eq:AX-diag} verifies that $P_X$ solves \eqref{eq:BK}.
The key simplification is that diagonal matrices commute over time in the spectral coordinates, so the
time-ordered exponential reduces to \eqref{eq:EX}.

\paragraph{Default column and spectral survival weights (a missing step in the literature).}
Suppress the subscript $X$ for brevity and write $P(t,T)=B E(t,T)B^{-1}$ with $E(t,T)$ as in \eqref{eq:EX}.
Entrywise,
\begin{equation}\label{eq:P-entry}
P_{im}(t,T)=\sum_{k=1}^{K} b_{ik}\,b^{-1}_{km}\,
\exp\!\Big(\int_t^T \mu_k(X_u)\,du\Big).
\end{equation}
Taking $m=K$ and using $\mu_K\equiv 0$ gives
\begin{equation}\label{eq:PiK-split}
P_{iK}(t,T)=\sum_{k=1}^{K-1} b_{ik}\,b^{-1}_{kK}\,e^{\int_t^T\mu_k(X_u)\,du}
\;+\;b_{iK}\,b^{-1}_{KK}.
\end{equation}
By Lemma~\ref{lem:bik} below, $b_{iK}b^{-1}_{KK}=1$ for all $i$, hence
\begin{equation}\label{eq:survival-sum}
1-P_{iK}(t,T)=\sum_{j=1}^{K-1}\beta_{ij}\,
\exp\!\Big(\int_t^T\mu_j(X_u)\,du\Big),
\qquad
\beta_{ij}:=-\,b_{ij}\,b^{-1}_{jK}\ \ge 0,\ \ \sum_{j=1}^{K-1}\beta_{ij}=1.
\end{equation}
The nonnegativity and unit-sum property follow because $1-P_{iK}(t,T)$ is a survival probability and the
representation is unique given the diagonal basis.

\paragraph{Corollary (spectral pricing formula).}
Combining \eqref{eq:v-goal} with \eqref{eq:survival-sum} yields
\begin{equation}\label{eq:pricing-final}
v^{\,i}(t,T)=
\sum_{j=1}^{K-1}\beta_{ij}\,
\mathbb E_t^{\mathbb Q}\!\left[
\exp\!\Big(\int_t^T\big(\mu_j(X_u)-R(X_u)\big)\,du\Big)
\ \Big|\ \mathcal G_t
\right].
\end{equation}
This is the desired one-piece (non-RS) pricing identity used in the main text: conditional on the state path,
credit enters only through the mode loadings $\beta_{ij}$ and the mode intensities $\mu_j(\cdot)$.

\begin{lemma}\label{lem:bik}
Let $A\in\mathbb R^{K\times K}$ be a rating generator with absorbing state $K$ (so the $K$-th row is zero and
each row sums to zero). Let $B$ be a right-eigenvector matrix of $A$ whose last column is a (nonzero) right
eigenvector associated with eigenvalue $0$. Then $b_{iK}b^{-1}_{KK}=1$ for all $i$.
\end{lemma}

\begin{proof}
Since each row of $A$ sums to zero, $A\mathbf 1=\mathbf 0$, hence $\mathbf 1$ is a right eigenvector with eigenvalue $0$.
Choose the last eigenvector as $v_K=c\,\mathbf 1$ with $c\neq 0$, so $b_{iK}=c$ for all $i$.
For any eigenvector $v$ with eigenvalue $\lambda\neq 0$, the $K$-th row of $Av=\lambda v$ gives
$0=\lambda v_K$, hence $v_K=0$. Therefore $B$ has the block form
\[
B=
\begin{pmatrix}
C & c\,\mathbf 1\\
\mathbf 0^\top & c
\end{pmatrix},
\qquad C\in\mathbb R^{(K-1)\times(K-1)}.
\]
By the inverse of an upper block-triangular matrix,
\[
B^{-1}=
\begin{pmatrix}
C^{-1} & -\,C^{-1}\mathbf 1\\
\mathbf 0^\top & c^{-1}
\end{pmatrix},
\]
so $b^{-1}_{KK}=c^{-1}$ and hence $b_{iK}b^{-1}_{KK}=c\cdot c^{-1}=1$ for all $i$.
\end{proof}

\section{Affine Pricing Building Blocks for CGB, CDB, and Corporate Bonds}
\label{app:pricing}

This appendix collects the \emph{fixed-regime} exponential-affine pricing blocks used to evaluate one-step conditional
transforms in the RS--GCIR model. In the full regime-switching model of the main text, regimes evolve according to CTMC
generators $Q^{r}$ and $Q^{c}$ (and $Q=Q^{r}\oplus Q^{c}$ for the joint regime). Pricing under full switching is implemented
by combining (i) fixed-regime one-step affine transforms over an observation interval $\tau$ (the grid spacing) with
(ii) the discrete-time transition matrices $P^{r}(\tau)=\exp(Q^{r}\tau)$ and $P(\tau)=\exp(Q\tau)$ in a backward recursion on
the observation grid. Throughout, expectations are taken under $\mathbb Q$.

Let $X_t=(X_{1,t},X_{2,t},X_{3,t},X_{4,t})^\top$, where $(X_{1,t},X_{2,t},X_{3,t})$ are rate-block factors driven by the
rate regime $s_t^r\in\mathcal S^r$, and $X_{4,t}$ is the credit factor driven by the credit regime $s_t^c\in\mathcal S^c$.
For each factor $k$, denote by $A_k^{(s)}(\tau;c_1,0)$ and $B_k^{(s)}(\tau;c_1,0)$ the one-factor coefficient functions in
Appendix~\ref{app:gcir}, evaluated at the risk-neutral parameters of factor $k$ under regime $s$. Under the main-text
normalization we set the deterministic shift $a=0$ (otherwise $A(\tau)$ below is replaced by $A(\tau)-a\tau$).

Under our normalization, the CGB and CDB short rates are
\begin{equation}\label{eq:app:short_rates}
r_t^{\mathrm{cgb}} = X_{1,t}+X_{2,t},\qquad
r_t^{\mathrm{cdb}} = X_{1,t}+X_{2,t}+X_{3,t}.
\end{equation}
Thus $X_{3,t}$ isolates the systematic convenience yield between the CGB and CDB curves and is governed by the rate regime
$s_t^r$.

For $m\in\{\mathrm{cgb},\mathrm{cdb}\}$ and each $s\in\mathcal S^{r}$, define the one-step conditional discount factor over an
observation interval $\tau$ by
\begin{equation}\label{eq:app:one_step_disc}
M_{m}^{(s)}(t;\tau)
\;\coloneqq\;
\mathbb E_t^{\mathbb Q}\!\left[
\exp\!\left(-\int_t^{t+\tau} r_u^{m}\,du\right)
\ \Big|\ X_t,\,s_t^{r}=s
\right]
=
\exp\!\Big\{A_{m}^{(s)}(\tau)-\big(B_{m}^{(s)}(\tau)\big)^\top X_t\Big\},
\end{equation}
where $(A_m^{(s)}(\tau),B_m^{(s)}(\tau))$ are constructed from the one-factor GCIR blocks in Appendix~\ref{app:gcir}.

\textit{CGB block.}
The factor loadings for $r_t^{\mathrm{cgb}}=X_{1,t}+X_{2,t}$ are $(c_{1,1},c_{1,2},c_{1,3},c_{1,4})=(1,1,0,0)$, hence
\begin{align}
A_{\mathrm{cgb}}^{(s)}(\tau)
&=A_1^{(s)}(\tau;1,0)+A_2^{(s)}(\tau;1,0), \label{eq:app:A_cgb_one}\\
B_{\mathrm{cgb}}^{(s)}(\tau)
&=\big(B_1^{(s)}(\tau;1,0),\ B_2^{(s)}(\tau;1,0),\ 0,\ 0\big)^\top. \label{eq:app:B_cgb_one}
\end{align}

\textit{CDB block.}
The factor loadings for $r_t^{\mathrm{cdb}}=X_{1,t}+X_{2,t}+X_{3,t}$ are $(c_{1,1},c_{1,2},c_{1,3},c_{1,4})=(1,1,1,0)$,
hence
\begin{align}
A_{\mathrm{cdb}}^{(s)}(\tau)
&=\sum_{k=1}^3 A_k^{(s)}(\tau;1,0), \label{eq:app:A_cdb_one}\\
B_{\mathrm{cdb}}^{(s)}(\tau)
&=\big(B_1^{(s)}(\tau;1,0),\ B_2^{(s)}(\tau;1,0),\ B_3^{(s)}(\tau;1,0),\ 0\big)^\top.
\label{eq:app:B_cdb_one}
\end{align}

Let $\mathbf P_m(t,T)\in\mathbb R^{|\mathcal S^r|}$ collect regime-conditional prices, with
$\big[\mathbf P_m(t,T)\big]_s \equiv P_m^{(s)}(t,T)$ for $s\in\mathcal S^r$ and $m\in\{\mathrm{cgb},\mathrm{cdb}\}$.
Define the diagonal matrix of one-step affine blocks
\[
\mathbf M_m(t;\tau)\;\coloneqq\;\mathrm{diag}\Big(M_m^{(s)}(t;\tau)\Big)_{s\in\mathcal S^r}.
\]
On the observation grid $T=t+n\tau$, the regime-conditional prices satisfy the backward recursion
\begin{equation}\label{eq:app:bond_vec_rec}
\mathbf P_m(t,t+n\tau)
\;=\;
\mathbf M_m(t;\tau)\,P^r(\tau)\,\mathbf P_m(t+\tau,t+n\tau),
\qquad
\mathbf P_m(T,T)=\mathbf 1,
\end{equation}
where $P^r(\tau)=\exp(Q^r\tau)$ and $\mathbf 1$ is an all-ones vector.

We price corporate bonds in a rating-based reduced-form framework in the spirit of \citet{lando1998cox,feldhutter2008decomposing}.
Let $\eta_t\in\{1,\ldots,K\}$ denote the rating, where $K$ is an absorbing default state and $\{1,\ldots,K-1\}$ are non-default ratings.

Let $\widetilde\Lambda\in\mathbb R^{(K-1)\times(K-1)}$ be the sub-generator on non-default ratings, diagonalizable as
$\widetilde\Lambda=\widetilde B\,\widetilde D\,\widetilde B^{-1}$ with $\widetilde D=\mathrm{diag}(d_1,\ldots,d_{K-1})$ and $d_j<0$.
Define the Lando weights
\begin{equation}\label{eq:app:w_ij}
w_{ij}\coloneqq -[\widetilde B]_{ij}\,[\widetilde B^{-1}]_{jK},
\qquad i=1,\ldots,K-1,\quad j=1,\ldots,K-1,
\end{equation}
which satisfy $w_{ij}\ge 0$ and $\sum_{j=1}^{K-1} w_{ij}=1$ (see proof in \ref{sec:cox-pricing-onepiece}).

\paragraph{Discounting and the joint-regime credit driver.}
Risky cash flows are discounted on the CDB curve, i.e.\ by $r_t^{\mathrm{cdb}}$ in \eqref{eq:app:short_rates}.
To allow four pass-through coefficients while keeping regime transitions separable, define the joint-regime intensity driver
\begin{equation}\label{eq:app:mu_joint}
\mu_t^{(s^r,s^c)}
\;\coloneqq\;
c_{s^c|s^r}\big(X_{1,t}+X_{2,t}\big)+X_{4,t},
\qquad (s^r,s^c)\in\mathcal S^r\times\mathcal S^c,
\end{equation}
where $c_{s^c|s^r}$ takes four values (e.g., $\{c_{E|L},c_{C|L},c_{E|H},c_{C|H}\}$). Importantly, $X_{3,t}$ enters corporate
valuation \emph{only} through discounting via $r_t^{\mathrm{cdb}}$, not through $\mu_t^{(s^r,s^c)}$.

Let $S_t=(s_t^r,s_t^c)$ denote the joint regime with generator $Q=Q^r\oplus Q^c$ and one-step transition matrix
$P(\tau)=\exp(Q\tau)$. For each joint regime $S=(s^r,s^c)\in\mathcal S$ and each eigenmode $j\in\{1,\ldots,K-1\}$,
define the one-step mode kernel over $\tau$ by
\begin{equation}\label{eq:app:one_step_mode}
M_{j}^{(S)}(t;\tau)
\;\coloneqq\;
\mathbb E_t^{\mathbb Q}\!\left[
\exp\!\left(\int_t^{t+\tau}\big(d_j\,\mu_u^{(S)}-r_u^{\mathrm{cdb}}\big)\,du\right)
\ \Big|\ X_t,\,S_t=S
\right]
=
\exp\!\Big\{A^{(j;S)}(\tau)-\big(B^{(j;S)}(\tau)\big)^\top X_t\Big\},
\end{equation}
where $\mu_u^{(S)}=c_{s^c|s^r}(X_{1,u}+X_{2,u})+X_{4,u}$ and the loadings defining
$(A^{(j;S)}(\tau),B^{(j;S)}(\tau))$ follow from the GCIR Laplace-transform blocks in Appendix~\ref{app:gcir}.

\textit{Loadings and the $(A^{(j;\cdot)},B^{(j;\cdot)})$ blocks.}
Using \eqref{eq:app:short_rates} and \eqref{eq:app:mu_joint}, for each eigenmode $j$ and joint regime $S=(s^r,s^c)$,
\[
r_u^{\mathrm{cdb}}-d_j\mu_u^{(S)}
=
\big(1-d_j c_{s^c|s^r}\big)\big(X_{1,u}+X_{2,u}\big) \;+\; X_{3,u} \;-\; d_j\,X_{4,u}.
\]
Therefore the factor loadings for the Laplace transform in Appendix~\ref{app:gcir} are
\begin{equation}\label{eq:app:c1_j}
c_{1,k}^{(j;S)}=
\begin{cases}
1-d_j c_{s^c|s^r}, & k=1,2,\\
1, & k=3,\\
-d_j, & k=4,
\end{cases}
\qquad c_{2,k}^{(j;S)}\equiv 0.
\end{equation}
Define the mode-specific affine blocks (evaluated at horizon $\tau$)
\begin{equation}\label{eq:app:AjBj_box}
\begin{aligned}
A^{(j; S)}(\tau)
&=\sum_{k=1}^3 A_k^{(s^r)}\!\big(\tau; c_{1,k}^{(j;S)}, 0\big)
  \;+\; A_4^{(s^c)}\!\big(\tau; c_{1,4}^{(j;S)}, 0\big),\\
B^{(j; S)}(\tau)
&=\Big(B_1^{(s^r)}\!\big(\tau; c_{1,1}^{(j;S)}, 0\big),\;
       B_2^{(s^r)}\!\big(\tau; c_{1,2}^{(j;S)}, 0\big),\;
       B_3^{(s^r)}\!\big(\tau; c_{1,3}^{(j;S)}, 0\big),\;
       B_4^{(s^c)}\!\big(\tau; c_{1,4}^{(j;S)}, 0\big)\Big)^{\top}.
\end{aligned}
\end{equation}

Let $\mathbf u_j(t,T)\in\mathbb R^{|\mathcal S|}$ collect the mode-$j$ continuation values across joint regimes, with
$\big[\mathbf u_j(t,T)\big]_S \equiv u_j^{(S)}(t,T)$. Define the diagonal matrix of one-step mode blocks
\[
\mathbf M_j(t;\tau)\;\coloneqq\;\mathrm{diag}\Big(M_j^{(S)}(t;\tau)\Big)_{S\in\mathcal S}.
\]
On the grid $T=t+n\tau$, the mode values satisfy
\begin{equation}\label{eq:app:corp_vec_rec}
\mathbf u_j(t,t+n\tau)
\;=\;
\mathbf M_j(t;\tau)\,P(\tau)\,\mathbf u_j(t+\tau,t+n\tau),
\qquad
\mathbf u_j(T,T)=\mathbf 1,
\end{equation}
where $P(\tau)=\exp(Q\tau)$ and $\mathbf 1$ is an all-ones vector. The joint-regime corporate discount-bond value is then
\begin{equation}\label{eq:app:corp_from_modes}
v^{\,i,S}(t,T)
=
\sum_{j=1}^{K-1} w_{ij}\,u_j^{(S)}(t,T),
\qquad S\in\mathcal S.
\end{equation}

In the main text, observable discount-bond prices are formed by convex mixtures at the price level using filtered \emph{beliefs}
about the rate regime:
\[
P_m(t,T)=\sum_{s\in\mathcal S^r}\pi_{t|t}^r(s)\,P_m^{(s)}(t,T),
\qquad m\in\{\mathrm{cgb},\mathrm{cdb}\}.
\]
Observable corporate prices are formed by price-level mixtures of joint-regime conditional values $v^{\,i,(s^r,s^c)}(t,T)$ using
the block-recursive beliefs $\pi_{t|t}^r(s^r)\,\pi_{t|t}^c(s^c\mid s^r;\widehat{\mathcal R}_t)$, as described in the main text.

\section{Regime-switching UKF: algorithmic details}
\label{app:rsukf}

This appendix describes the regime-switching unscented Kalman filter (RS--UKF)
used for maximum-likelihood estimation under the physical measure. The filter
combines regime-specific UKF recursions with a Gray-style collapsing approximation
for the continuous state. The resulting procedure delivers (i) regime-conditional
filtered moments and (ii) filtered regime probabilities, along with the
corresponding log-likelihood contributions.

Let $X_t\in\mathbb R^{d}$ denote the continuous latent state and
$s_t\in\{1,\ldots,S\}$ the discrete regime at time $t$.
The information set is $\mathcal I_t=\{\mathbf y_\tau:\tau\le t\}$, where
$\mathbf y_t$ is the (stacked) vector of observed yields at time $t$.
For each regime $j$, denote the filtered moments by
$\widehat X_{t|t}(j)$ and $\widehat P_{t|t}(j)$, and the filtered regime
probability by
\[
\pi_t(j)\equiv\mathbb P(s_t=j\mid\mathcal I_t),\qquad j=1,\ldots,S.
\]
Let $F_j(\cdot)$ be the (possibly nonlinear) state transition mapping in regime $j$
and $h_j(\cdot)$ the measurement mapping implied by the bond-pricing formulas.
Process and measurement noise covariances are $Q_j$ and $R_j$, respectively.
The regime chain follows a Markov transition matrix
$P(\Delta t)=[P_{ij}(\Delta t)]_{i,j=1}^{S}$ with
$P_{ij}(\Delta t)=\mathbb P(s_t=j\mid s_{t-\Delta t}=i)$.
For a Gaussian $\mathcal N(\mu,\Sigma)$, the unscented transform produces sigma
points $\{\chi^{(k)}\}$ and weights $\{w_m^{(k)},w_c^{(k)}\}$.

\paragraph{Initialization.}
For each regime $j=1,\ldots,S$, specify a Gaussian prior
\begin{equation}
  X_0\mid s_0=j \sim \mathcal N\!\big(\widehat X_{0|0}(j),\,\widehat P_{0|0}(j)\big),
\end{equation}
together with initial regime probabilities $\pi_0(j)$ (e.g., the stationary
distribution of the regime chain).

\paragraph{Filtering recursion.}
For $t=1,\ldots,T$, given $\{\widehat X_{t-\Delta t|t-\Delta t}(j),
\widehat P_{t-\Delta t|t-\Delta t}(j),\pi_{t-\Delta t}(j)\}_{j=1}^{S}$, perform:

\paragraph{(1) Gray-style collapse of the continuous state.}
We approximate the regime mixture posterior by a single Gaussian matching the
first two moments,
\begin{equation}
\label{eq:gray_collapse_app}
\begin{aligned}
  \widehat X_{t-\Delta t|t-\Delta t}
  &=\sum_{j=1}^{S}\pi_{t-\Delta t}(j)\,\widehat X_{t-\Delta t|t-\Delta t}(j),\\[0.3em]
  \widehat P_{t-\Delta t|t-\Delta t}
  &=\sum_{j=1}^{S}\pi_{t-\Delta t}(j)\Big[
      \widehat P_{t-\Delta t|t-\Delta t}(j)
      +\big(\widehat X_{t-\Delta t|t-\Delta t}(j)-\widehat X_{t-\Delta t|t-\Delta t}\big)
       \big(\widehat X_{t-\Delta t|t-\Delta t}(j)-\widehat X_{t-\Delta t|t-\Delta t}\big)^{\!\top}
     \Big].
\end{aligned}
\end{equation}
Generate sigma points $\{\chi^{(k)}_{t-\Delta t|t-\Delta t}\}$ and weights
$\{w_m^{(k)},w_c^{(k)}\}$ from
$\mathcal N(\widehat X_{t-\Delta t|t-\Delta t},\widehat P_{t-\Delta t|t-\Delta t})$.

\paragraph{(2) Regime-specific state prediction.}
For each regime $j$, propagate sigma points through the transition map $F_j(\cdot)$:
\begin{equation}
\begin{aligned}
  \chi^{(k)}_{t|t-\Delta t}(j)
  &=F_j\!\big(\chi^{(k)}_{t-\Delta t|t-\Delta t}\big),\\[0.3em]
  \widehat X_{t|t-\Delta t}(j)
  &=\sum_k w_m^{(k)}\,\chi^{(k)}_{t|t-\Delta t}(j),\\[0.3em]
  \widehat P_{t|t-\Delta t}(j)
  &=\sum_k w_c^{(k)}\big(\chi^{(k)}_{t|t-\Delta t}(j)-\widehat X_{t|t-\Delta t}(j)\big)
                    \big(\chi^{(k)}_{t|t-\Delta t}(j)-\widehat X_{t|t-\Delta t}(j)\big)^{\!\top}
    +Q_j.
\end{aligned}
\end{equation}

\paragraph{(3) Regime-specific measurement prediction and likelihood.}
Apply the measurement map $h_j(\cdot)$ to the predicted sigma points:
\begin{equation}
\begin{aligned}
  \upsilon^{(k)}_{t|t-\Delta t}(j)
  &=h_j\!\big(\chi^{(k)}_{t|t-\Delta t}(j)\big),\\[0.3em]
  \widehat{\mathbf y}_{t|t-\Delta t}(j)
  &=\sum_k w_m^{(k)}\,\upsilon^{(k)}_{t|t-\Delta t}(j),\\[0.3em]
  S_t(j)
  &=\sum_k w_c^{(k)}\big(\upsilon^{(k)}_{t|t-\Delta t}(j)-\widehat{\mathbf y}_{t|t-\Delta t}(j)\big)
                    \big(\upsilon^{(k)}_{t|t-\Delta t}(j)-\widehat{\mathbf y}_{t|t-\Delta t}(j)\big)^{\!\top}
    +R_j.
\end{aligned}
\end{equation}
The state--observation cross-covariance and the Kalman gain are
\begin{equation}
  P^{xy}_t(j)=\sum_k w_c^{(k)}\big(\chi^{(k)}_{t|t-\Delta t}(j)-\widehat X_{t|t-\Delta t}(j)\big)
                           \big(\upsilon^{(k)}_{t|t-\Delta t}(j)-\widehat{\mathbf y}_{t|t-\Delta t}(j)\big)^{\!\top},
  \qquad
  K_t(j)=P^{xy}_t(j)\,S_t(j)^{-1}.
\end{equation}
We approximate the regime-conditional predictive density by a Gaussian,
\begin{equation}
  f_t(j)\equiv p(\mathbf y_t\mid\mathcal I_{t-\Delta t},s_t=j)
  \approx \phi\!\big(\mathbf y_t;\widehat{\mathbf y}_{t|t-\Delta t}(j),S_t(j)\big),
\end{equation}
where $\phi(\cdot;\mu,\Sigma)$ denotes the multivariate normal density.

\paragraph{(4) Regime-specific state update.}
Update the regime-conditional posterior moments using $K_t(j)$:
\begin{equation}
\begin{aligned}
  \widehat X_{t|t}(j)
  &=\widehat X_{t|t-\Delta t}(j)+K_t(j)\big(\mathbf y_t-\widehat{\mathbf y}_{t|t-\Delta t}(j)\big),\\[0.3em]
  \widehat P_{t|t}(j)
  &=\widehat P_{t|t-\Delta t}(j)-K_t(j)\,S_t(j)\,K_t(j)^{\!\top}.
\end{aligned}
\end{equation}

\paragraph{(5) Regime probability update and likelihood contribution.}
The one-step-ahead regime probabilities are
\begin{equation}
  \pi_{t|t-\Delta t}(j)=\sum_{i=1}^{S}\pi_{t-\Delta t}(i)\,P_{ij}(\Delta t).
\end{equation}
The predictive density at time $t$ is the finite mixture
\begin{equation}
  p(\mathbf y_t\mid\mathcal I_{t-\Delta t})
  \approx \sum_{j=1}^{S}\pi_{t|t-\Delta t}(j)\,f_t(j),
\end{equation}
so that the log-likelihood contribution is
\begin{equation}
  \ell_t(\Theta)=\log p(\mathbf y_t\mid\mathcal I_{t-\Delta t}).
\end{equation}
Finally, Bayes' rule yields the filtered regime probabilities,
\begin{equation}
  \pi_t(j)=
  \frac{\pi_{t|t-\Delta t}(j)\,f_t(j)}
       {\sum_{h=1}^{S}\pi_{t|t-\Delta t}(h)\,f_t(h)},
  \qquad j=1,\ldots,S.
\end{equation}
Iterating over $t=1,\ldots,T$ gives the approximate full-sample log-likelihood
\begin{equation}
  \mathcal L(\Theta)=\sum_{t=1}^{T}\ell_t(\Theta),
\end{equation}
which is maximized numerically to obtain the parameter estimates $\widehat\Theta$.

\section{Supplementary Empirical Evidence}\label{app:empirical_supp}

\subsection*{Gaussian-HMM regime diagnostics from weekly yields}\label{app:hmm_diagnostics}

This subsection reports reduced-form regime diagnostics based on Gaussian hidden Markov models (HMMs)
estimated on weekly yield curves. These HMM results are \emph{not} used for pricing or structural estimation;
they serve as an auxiliary check that (i) yield dynamics exhibit persistent, discrete shifts in level and
within-regime variability, and (ii) a two-state partition provides a parsimonious summary consistent with
the regime-switching design adopted in the main model.
Throughout, we relabel the two HMM states by economic content: \emph{low-rate} state \(L\) and
\emph{high-rate} state \(H\), where \(H\) is the state with the larger regime-specific mean yield.

Figure~\ref{fig:app_hmm_yields} overlays yields across maturities for the benchmark rate block
(CGB and CDB) and investment-grade corporate buckets. Shaded bands indicate the inferred
HMM state sequence \(\widehat{s}^{\,\mathrm{HMM}}_{t}\in\{L,H\}\) (based on smoothed probabilities with a 0.5 cutoff).
The figure highlights that regime-like episodes are common across the sovereign and high-grade credit
segments, supporting the view that shifts in the discount-rate environment are a first-order driver of
comovement in yields.

\begin{figure}[!t]
  \centering
  \includegraphics[width=0.95\linewidth]{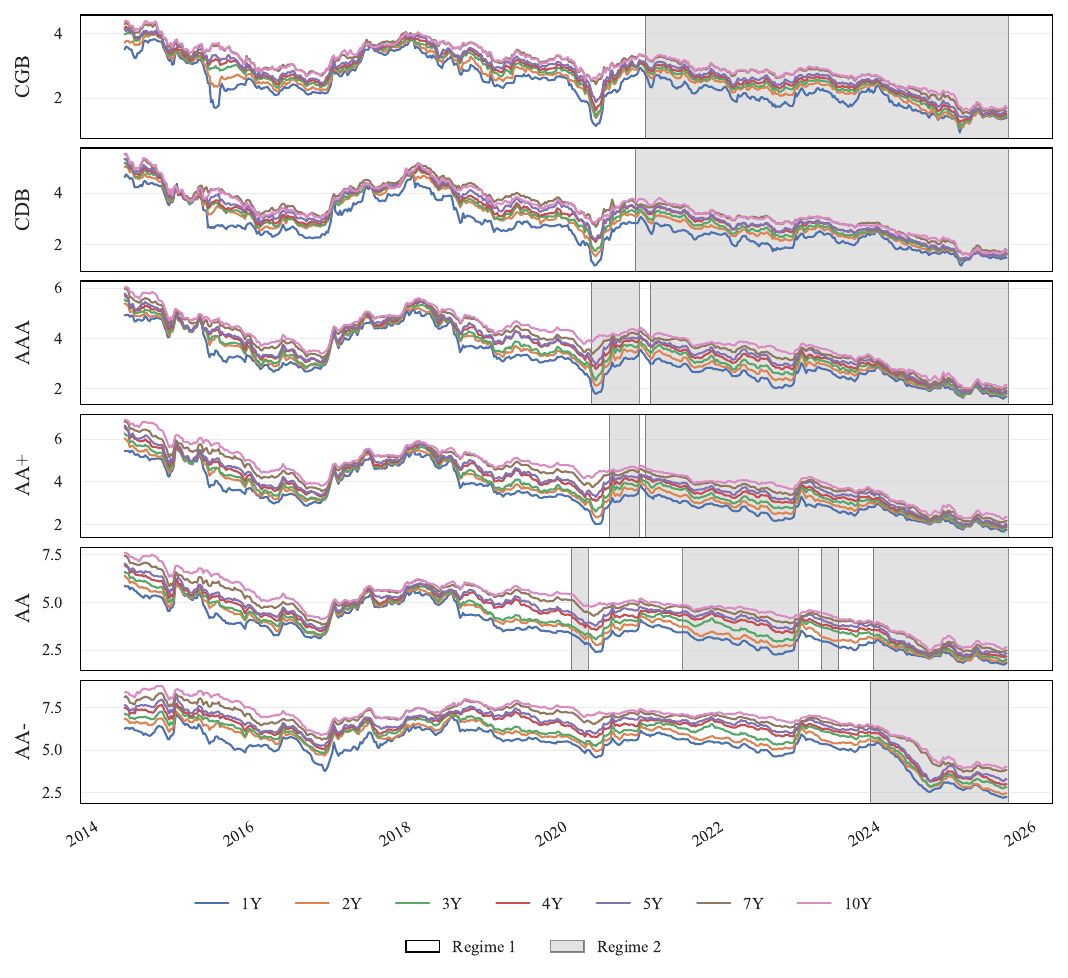}
  \caption{Weekly zero-coupon yields for CGB, CDB, and investment-grade corporate buckets (AAA, AA\(+\), AA) across maturities.
  Shaded areas mark the HMM-inferred rate environment \(\widehat{s}^{\,\mathrm{HMM}}_{t}\in\{L,H\}\), where \(H\) denotes the \emph{high-rate} state (higher regime-specific mean yield) and \(L\) the \emph{low-rate} state.}
  \label{fig:app_hmm_yields}
\end{figure}


Table~\ref{tab:app_hmm_moments_yields} summarizes regime-specific unconditional moments for \(K=2\) and \(K=3\)
Gaussian HMMs. For \(K=2\), the two regimes typically differ materially in mean yield levels and in a scalar
dispersion proxy, consistent with a persistent low-rate versus high-rate environment.
For \(K=3\), an additional state often appears to split one of the \(K=2\) regimes into two nearby components
with similar mean levels but larger dispersion. This pattern is indicative of over-partitioning rather than a
stable third regime, reinforcing the use of a two-state specification in the structural RS--GCIR model.

\begin{table*}[!t]
  \centering
  \caption{Regime-specific unconditional moments from Gaussian HMM diagnostics (weekly yields).
  For \(K=2\), states are relabeled as \(L\) (low-rate) and \(H\) (high-rate) by sorting regime-specific mean yield levels.
  For \(K=3\), states are reported by index.}
  \label{tab:app_hmm_moments_yields}
  \footnotesize
  \renewcommand{\arraystretch}{1.15}
  \setlength{\tabcolsep}{4pt}
  \begin{threeparttable}
    \begin{tabular*}{\linewidth}{@{\extracolsep{\fill}} lcccrcc}
      \toprule
      Segment & \(K\) & State & Count & Weight & Mean level\tnote{a} & Dispersion proxy\tnote{b} \\
      \midrule
      \multirow{5}{*}{CGB}
        & 2 & \(L\) & 343 & 0.624 & 2.792 & 2.351 \\
        & 2 & \(H\) & 207 & 0.376 & 3.576 & 1.169 \\
        & 3 & 1     & 133 & 0.242 & 3.028 & 5.162 \\
        & 3 & 2     & 132 & 0.240 & 3.011 & 5.131 \\
        & 3 & 3     & 285 & 0.518 & 3.149 & 0.752 \\
      \midrule
      \multirow{5}{*}{CDB}
        & 2 & \(H\) & 338 & 0.615 & 3.952 & 2.518 \\
        & 2 & \(L\) & 212 & 0.385 & 2.648 & 1.971 \\
        & 3 & 1     & 123 & 0.224 & 3.438 & 9.900 \\
        & 3 & 2     & 124 & 0.225 & 3.471 & 9.917 \\
        & 3 & 3     & 303 & 0.551 & 3.445 & 1.135 \\
      \midrule
      \multirow{5}{*}{AAA}
        & 2 & \(L\) & 367 & 0.667 & 3.106 & 2.404 \\
        & 2 & \(H\) & 183 & 0.333 & 4.602 & 1.704 \\
        & 3 & 1     & 112 & 0.204 & 3.801 & 11.642 \\
        & 3 & 2     & 112 & 0.204 & 3.764 & 11.572 \\
        & 3 & 3     & 326 & 0.593 & 3.481 & 1.258 \\
      \midrule
      \multirow{5}{*}{AA\(+\)}
        & 2 & \(H\) & 306 & 0.556 & 4.558 & 4.216 \\
        & 2 & \(L\) & 244 & 0.444 & 3.116 & 2.802 \\
        & 3 & 1     &  91 & 0.165 & 3.819 & 13.455 \\
        & 3 & 2     &  93 & 0.169 & 3.854 & 13.682 \\
        & 3 & 3     & 366 & 0.665 & 3.959 & 3.857 \\
      \midrule
      \multirow{5}{*}{AA}
        & 2 & \(H\) & 324 & 0.589 & 4.945 & 4.223 \\
        & 2 & \(L\) & 226 & 0.411 & 3.369 & 3.576 \\
        & 3 & 1     & 106 & 0.193 & 4.154 & 13.293 \\
        & 3 & 2     & 106 & 0.193 & 4.140 & 13.235 \\
        & 3 & 3     & 338 & 0.615 & 4.392 & 4.651 \\
      \midrule
      \multirow{5}{*}{AA\(-\)}
        & 2 & \(H\) & 484 & 0.880 & 6.422 & 2.087 \\
        & 2 & \(L\) &  66 & 0.120 & 3.612 & 1.804 \\
        & 3 & 1     &  64 & 0.116 & 4.635 & 9.252 \\
        & 3 & 2     &  64 & 0.116 & 4.605 & 9.106 \\
        & 3 & 3     & 422 & 0.767 & 6.529 & 1.594 \\
      \bottomrule
    \end{tabular*}
    \begin{tablenotes}[flushleft]
      \footnotesize
      \item[a] Regime-specific unconditional mean of the yield series (percent per annum) implied by the state-dependent Gaussian emission distribution.
      \item[b] Scalar proxy for within-regime variability (e.g., an aggregate variance/trace measure used to summarize multi-maturity dispersion when applicable).
      \item[] \emph{Relabeling rule for \(K=2\):} within each segment, the state with the lower mean level is labeled \(L\), and the one with the higher mean level is labeled \(H\).
    \end{tablenotes}
  \end{threeparttable}
\end{table*}

\begin{table}[!t]
    \centering
    \caption{Implied average regime duration from the two-state Gaussian HMM.}
    \label{tab:gauss_hmm_duration_k2}
    \footnotesize
    \renewcommand{\arraystretch}{1.2}
    \setlength{\tabcolsep}{6pt}
    \begin{threeparttable}
    \begin{tabular*}{\linewidth}{@{\extracolsep{\fill}} lcccc}
        \toprule
        Segment & $p_{11}$ & $p_{22}$ & $E[D_1]$ (years) & $E[D_2]$ (years) \\
        \midrule
        CGB        & 0.989 & 0.976 & 1.77 & 0.82 \\
        CDB spread & 0.990 & 0.984 & 1.86 & 1.23 \\
        AAA spread & 0.982 & 0.987 & 1.10 & 1.42 \\
        AA+ spread & 0.992 & 0.990 & 2.38 & 1.93 \\
        AA spread  & 0.994 & 0.987 & 2.99 & 1.52 \\
        AA- spread & 0.993 & 0.992 & 2.81 & 2.45 \\
        \bottomrule
    \end{tabular*}
    \begin{tablenotes}[flushleft]
        \footnotesize
        \item Notes: For each segment, we estimate a two-state Gaussian HMM on weekly data and report the one-step self-transition probabilities $p_{ii}$. 
        The implied expected duration of regime $i$ is $E[D_i]=\Delta/(1-p_{ii})$, where $\Delta$ is one week, reported in years.
    \end{tablenotes}
    \end{threeparttable}
\end{table}

To assess whether the inferred regimes reflect persistent episodes rather than noisy state flips, Table~\ref{tab:gauss_hmm_duration_k2} reports the implied average regime duration from the estimated two-state Gaussian HMM. Across all segments, the self-transition probabilities are close to one, implying expected regime durations on the order of roughly one to three years (weekly frequency). This persistence supports the interpretation that the HMM captures slow-moving shifts in the level/volatility environment, providing diagnostic evidence for regime behavior in the data, while the structural pricing model is developed separately.

\subsection*{Filter-based benchmark: EKF versus RS--KF}\label{app:kf_comparison}

The main estimation adopts a nonlinear regime-switching filter to accommodate the curvature induced by
(i) state-dependent affine pricing kernels and (ii) the price-to-yield transformation.
As a benchmark, this subsection reports pricing errors from (a) a single-regime Kalman filter (EKF) and
(b) a two-regime regime-switching Kalman filter (RS--EKF), both evaluated on the same maturity grid up to 10 years.
The comparison isolates the incremental contribution of allowing for regime shifts within a linear-Gaussian filtering framework.

For maturity \(\tau\in\{1,2,3,4,5,7,10\}\) years, the pricing error is
\(\epsilon_{t}(\tau)=y_{t}(\tau)-\widehat{y}_{t}(\tau)\).
Reported mean and standard deviation are in basis points (bp), i.e., \(10{,}000\times \epsilon_{t}(\tau)\).
RRMSE is unitless, and the ``Average'' column is the simple average across maturities.

\begin{table*}[!htbp]
\centering
\caption{Pricing error statistics under the RS--EKF benchmark (up to 10Y).}
\label{tab:app_rskf_pricing_error_stats}
\footnotesize
\setlength{\tabcolsep}{4pt}
\renewcommand{\arraystretch}{1.08}
\begin{threeparttable}
\begin{tabular*}{\linewidth}{@{\extracolsep{\fill}} lrrrrrrrr}
\toprule
& \(\epsilon_{1\mathrm{Y}}\) & \(\epsilon_{2\mathrm{Y}}\) & \(\epsilon_{3\mathrm{Y}}\) & \(\epsilon_{4\mathrm{Y}}\) & \(\epsilon_{5\mathrm{Y}}\) & \(\epsilon_{7\mathrm{Y}}\) & \(\epsilon_{10\mathrm{Y}}\) & Average \\
\midrule
\multicolumn{9}{l}{\textit{CGB}} \\
Mean      & 11.20 & -0.52 & 1.67 & 0.27 & 0.86 & -6.89 & 12.66 & 2.75 \\
St.~dev.  & 20.00 & 15.65 & 16.23 & 16.36 & 15.31 & 15.22 & 14.59 & 16.19 \\
RRMSE     & 0.095 & 0.056 & 0.052 & 0.046 & 0.043 & 0.043 & 0.054 & 0.056 \\
\addlinespace
\multicolumn{9}{l}{\textit{CDB}} \\
Mean      & 3.06 & -3.08 & 1.78 & 1.63 & 2.75 & -4.88 & 14.14 & 2.20 \\
St.~dev.  & 13.92 & 10.99 & 11.87 & 12.77 & 12.31 & 12.38 & 12.54 & 12.40 \\
RRMSE     & 0.058 & 0.041 & 0.041 & 0.039 & 0.040 & 0.036 & 0.054 & 0.044 \\
\addlinespace
\multicolumn{9}{l}{\textit{AAA}} \\
Mean      & 0.51 & 3.90 & 6.58 & 1.09 & 0.63 & -6.16 & -11.14 & -0.66 \\
St.~dev.  & 12.71 & 14.41 & 15.38 & 14.33 & 14.46 & 11.78 & 13.17 & 13.75 \\
RRMSE     & 0.045 & 0.055 & 0.058 & 0.051 & 0.050 & 0.045 & 0.047 & 0.050 \\
\addlinespace
\multicolumn{9}{l}{\textit{AA\(+\)}} \\
Mean      & -2.67 & 1.31 & 3.96 & -2.11 & -2.45 & -10.75 & -5.30 & -2.14 \\
St.~dev.  & 18.83 & 21.58 & 21.20 & 18.97 & 18.24 & 19.62 & 26.45 & 20.70 \\
RRMSE     & 0.058 & 0.067 & 0.067 & 0.055 & 0.050 & 0.043 & 0.052 & 0.053 \\
\addlinespace
\multicolumn{9}{l}{\textit{AA}} \\
Mean      & -3.40 & 0.70 & 2.79 & -5.74 & -3.52 & -6.64 & 8.52 & -1.04 \\
St.~dev.  & 25.73 & 26.61 & 23.45 & 18.78 & 19.08 & 24.42 & 35.74 & 24.83 \\
RRMSE     & 0.072 & 0.070 & 0.060 & 0.046 & 0.044 & 0.045 & 0.076 & 0.059 \\
\addlinespace
\multicolumn{9}{l}{\textit{AA\(-\)}} \\
Mean      & -1.74 & -2.07 & 0.86 & -7.11 & -5.24 & -12.78 & -1.06 & -4.16 \\
St.~dev.  & 36.53 & 24.82 & 24.74 & 24.20 & 23.77 & 19.49 & 20.53 & 24.87 \\
RRMSE     & 0.080 & 0.049 & 0.040 & 0.039 & 0.037 & 0.036 & 0.030 & 0.044 \\
\bottomrule
\end{tabular*}
\vspace{0.25em}
\raggedright\footnotesize
Notes: ``Average'' is the simple average across maturities \(\{1,2,3,4,5,7,10\}\) years. Mean and standard deviation are reported in bp; RRMSE is unitless.
\end{threeparttable}
\end{table*}

\begin{table*}[!htbp]
\centering
\caption{Pricing error statistics under the single-regime EKF benchmark (up to 10Y).}
\label{tab:app_kf_pricing_error_stats}
\footnotesize
\setlength{\tabcolsep}{4pt}
\renewcommand{\arraystretch}{1.08}
\begin{threeparttable}
\begin{tabular*}{\linewidth}{@{\extracolsep{\fill}} lrrrrrrrr}
\toprule
& \(\epsilon_{1\mathrm{Y}}\) & \(\epsilon_{2\mathrm{Y}}\) & \(\epsilon_{3\mathrm{Y}}\) & \(\epsilon_{4\mathrm{Y}}\) & \(\epsilon_{5\mathrm{Y}}\) & \(\epsilon_{7\mathrm{Y}}\) & \(\epsilon_{10\mathrm{Y}}\) & Average \\
\midrule
\multicolumn{9}{l}{\textit{CGB}} \\
Mean     & 9.16 & -0.36 & 2.88 & 2.04 & 2.95 & -4.37 & 15.55 & 3.98 \\
St.~dev. & 17.48 & 13.77 & 14.64 & 15.21 & 13.91 & 13.94 & 14.41 & 14.77 \\
RRMSE    & 0.091 & 0.053 & 0.053 & 0.048 & 0.047 & 0.042 & 0.069 & 0.058 \\
\addlinespace
\multicolumn{9}{l}{\textit{CDB}} \\
Mean     & 5.12 & -3.06 & 0.90 & 0.46 & 1.64 & -5.39 & 14.78 & 2.06 \\
St.~dev. & 15.97 & 12.89 & 14.56 & 15.16 & 13.67 & 13.70 & 14.29 & 14.32 \\
RRMSE    & 0.067 & 0.048 & 0.052 & 0.048 & 0.041 & 0.044 & 0.056 & 0.051 \\
\addlinespace
\multicolumn{9}{l}{\textit{AAA}} \\
Mean     & -11.79 & -3.15 & 3.37 & 1.12 & 3.77 & 3.45 & 8.89 & 0.81 \\
St.~dev. & 19.03 & 18.71 & 19.07 & 17.07 & 17.12 & 13.37 & 13.40 & 16.82 \\
RRMSE    & 0.075 & 0.073 & 0.072 & 0.064 & 0.060 & 0.047 & 0.045 & 0.062 \\
\addlinespace
\multicolumn{9}{l}{\textit{AA\(+\)}} \\
Mean     & -13.57 & -4.62 & 1.45 & -1.99 & -0.02 & -2.29 & 7.40 & -1.95 \\
St.~dev. & 27.40 & 25.91 & 24.05 & 20.23 & 18.69 & 18.04 & 23.42 & 22.53 \\
RRMSE    & 0.084 & 0.079 & 0.072 & 0.059 & 0.051 & 0.043 & 0.058 & 0.064 \\
\addlinespace
\multicolumn{9}{l}{\textit{AA}} \\
Mean     & -13.66 & -4.88 & 0.46 & -6.12 & -3.43 & -2.57 & 14.41 & -2.26 \\
St.~dev. & 34.54 & 30.41 & 25.45 & 19.98 & 19.44 & 23.00 & 31.26 & 26.30 \\
RRMSE    & 0.096 & 0.082 & 0.067 & 0.050 & 0.047 & 0.046 & 0.075 & 0.066 \\
\addlinespace
\multicolumn{9}{l}{\textit{AA\(-\)}} \\
Mean     & 0.33 & 1.58 & 3.94 & -5.65 & -5.78 & -17.36 & -10.87 & -4.83 \\
St.~dev. & 32.80 & 23.12 & 22.92 & 23.59 & 24.11 & 20.53 & 23.87 & 24.42 \\
RRMSE    & 0.077 & 0.051 & 0.041 & 0.039 & 0.039 & 0.042 & 0.038 & 0.047 \\
\bottomrule
\end{tabular*}
\vspace{0.25em}
\raggedright\footnotesize
Notes: ``Average'' is the simple average across maturities \(\{1,2,3,4,5,7,10\}\) years. Mean and standard deviation are in bp; RRMSE is unitless.
\end{threeparttable}
\end{table*}

Tables~\ref{tab:app_rskf_pricing_error_stats}--\ref{tab:app_kf_pricing_error_stats} show that allowing for regime switching
within a EKF benchmark can improve fit in several segments/maturities, but the remaining errors reflect the fact that
the measurement equation for yields inherits curvature from defaultable-bond valuation and the price-to-yield map.
This motivates the nonlinear filtering strategy adopted in the main estimation and the emphasis on regime-conditional
pricing combined at the price level rather than averaging state variables prior to valuation.

\end{document}